\documentclass[11pt]{article}
\usepackage[letterpaper, left=1in, right=1in, top=1in,bottom=1in]{geometry}
\usepackage{amsfonts}       
\usepackage{nicefrac}       
\usepackage{bbm}
\usepackage[ruled,noend]{algorithm2e}
\usepackage{bm}
\usepackage{tikz}
\usetikzlibrary{positioning,chains,fit,shapes,calc,arrows}
\usetikzlibrary{arrows.meta}
\usepackage{graphicx}
\usepackage{pgf}
\usepackage{mathtools}
\usepackage{mathrsfs}
\usepackage{amsmath,amsthm,amssymb}
\usepackage{comment}
\usepackage{xfrac}
\usepackage{accents}
\usepackage{cancel}
\usepackage[normalem]{ulem}
\usepackage{subcaption}

\usepackage{natbib}

\usepackage[pagebackref=true]{hyperref} 
\PassOptionsToPackage{pagebackref=true}{hyperref}

\usepackage{algpseudocode}
\usepackage{tikz}

\DeclarePairedDelimiter{\floor}{\lfloor}{\rfloor}
\DeclareMathOperator*{\argmax}{arg\,max}

\DeclareMathAlphabet{\pazocal}{OMS}{zplm}{m}{n}

\usepackage{auxiliary}

\usepackage{float}

\usepackage{url}      

\usepackage[pagebackref=true]{hyperref} \PassOptionsToPackage{pagebackref=true}{hyperref}   
\usepackage[capitalize]{cleveref}

\usepackage{graphicx}
\usepackage{nicefrac}       

\usepackage{tikz}
\usetikzlibrary{positioning}
\usepackage{hyperref}
\definecolor{darkred}{RGB}{150,0,0}
\definecolor{darkgreen}{RGB}{0,150,0}
\definecolor{darkblue}{RGB}{0,0,200}
\hypersetup{colorlinks=true, linkcolor=darkred, citecolor=darkgreen, urlcolor=darkblue}
\usepackage{amssymb}
\usepackage{soul}

\usepackage{caption} 

\usepackage{mathtools}

\usepackage{auxiliary}

\usepackage{cancel} 

\usepackage{color}              

\title{
Social Learning with Limited Attention:\\ Negative Reviews Persist under Newest First
}
\author{
Jackie Baek\thanks{Stern School of Business, New York University, \texttt{baek@stern.nyu.edu}} \and 
Atanas Dinev\thanks{Massachusetts Institute of Technology, \texttt{adinev@mit.edu}} \and Thodoris Lykouris\thanks{Massachusetts Institute of Technology, \texttt{lykouris@mit.edu}}
}

\date{First version: June 2024\\Current Version: May 2026\footnote{An extended abstract appeared at the ACM Conference on Economics and Computation (EC) 2024.}}

\begin{document}

\maketitle

\thispagestyle{empty}

\begin{abstract}
We study a model of \emph{social learning from reviews} where customers are computationally limited and make purchases based on reading only the first few reviews displayed by the platform. Under this limited attention, we establish that the \emph{review ordering} policy can have a significant  impact. In particular, the popular \emph{Newest First} ordering induces a negative review to persist as the most recent review longer than a positive review. This phenomenon, which we term the \emph{Cost of Newest First}, can make the long-term revenue unboundedly lower than a counterpart where reviews are exogenously drawn for each customer. 

We show that the impact of the \emph{Cost of Newest First} can be mitigated under dynamic pricing, which allows the price to depend on the set of displayed reviews. Under the optimal dynamic pricing policy, the revenue loss is at most a factor of 2. On the way, we identify a structural property for this optimal dynamic pricing: the prices should ensure that the probability of a purchase is always the same, regardless of the state of reviews. 
We also consider a setting where product quality evolves over time according to a Markov chain; we find that Newest First better tracks current quality but still leads to lower revenue, highlighting a trade-off between customer belief accuracy and revenue. Finally, numerical simulations confirm the robustness of the \emph{Cost of Newest First} across several modeling variants.
\end{abstract}

\newpage
\setcounter{page}{1}

\section{Introduction}
\label{sec:into}
The use of product reviews to inform customer purchase decisions has become ubiquitous in a variety of online platforms, ranging from electronic commerce to accommodation and recommendation platforms. 
While the online nature of such platforms may hinder the ability of customers to confidently evaluate the product compared to an in-person experience, reviews written by previous customers can shed light on the product's quality. 
It is well established that product reviews play a significant role on customer purchase decisions \cite{chevalier2006effect,zhu2010impact,luca2016reviews}. 

The process in which reviews impact product purchases can be seen as a problem of \emph{social learning}, which generically studies how agents update their beliefs for an unknown quantity of interest (e.g., product quality) based on observing actions of past agents (e.g., reading reviews written by past customers). The typical assumption in the literature of social learning with reviews is that, when deciding whether to purchase a product, customers consider either \textit{all} reviews provided by previous customers \cite{cims17,imsz19,ghkv23} or a summary statistic such as their average rating \cite{bs18,chen2021reviews,acemoglu2022learning}. The motivation for the latter assumption is that customers have limited time and computational power and thus rely on a summary statistic, often provided by the platform (see Section~\ref{sec:relatedwork} for a further discussion on these lines of work).

However, in practice, a common scenario may be somewhere ``in between'' the above two assumptions: customers read a small number of reviews in detail. Existing works have found that the textual content of a review contains important information that goes beyond its numeric score and such information can heavily influence purchase decisions \cite{ghose2010estimating,archak2011deriving,ludwig2013more,liu2019large, vana2021effect, lei2022swayed}. Therefore, customers look beyond the average review rating and read a small number of reviews in detail. In particular, \cite{kavanagh2021impact} finds that 76\% of customers read between 1 and 9 reviews before making a purchase. This motivates the main questions of our work:
\begin{center}\emph{
When customers read a limited number of reviews, how does this impact social learning?\\
Are there operational decisions that should be reconsidered due to this limited attention?}
 \end{center}
 
We study a model that builds on existing models of social learning, where we incorporate the behavior that customers read a limited number of reviews. This simple change has two key implications.
First, we show that the \textit{ordering} of those reviews plays a critical role in impacting social learning and revenue. 
This is in contrast to existing models, where customers read all reviews or rely solely on a summary statistic, in which the review ordering has no effect. 
Our work thus introduces a new operational lever and enables us to compare review ordering policies.
Second, unlike models in the literature where customer beliefs eventually converge, beliefs in our model do not stabilize but rather \textit{fluctuate indefinitely} because customers read only a small number of reviews. This provides a fundamentally different perspective on social learning when customers read a limited number of reviews.

Concretely, we study a model for a single product (formalized in Section~\ref{sec:model}), where a platform makes decisions regarding how reviews are \emph{ordered} and how the product is \emph{priced}.
Customers arrive sequentially and each customer takes only the first $c$ reviews into account to inform their purchase decision, where $c$ is a small constant. We assume that the $t$'th customer's valuation can be decomposed as the sum of a) an idiosyncratic valuation $\Theta_t$ that is known to them and b) a product quality $X_t$ that has a fixed mean $\mu$; the latter quantity is unknown to the customer and can only be inferred via the reviews. The randomness of $X_t$ reflects variability in the service quality, product defects, or exogenous factors affecting the customer experience \cite{cims17}.We assume that when a customer reads a review written by customer $s$, they observe $X_s$ (see Section~\ref{subsec:discussion_modeling_assumptions} for a discussion of this assumption).
Each customer uses $c$ reviews to update their belief about $\mu$, and makes a purchase if their estimate of their valuation is higher than the price. In the event of a purchase, they leave a review that future customers can read. 

\subsection{Our contribution}

A popular review ordering policy is to display reviews in 
reverse chronological order (newest to oldest);
we refer to this policy as $\newest$. This is the default option in platforms such as Airbnb, Tripadvisor and Macy's\footnote{This statement is based on access on May 14, 2025. Many other platforms such as Amazon and Yelp list \textsc{newest} as the second default and have their own ordering mechanism as the default option.} 
as it allows customers to get access to the most up-to-date reviews. In the context of our model, under the $\newest$ ordering, a customer considers the $c$ most recent reviews. The set of these $c$ reviews evolves as a \emph{stochastic process} over customer arrivals: when a new purchase happens and thus a new review is provided, this review replaces the $c$-th most recent review. 

\paragraph{Cost of Newest First.} 
By analyzing the steady state of the aforementioned stochastic process, we observe that the $\newest$ ordering policy induces an undesirable behavior where negative reviews are read more than positive reviews, leading to a significant loss in overall revenue (Section~\ref{sec:ordering}). 

To illustrate this phenomenon, consider a simple setting where customers only read the first review ($c=1$) and the probability of a purchase is higher when the review is positive.
When the $t$'th customer arrives, if the first review is positive, this customer is more likely to buy the product and subsequently leave a review; the new review from the $t$'th customer then becomes the ``first review'' for the $(t+1)$-st customer. 
On the other hand, if the first review is negative for the $t$'th customer, they are less likely to buy the product, and hence the same negative review remains as the ``first review'' for the $(t+1)$-st customer. 
Therefore, negative reviews persist under the newest first ordering: a review will stay longer as the first review if it is negative compared to if it is positive. 

We show that this arises due to the \emph{endogeneity} of the stochastic process that $\newest$ induces and results in a loss of long-term revenue. 
To formalize this notion, we compare $\newest$ to an \emph{exogenous process} where each arriving customer sees an independently drawn random set of reviews; we refer to this ordering policy as $\random$. We establish that the long-term revenue under $\newest$ is strictly smaller than that of $\random$ under any non-degenerate instance (Theorem~\ref{theorem:negative_bias}) and that the revenue under $\newest$ can be arbitrarily smaller, in a multiplicative sense, compared to $\random$ (Theorem~\ref{thm:conf_arbitrarily_bad}).
We refer to this phenomenon as the \textsc{Cost of Newest First} (\textsc{CoNF}). 

\paragraph{Dynamic pricing mitigates the impact of \textsc{CoNF}.} Seeking to mitigate this phenomenon, we consider the impact of optimizing the product's price (Section \ref{sec:pricing}). 
We show that even under the \emph{optimal} static price, the \textsc{CoNF} remains arbitrarily large (Theorem~\ref{thm:CoNF_opt_static_arbitrarily_bad}). However, if we allow for dynamic pricing, where the price can depend on the state of the reviews, we show that the \textsc{CoNF} is upper bounded by
a factor of $2$ when the idiosyncratic valuation distribution is non-negative and gracefully decays with its negative mass otherwise (Theorem~\ref{theorem:dynamic_pricing_CoNF_bound}). 

This improvement stems from the fact that dynamic pricing allows us to \emph{change the steady-state distribution} of the stochastic process.
Recall that, under $\newest$, 
the stochastic process spends more time on states with negative reviews than states with positive reviews. The optimal dynamic pricing policy sets prices so that the purchase probability is \emph{equal} across all states (\cref{theorem:newest_first_dynamic_optimal_policy}) --- this ensures that the steady-state distribution under $\newest$ is the same as that of $\random$. 

A broader implication of this result is that, when purchase decisions depend on the state of the first reviews, platforms that offer state-dependent prices can be arbitrarily better off than platforms that are unaware of this phenomenon and statically optimize prices (\cref{thm:dynamic_vs_static}).

\paragraph{Non-stationarity in product quality.}
One potential benefit of Newest First is that newer reviews are more representative of current product quality, when the quality changes over time. 
To formalize this, we define a model with non-stationarity where the product quality changes over time according to a Markov chain (\cref{sec:dynamic_quality}). 
We define a \textit{belief error} metric that represents the difference between the customer's posterior belief and the current product quality.
Indeed, the belief error is lower under Newest First compared to Random, where the gap is higher when the level of non-stationarity is low (i.e., quality switches are less frequent). 
However, we show that the CoNF still exists under non-stationarity; revenue is lower under Newest First compared to Random.
This introduces a trade-off between revenue and belief error. 
Revenue is always higher under Random while the belief error is better (lower) under Newest First, but the benefit of Newest First for belief error diminishes when quality switches are more frequent. 
Therefore, when the level of non-stationarity is high, the benefit of using Newest First diminishes.

\paragraph{Numerical evaluations.} 
We conduct five sets of numerical evaluations to support and extend our theoretical findings (\cref{sec:numerics}). In Section~\ref{subsec:numerics_limited_attention} we investigate the impact of the limited-attention parameter $c$ on the CoNF and we show that $c = 1$ need not be the worst case. In Section~\ref{subsec:numerics_time_varying_prior} we numerically investigate a setting where the customer's prior depends endogenously on the entire review history and show that the CoNF continues to hold under this model. In Section~\ref{subsec:numerics_nonstationarity} we numerically investigate a model where the product quality is non-stationary and increases with time, a setting which benefits Newest First the most. Despite this, we show that there exist settings where the CoNF exists. In Section~\ref{subsec:conf_with_self_selection_bias} we numerically investigate a setting in the presence of self-selection bias and show that the CoNF exists under this setting. Lastly, in Section~\ref{subsec:numerics_dynamic_vs_static_conf} we numerically investigate the impact of dynamic pricing on CoNF and on the revenues generated by Newest, as our theoretical results on the CoNF provide worst-case bounds. 

\subsection{Key Behavioral Assumption and Managerial Implications}
In this paper, we posit that customers have \emph{limited attention}, i.e., they read a limited number of the top-ranked reviews (see Section~\ref{sec:relatedwork} for a discussion on this assumption). Under this behavioral assumption, our paper sheds light to three novel managerial implications:

\begin{enumerate}
    \item \textbf{Impact of review ordering.} 
    Our paper identifies that if customers only read a limited number of the top-ranked reviews, then the \textit{ordering} of reviews becomes an important operational lever. 
    This lever was previously overlooked by the literature as past works did not model customers' limited attention. Under this assumption, we identify and characterize the \emph{Cost of Newest First} phenomenon, that presenting reviews in a Newest First order leads to the persistence of negative reviews, which negatively impacts customer beliefs and consequently reduces platform revenue.
    \item \textbf{Incorporating reviews in dynamic pricing algorithms.} 
    We find that dynamic pricing strategies mitigate the Cost of Newest First. Theoretically, the negative impact is bounded by a factor of two, and our numerical experiments show that this effect is often negligible.
    Practically speaking, platforms already widely employ sophisticated and highly dynamic pricing algorithms \cite{patterson2017highspeed}. 
    Our results imply that the \textit{current set of reviews} is an important feature that platforms should incorporate in their pricing algorithms, alongside other important factors such as competitors' prices, demand forecasts, seasonal trends, and others.
    \item \textbf{Non-stationarity: Randomizing over recent reviews.} One potential benefit of Newest First is that newer reviews are more representative 
    of the current product quality. Indeed, when product quality changes over time, we show that customer 
    beliefs are more accurate under Newest First than under Random, yet revenue remains 
    lower, revealing a trade-off between belief accuracy and revenue. This trade-off 
    shifts as the level of non-stationarity increases: the accuracy advantage of Newest 
    First diminishes as quality changes more frequently, while the CoNF persists. 
    Since Newest First and Random represent two extremes of this trade-off, we show 
    that a small amount of randomization, i.e., selecting reviews uniformly at random 
    from a larger window of the most recent reviews, can outperform both. This policy 
    captures the recency benefit of Newest First while mitigating the CoNF.
\end{enumerate}

\subsection{Related Work and Comparison of Key Modeling Assumptions} \label{sec:relatedwork}\label{subsec:discussion_modeling_assumptions}

\paragraph{Social learning and incentivized exploration.}
Classical models of social learning from \cite{b92} and \cite{bhw92} study a setting in which there is an unknown state of the world and each agent observes an independent, noisy signal about the state as well as the actions of past agents. The agent uses this information to update their beliefs and then takes an action. In this setting, undesirable ``herding'' behavior can arise: agents may converge to taking the wrong action. Conceptually closer to our work, \cite{sayedi2018pricing} shows that dynamic pricing can mitigate the aforementioned herding behavior.
Subsequent works study how social learning is affected by the agent's signal distribution \cite{smith2000pathological}, prior for the state \cite{chakraborty2022consumers}, heterogeneous preferences \cite{goeree2006social,lobel2016preferences}, as well as the structure of their observations \cite{acemoglu2022learning}.
From a different perspective, there is a stream of literature that aim to design mechanisms to \textit{help} the learning process, either by modifying the information structure  \cite{kremer2014implementing, mansour2020bayesian,bimpikis2018crowdsourcing} or by incentivizing exploration through payments  \cite{frazier2014incentivizing,kannan2017fairness}.

\paragraph{Social learning with reviews.}
Closer to our work, several papers focus on the setting where customers learn about a product's quality through reviews 
\cite{hu2017self,cims17,bs18,schmit2018human,imsz19,chen2021reviews,acemoglu2022learning,ghkv23,bondi2023alone,carnehl2024pricing}. 
This literature induces several modeling differences compared to classical social learning.
First, agents do not receive independent signals of the unknown state (the product quality).
Second, agents not only observe the binary purchase decision of previous agents, but also the reviews of previous agents who purchased the product. We highlight the key modeling assumptions of our work and how they relate to existing works.

\paragraph{No self-selection bias.}
In prior works of social learning with reviews, the main difficulty stems from the \textit{self-selection bias}, the idea that only customers who value the product highly will buy the product and hence these customers leave reviews with higher ratings. 
In the presence of self-selection bias, 
\cite{cims17}, \cite{bs18} and \cite{imsz19} study conditions in which customer beliefs eventually successfully learn the quality of a product, where customers update their beliefs based on the entire history of reviews.
\cite{bs18,schmit2018human, chen2021reviews,acemoglu2022learning} 
consider models in which customers only incorporate summary statistics of prior reviews (e.g., average rating) into their beliefs. 
\cite{bondi2023alone} analyzes how the magnitude of the self-selection bias depends on the product's quality and polarization.
\cite{carnehl2024pricing} consider a model where the platform's pricing decision affect the review ratings and characterizes the impact of the price on the average rating. This is also empirically supported by \cite{byers2012groupon} which shows that Groupon discounts lead to lower ratings. \cite{hu2017self} study a two-stage model which quantifies both self-selection bias and under-reporting bias (reviews are provided only by customers with extreme experiences); see references within for further related work. 

In contrast, our work studies a model where self-selection bias does not arise. 
Specifically, we assume that customer $t$'s valuation can be decomposed as $V_t = \Theta_t + X_t$
, where $\Theta_t$ is customer-specific and $X_t$ has a fixed mean~$\mu$ shared across customers. The quantity $\mu$ is the unknown quantity of interest for all customers. Our model assumes that a review reveals $X_t$. In contrast, prior works assume that a review reveals $\Theta_t + X_t$ and one cannot separate the contribution from each term. This means that, in our model, the customer-specific valuation and the pricing decision affect the purchase probability  but do not affect the review itself conditioning on a purchase.
Although our assumption makes it ``easier'' for customers to learn $\mu$, we study a new phenomenon that arises due to the fact that customers only read a small number of reviews.

The practical motivation for our modeling assumption of no selection bias is the following.
On most online platforms, a review is composed by both a numeric score (e.g., 4 out of 5) and a textual description that further explains the reviewer's thoughts.
Within our model, one interpretation is that the numeric score reveals $\Theta_t + X_t$, but one can use the textual content of the review to separate $\Theta_t$ from $X_t$.
Therefore, we assume that reading the text of the review reveals $X_t$, but we also assume that each customer only reads a small number of reviews since reading the text takes time. That said, in Section~\ref{subsec:conf_with_self_selection_bias}, we numerically demonstrate that \textsc{CoNF} persists even when the customers read the full review rating $\Theta_t + X_t$ and its negative bias can outweigh the positive self-selection bias.

\cite{ghkv23} also studies a setting with no self-selection bias (without limited attention) with dynamic pricing. Their model assumes that customers are partitioned into a finite number of types and only read reviews written by customers of the same type. This overcomes self-selection bias as customers of the same type can be thought of as having the same value of $\Theta_t$ in our model.

\paragraph{Individual reviews vs. average rating.}
Most online platforms show the average rating of all reviews for each product, as well as individual reviews.  In our model, the average rating is captured through the customers' \textit{prior belief} about the product quality $\mu$, which captures any information customers have before reading individual reviews (e.g., average rating, product features, etc.).
The question we study is: conditioning on all other available information, how does the review ordering policy affect purchasing decisions? For example, if a customer starts with a high prior due to a 4.8 out of 5 average rating but still reads a few reviews, the order in which those reviews are presented can influence their final belief and purchase decision. 
Our paper therefore speaks directly to the incremental effect of review ordering, holding other information constant.

Existing empirical studies support the assumption that customers read individual reviews, and that this can have a significant impact on purchasing behavior \citep{archak2011deriving,liu2019large,vana2021effect,lei2022swayed}. \cite{archak2011deriving} emphasize that product quality is multi-dimensional and consumers have heterogeneous preferences over the dimensions; this is often lost in average ratings and suggests that individual reviews matter. Moreover, \cite{lei2022swayed} suggest that individual reviews only matter when the average rating of the product is sufficiently high to put the product into the consumer's consideration set. More closely to our work, \cite{liu2019large,vana2021effect} show that individual reviews have highest impact on purchase decisions  when they resolve uncertainty about a product. Moreover, \cite{vana2021effect} claim that the \textit{top} reviews have a substantial effect on customer behavior, which highlights the importance of review ordering. Our work is aligned with this empirical evidence as we focus on a) the effect of the $c$ top reviews on customer behavior and b) settings where individual reviews resolve uncertainty about a product by refining the customers' prior. We do not explicitly model the multi-dimensional nature of reviews and the consideration set of consumers as our insights are present even in a simpler single-dimensional model where the product is always in the customers' consideration set. 

\paragraph{Belief convergence vs.\ stochastic process.}
In the existing literature, social learning is deemed ``successful'' if the customer's estimate of product quality converges to the true quality.
This convergence can either be that their belief distribution converges to a single point \cite{imsz19,acemoglu2022learning}, or that the customer's scalar estimate of the product quality (e.g., average rating) converges to the true quality \cite{cims17,bs18}.
In our setting, customers update their beliefs based on the first $c$ reviews and  these $c$ reviews evolve across customers as a stochastic process. 
Therefore, customer beliefs do not converge but rather oscillate based on the state of those $c$ reviews, even as the number of customers goes to infinity. 

Closer to our work, \cite{park2021fateful} study a similar model (without limited attention) and show that the initial review can have an effect on the proclivity of customer purchases and the number of reviews.  This bias introduced by initial reviews is also empirically observed by the work of \cite{le2018endogenous}. Unlike our model, this effect diminishes over time as the product acquires more reviews and the initial review becomes less salient. Our results on the \textsc{Cost of Newest First} can thus be viewed as a stronger version of the result in \cite{park2021fateful}: we show that, in the presence of limited attention, the effect of negative reviews persists even in the steady-state of the system. 

\paragraph{Fully Bayesian vs.\ non-Bayesian.}
Existing papers differ in whether customers incorporate information from reviews in a fully Bayesian or non-Bayesian manner.
For example, \cite{imsz19} and \cite{acemoglu2022learning} study a fully Bayesian setting where all distributions (prior on $\mu$, distribution of $\Theta_t$) and purchasing behaviors are common knowledge and each customer forms a posterior belief on $\mu$ using the information given to them.
In contrast, \cite{cims17} and  \cite{chen2021reviews} assume that customers use a simple non-Bayesian rule when making their purchase decision. Moreover,
\cite{bs18} study both Bayesian and non-Bayesian update rules and compare them. 

In our paper, customers use a Bayesian \emph{framework} but their update rule is not \emph{fully} Bayesian.
Specifically, customers start with a Beta prior on $\mu$. 
This prior need not be correct as we assume that $\mu$ is a fixed number.
We assume reviews are binary (0 or 1), and customers read $c$ of them and update their beliefs \textit{assuming} that these reviews are independent draws from $Bernoulli(\mu)$.
However, this is not necessarily the ``correct'' Bayesian update rule for the customers, due to the \textit{endogeneity} of the stochastic process of the first $c$ reviews.
That is, under the $\newest$ ordering rule, negative reviews are more likely to persist as the most recent review --- hence even if $\mu = 1/2$, the most recent review is more likely to be negative than positive.
Therefore, a fully Bayesian customer should take this phenomenon into account when updating their beliefs. We assume that customers do \emph{not} account for this (and hence are not \textit{fully} Bayesian), and we study the impact of how this endogeneity impacts the steady state of the process. Our model also implicitly assumes that customers do not use additional information about previous customers' non-purchase decisions (which are typically non-observable) and the platform's pricing policy (which is often opaque).

Finally, our model is flexible in that it allows customers to map their belief distribution to a scalar estimate of $\mu$ in an arbitrary manner, e.g.,  the mean of the belief distribution (which is considered in most prior work) as well as a pessimistic estimate thereof (as studied in \cite{ghkv23}).

\paragraph{Other work on social learning with reviews.} 
\cite{hamilton2024fresh} studies the design of rating systems motivated by the idea that older reviews become less relevant. In a setting where the product's quality changes, they show that a moving average rating system is optimal in reflecting the true quality.
Social learning with reviews has also been studied for ranking \cite{maglaras2023product}, dealing with non-stationary environments \cite{boursier2022social}, and has been applied to green technology adoption \cite{ren2023impact}.

\paragraph{Limited attention in operations.}
Our paper also connects to a growing literature in operations that models customer decision-making under limited attention in various operational contexts. For example, \cite{akcay2019consumer} and \cite{boyaci2021pricing} model consumers who face costs in acquiring information about products, leading to non-standard choice behavior and strategic pricing by firms. \cite{feldman2019assortment} and \cite{gallego2024random} model limited consideration sets, where consumers only attend to a subset of available products when making purchasing decisions. \cite{gao2021assortment} model impatient customers who sequentially view product recommendations and may stop searching once a product surpasses their threshold utility.
Our model complements this literature by applying similar behavioral assumptions of limited attention, but in the context of customers reading reviews rather than evaluating products directly. 

\paragraph{Rational inattention.}
Conceptually, our paper is also related to the rational inattention literature \cite{sims2003implications,matvejka2015rational,fosgerau2020discrete,turlo2025discrete,chen2025multivariate}, which studies how agents make decisions under information-processing constraints. In contrast, we do not model the customer’s information acquisition problem. Instead, we take the customer's limited review-reading capacity as \textit{exogenous}: the customer observes only the top $c$ reviews and updates their belief based on them. Our focus is therefore not on how an agent optimally allocates attention, but rather on how platform decisions, such as review ordering and pricing, interact with this exogenous attention constraint.

\paragraph{Negative bias due to quality variability.} 
\cite{decroix2021service} considers a firm that repeatedly prices and sells a product to a single customer and establishes a negative bias due to service quality variability. There are two similarities with our work. First, they identify a negative bias in average beliefs driven by quality variability in experiences and the customer weighing recent experiences more highly. Second, their optimal pricing policy has the same purchase probability in each state. 

That said, our setting and results are conceptually different as we focus on the effect of reviewer ordering policies and the benchmark we consider is the best ordering policy. On a technical level, we can show that our benchmark is stronger than the one considered in \cite{decroix2021service} and our negative bias result is not restricted to the logit model for purchasing behavior that \cite{decroix2021service} focuses on. 
In fact, the negative bias result does not hold using their benchmark under our more general purchasing model (see Proposition~\ref{prop:newest_arbitrarily_better_than_knowledgeable_static_price} in Section~\ref{subsec:negative_bias_with_respect_to_no_variability_benchmark_does_not_always_hold_in_our_model}).
Finally, due to our focus on the review ordering policy, we can study non-stationarity and self-selection bias, which is outside the scope of \cite{decroix2021service}. We provide a more extensive comparison in Appendix~\ref{appendix:comparison_negative_bias_due_to_quality_variability}.

\section{Model}
\label{sec:model}
We consider a platform that repeatedly offers a product to customers that arrive 
in consecutive rounds $t=1,2,\ldots$. 
The customer makes a purchase decision based on a finite number of reviews and the price; if a purchase occurs, they leave a new review for the product. 
We consider the platform's decisions regarding the \emph{ordering of reviews} as well as the \emph{price}.

\noindent\textbf{Customer valuation.}
The customer at round $t$ has a realized valuation $V_t = X_t+\Theta_t \in \mathbb{R}$ for the product, where $X_t$ and $\Theta_t$ represent the contribution from the product's unobservable and observable parts respectively.
Specifically, when customer $t$ makes a purchase, $X_t \in \{0, 1\}$ is drawn independently from $Bernoulli(\mu)$ where $\mu \in (0, 1)$ is the same across all customers and is unknown to them.  
Contrastingly, the quantity $\Theta_t$ is customer-specific and is known to customer $t$ before they purchase. We assume that, at every round $t$, $\Theta_t$ is drawn independently from a distribution $\mathcal{F}$ with bounded support. The platform knows the distribution $\mathcal{F}$ but not $\Theta_t$.

If the customer at round $t$ knew their exact valuation $V_t$, then they would purchase the product if and only if $V_t \geq p_t$, where $p_t$ is the price of the product at time $t$.
However, $X_t$ is unknown and hence so is $V_t$. We assume that the customers read reviews to learn about $\mu$, and their purchase decision depends on their belief about $\mu$ after reading the reviews.
Note that customers cannot aim to estimate $X_t$ beyond estimating $\mu$, since $X_t$ is drawn independently for each customer.

\noindent\textbf{Review generation.}
If the customer at round $t$ purchases the product, they write a review that future customers may read.
The review rating given by customer $t$ is equal to $X_t$ (see \cref{subsec:discussion_modeling_assumptions} for a discussion of this assumption). We often refer to a review with $X_t= 1$ as \emph{positive} and to a review with $X_t= 0$ as \emph{negative}.
We refer to $X_{t} \in \{0, 1, \perp\}$ as the \emph{rating of customer $t$'s review} where $X_t = \perp$ if customer $t$ did not purchase the product. 

\subsection{Customer Purchase Behavior}\label{subsec:model_customer_purchase_behavior}
We describe the customer purchase behavior at one round, taking the price and the review ordering as fixed.
Customers have a prior $\mathrm{Beta}(a,b)$ for the value of $\mu$, for some fixed $a, b > 0$. This prior need not be correct and could be based on
information about the features of the product or summary statistics of all reviews which are subject to self-selection bias (see discussion in Section~\ref{subsec:discussion_modeling_assumptions}).
Customers read the first $c$ reviews that are shown to them to update their prior. Formally, letting $\bm{Z}_t = (Z_{t,1}, \ldots, Z_{t,c}) \in \{0,1\}^c$ denote the ratings of the first $c$ reviews shown, the customer creates the following posterior for the unobservable quality $\mu$:
\begin{equation*} 
\Phi_t \coloneqq \mathrm{Beta}\left(a + \sum_{i=1}^{c} Z_{t,i}, b + c-\sum_{i=1}^{c} Z_{t,i}\right). 
\end{equation*}
This corresponds to the natural posterior update for $\mu$ if each $Z_{t,i}$ is an independent draw from $Bernoulli(\mu)$.\footnote{The reviews $\bm{Z}_t$ are not necessarily independent draws from Bernoulli(µ), hence the customers are not completely
Bayesian. See the second-to-last point in Section \ref{sec:relatedwork} for a detailed discussion.} Note that the customer places equal weight on the first $c$ reviews. Based on this posterior, the customer creates an estimate $\hat{V}_t$ for their valuation. We assume that there is a mapping $h(\Phi_t) \in (0, 1)$ from their posterior to a real number that represents an estimate of the fixed valuation~$\mu$.
For example, the mapping $h(\Phi_t) = \expect[\Phi_t]$ represents risk-neutral customers, while if $h(\Phi_t)$ corresponds to the $\phi$-quantile of $\Phi_t$ for $\phi < 0.5$, this represents pessimistic customers \cite{ghkv23}. The customer then forms their estimated valuation
$\hat{V}_t \coloneqq \Theta_t + h(\Phi_t)$ and buys the product at price $p_t$ if and only if $\hat{V}_t \geq p_t$.  Finally, the customer leaves a review $X_t \sim Bernoulli(\mu)$ if they bought the product, otherwise $X_t = \perp$.

To ease exposition, we often use $n$ to denote the number of positive ratings  (i.e., $n = \sum_{i=1}^{c} Z_{t,i}$), and we overload notation to denote by $h(n)$ to refer to $h(\mathrm{Beta}(a + n, b + c-n))$.
We make the natural assumption that higher number of positive ratings leads to a higher purchase probability.
\begin{assumption}\label{assumption:h_monotonicity_positive_reviews}
The estimate $h(n)$ is strictly increasing in the number of positive reviews $n$.
\end{assumption}
We also assume that the idiosyncratic valuation has positive mass on non-negative values.
\begin{assumption}\label{assumption:non-negative_mass_idiosyncratic}
        The distribution $\mathcal{F}$ has positive mass on non-negative values: $\prob_{\Theta \sim \mathcal{F}}[\Theta \geq 0] > 0$. 
\end{assumption}
We denote the above problem instance as $\mathcal{E}(\mu,\mathcal{F},a,b,c, h)$ for product quality $\mu$, prior parameters $a,b$, idiosyncratic distribution  $\mathcal{F}$, customers' attention budget $c$, and an estimate mapping~$h$. 

\subsection{Platform Decisions}\label{subsec: policy_decison}
We consider two platform decisions, review ordering and pricing.

\noindent\textbf{Review ordering.}
With respect to ordering,
since customers only take the first $c$ reviews into account, choosing an ordering is equivalent to selecting a set of $c$ reviews to show.
Let $I_t = \{ \tau < t : X_{\tau} \neq \perp \}$ be the set of previous rounds in which a review was submitted
and let $\mathcal{H}_{t} = \{I_t, \{X_{\tau}\}_{\tau < t}, \{p_{\tau}\}_{\tau< t},  \{\bm{Z}_{\tau}\}_{\tau< t}\}$ be the observed history before round $t$. At round $t$, the platform maps  (possibly in a randomized way) its observed history $\mathcal{H}_{t}$ to the set of review ratings $\sigma(\mathcal{H}_{t})$ corresponding to the $c$ reviews shown. 
We study the steady-state distribution of the system; to avoid initialization corner cases, we assume that at time $t = 1$, there is an infinite pool of reviews $\{X_{\tau}\}_{\tau=-\infty}^{-1}$ where $X_{\tau} \stackrel{\text{i.i.d.}}{\sim}  Bernoulli(\mu)$. 
We consider the following review ordering policies:
\begin{itemize}
    \item $\sigma^{\textsc{newest}}$ selects the $c$ newest reviews. This is formally defined as $\sigma^{\textsc{newest}}(\mathcal{H}_t) = (Z_{t,i})_{i=1}^c$ where $Z_{t,i}$ is the rating of the $i$-th most recent review. 
    \item $\sigma^{\textsc{random}}$ shows $c$ random reviews. 
    This corresponds to ratings being drawn independently from $Bernoulli(\mu)$ at each round $t$; i.e.,
    $\sigma^{\textsc{random}}(\mathcal{H}_t) = (Z_{t,i})_{i=1}^{c}$ where $Z_{t,i} \stackrel{\text{i.i.d.}}{\sim} Bernoulli(\mu)$. We use $\sigma^{\textsc{random}}$ as a benchmark because it maximizes revenue across all policies that only consider the relative arrival sequence and are agnostic to the exact time and rating of reviews (see Appendix~\ref{appendix:why_random_right_benchmark}). 
\end{itemize}

Note that, under $\sigma^{\textsc{random}}$, the customers' and platform's actions at the current round do not influence the reviews shown in future rounds (i.e. the reviews are exogenous). 
In contrast, under $\sigma^{\textsc{newest}}$, the customers' and platform's actions influence what reviews are shown in future rounds (i.e. the reviews are endogenous to the underlying stochastic process).  

\noindent\textbf{Pricing.}
The platform also decides on the pricing policy, where the price at each round can depend on the set of $c$ displayed reviews. 
We denote a pricing policy by a function $\rho: \{0,1\}^{c} \to \mathbb{R}$, which maps the set of displayed review ratings $\bm{z} \in \{0,1\}^{c}$ to a price.
We study two classes of pricing policies: static and dynamic.
We let $\Pi^{\textsc{static}}$ be the set of pricing policies $\rho$ that assign a fixed price $p$, i.e., $\rho(\bm{z}) = p$ for any $\bm{z}$. Similarly, $\Pi^{\textsc{dynamic}}$ includes the set of pricing policies $\rho$ where $\rho(\bm{z})$ can depend on the review ratings~$\bm{z}$. 

\subsection{Revenue and the Cost of Newest First}\label{subsec:rev_cost_of_newest_first_definitions}
For an ordering policy $\sigma$ and pricing policy $\rho$, we define the revenue as the steady-state revenue:~\footnote{The policies we consider have a stationary distribution so, in our analysis, we replace the $\liminf$ with a $\lim$.}
\begin{equation}\label{equation:reveneue_def}
\textsc{Rev}(\sigma, \rho) \coloneqq \liminf_{T \to \infty} \expect\Big[\frac{\sum_{t=1}^{T} \rho(\bm{Z}_t)\mathbf{1}_{\Theta_{t} + h(\Phi_t) \geq \rho(\bm{Z}_t)}}{T}\Big].
\end{equation}
Our main focus lies in understanding the effect of the ordering policy on the revenue. Specifically, we compare the revenues of the ordering policies $\sigma^{\textsc{newest}}$ and $\sigma^{\textsc{random}}$.
For a pricing policy $\rho$, we define the \textsc{Cost of Newest First} (\textsc{CoNF}) as (for all policies $\rho$ we consider, $\textsc{Rev}(\sigma^{\textsc{newest}}, \rho) > 0$): 
$$\chi(\rho) \coloneqq \frac{\textsc{Rev}(\sigma^{\textsc{random}}, \rho)}{{\textsc{Rev}(\sigma^{\textsc{newest}}, \rho)}}.$$ 
If $\rho\in\Pi^{\textsc{static}}$, i.e., $\rho(\bm{z})=p$ for any $\bm{z} \in \{0,1\}^c$ we use $\textsc{Rev}(\sigma, p)$ and $\chi(p)$ as shorthand for the corresponding steady-state revenue and \textsc{CoNF}.

In Section~\ref{sec:pricing}, we study the \textsc{CoNF} when the platform can optimize its pricing policy over a class of policies.
For a class of pricing policies $\Pi$, we define similarly the optimal revenue within-class with respect to an ordering policy $\sigma$ and the corresponding \textsc{CoNF} as:
$$\textsc{Rev}(\sigma, \Pi) \coloneqq \sup_{\rho \in \Pi} \textsc{Rev}(\sigma, \rho) \qquad \text{and} \qquad\chi(\Pi) \coloneqq \frac{\textsc{Rev}(\sigma^{\textsc{random}}, \Pi)}{\textsc{Rev}(\sigma^{\textsc{newest}}, \Pi)} \qquad \text{respectively}.$$
To ease exposition, we make the mild assumption that
$0 < \rev(\sigma, \Pi) < \infty$. \footnote{For the policies we consider, $\rev(\sigma, \Pi) > 0$ by Assumption~\ref{assumption:non-negative_mass_idiosyncratic} and the fact that $h(0)>0$; $\rev(\sigma, \Pi) < \infty$ is satisfied if the maximum revenue from a customer's idiosyncratic valuation is finite, i.e., $\sup_{p} p\prob_{\Theta \sim \mathcal{F}}(\Theta \geq p) < +\infty$. }

Finally, we note that, unlike classical revenue maximization works, our focus is \emph{not} on identifying policies that maximize revenue but rather in comparing the performance of different ordering policies ($\sigma^{\textsc{newest}}$ vs $\sigma^{\textsc{random}}$) and different pricing policies ($\Pi^{\textsc{static}}$ vs $\Pi^{\textsc{dynamic}}$).

\begin{remark}\label{remark_model_generalization}
Our assumption of Bernoulli reviews and Beta prior is made for ease of exposition. Our results extend to a more general model where reviews come from an arbitrary distribution with finite support (not only  $\{0,1\}$) and the estimator $h$ arbitrarily maps reviews to an estimate for the fixed valuation $\mu$; see Appendix \ref{appendix: model_generalization} for details.
\end{remark}

\section{Cost of Newest First (CoNF) with a Fixed Static Price}
\label{sec:ordering}
Throughout this section, we assume a static pricing policy where the price $p$ is fixed and given. 
We establish the main phenomenon of the \textsc{Cost of Newest First} by showing that $\chi(p) > 1$ under very mild conditions on the price $p$.
We then show that $\chi(p)$ can be arbitrarily large.

Recall that $h(n)$ refers to $h(\mathrm{Beta}(a + n, b + c-n))$, where $n \in \{0,\ldots, c\}$ is the number of positive reviews. We first introduce an assumption on two natural conditions that the price satisfies.

\begin{assumption}\label{assumption:non_abs_non_degen}
    For our results on a fixed price $p$, we assume that $p>0$ and that it is: 
    \begin{enumerate}
        \item \emph{Non-absorbing:} the purchase probability is positive for any displayed review ratings; i.e., for all $n \in \{0,1, \dots, c\}$, i.e., $\prob_{\Theta \sim \mathcal{F}}[\Theta + h(n) \geq p] > 0$.
        \item \emph{Non-degenerate:} the purchase probability given all negative review ratings is strictly smaller than given all positive review ratings, i.e.,
    $\prob_{\Theta \sim \mathcal{F}}[\Theta + h(0) \geq p] < \prob_{\Theta \sim \mathcal{F}}[\Theta + h(c) \geq p].$
\end{enumerate}
\end{assumption}
A non-absorbing price guarantees that $\sigma^{\textsc{newest}}$ does not get ``stuck'' in a zero-revenue state. A non-degenerate price implies that the review ratings \emph{matter}, since there exist distinct review ratings where the purchase probability differs.

\subsection{Existence of Cost of Newest First}\label{subsec: negative_recency_bias}
Our main result is that the revenue under $\sigma^{\textsc{random}}$ is strictly higher than that of $\sigma^{\textsc{newest}}$; that is, $\chi(p) > 1$.
As a building block towards this result, we first provide simple and interpretable closed form expressions for $\textsc{Rev}(\sigma^{\textsc{random}}, p)$ and $\textsc{Rev}(\sigma^{\textsc{newest}}, p)$ for a static price $p$, which are given in Propositions \ref{theorem:random_C_revenue} and \ref{theorem:most_recent_C_revenue} respectively. Let $\bern(\mu)$ denote the Bernoulli distribution with success probability $\mu$ and $\binomial(c, \mu)$ denote the Binomial distribution with $c$ i.i.d. $\bern(\mu)$ trials. 

\begin{proposition}[Revenue of $\sigma^{\textsc{random}}$]\label{theorem:random_C_revenue}
    For any fixed price $p$,
    $$\textsc{Rev}(\sigma^{\textsc{random}}, p)=p \expect_{N \sim \binomial(c, \mu)}\Big[\prob_{\Theta \sim \mathcal{F}}[\Theta + h(N) \geq p] \Big].$$
\end{proposition}
\begin{proof}
 By definition, $\sigma^{\textsc{random}}$ displays $c$ i.i.d.\ $\bern(\mu)$ reviews at every round $t$. As a result, the number of positive reviews is 
 distributed as $\binomial(c, \mu)$, yielding expected revenue, at every round $t$, equal to the right hand side of the theorem. Given that this quantity does not depend on $t$, recalling Eq.\eqref{equation:reveneue_def}, it equals the steady-state revenue. 
\end{proof}
Unlike $\sigma^{\textsc{random}}$ which displays $c$ i.i.d. $\bern(\mu)$ reviews at every round, the reviews displayed by $\sigma^{\textsc{newest}}$ are an endogenous function of the history. The proof of the next result underlies the technical crux of this section and is presented in Section~\ref{subsec:proof_most_recent_revenue}. 
\begin{proposition}[Revenue of $\sigma^{\textsc{newest}}$ ]\label{theorem:most_recent_C_revenue}
    For any fixed  price $p$ satisfying Assumption~\ref{assumption:non_abs_non_degen}, 
    $$\textsc{Rev}(\sigma^{\textsc{newest}}, p) = \frac{p}{\expect_{N \sim \binomial(c, \mu)}\Big[\frac{1}{\prob_{\Theta \sim \mathcal{F}}[\Theta + h(N) \geq p]} \Big]}.$$
\end{proposition}
Intuitively, when the newest reviews are positive, the customer is more likely to buy the product and leave a new review, which then updates the set of newest reviews.
On the other hand, when the newest reviews are negative, the customer is less likely to buy, and hence the set of newest reviews is less likely to be updated.
This implies that $\sigma^{\textsc{newest}}$ spends more time in a state with negative reviews (which yield lower revenue) compared to $\sigma^{\textsc{random}}$.
This phenomenon is the driver of our main result and we refer to it as  the \textsc{Cost of Newest First} (CoNF).
\begin{theorem}[\textsc{Cost of Newest First}]\label{theorem:negative_bias} 
For any fixed price $p$ satisfying Assumption~\ref{assumption:non_abs_non_degen}, the revenue of $\sigma^{\textsc{newest}}$ is strictly smaller than that of $\sigma^{\textsc{random}}$, i.e., $\textsc{Rev}(\sigma^{\textsc{random}}, p) > \textsc{Rev}(\sigma^{\textsc{newest}}, p)$. 
\end{theorem}
\begin{proof} We show that the expression of \cref{theorem:random_C_revenue} is higher than the one of \cref{theorem:most_recent_C_revenue}.
By Jensen's inequality,  $\expect[X] \geq \frac{1}{\expect[\frac{1}{X}]}$ for any non-negative random variable $X$ and equality is achieved if and only if $X$ is a constant. Letting $X(N) \coloneqq \prob_{\Theta \sim \mathcal{F}}[\Theta + h(N) \geq p]$ be the purchase probability in a state with $N$ positive reviews, we apply the inequality for $X = X(N)$,
$$\expect_{N \sim \binomial(c, \mu)}\big[\prob_{\Theta \sim \mathcal{F}}[\Theta + h(N) \geq p]\big] \geq \frac{1}{\expect_{N \sim \binomial(c, \mu)}\Big[\frac{1}{\prob_{\Theta \sim \mathcal{F}}[\Theta + h(N) \geq p]}\Big]}.$$
Multiplying with $p > 0$ on both sides we obtain that 
$\rev(\sigma^{\textsc{random}}, p) \geq \rev(\sigma^{\textsc{newest}},p)$. By Assumption \ref{assumption:non_abs_non_degen}, it holds that $\prob_{\Theta \sim \mathcal{F}}[\Theta + h(0) \geq p] < \prob_{\Theta \sim \mathcal{F}}[\Theta + h(c) \geq p]$ and thus $\prob_{\Theta \sim \mathcal{F}}[\Theta + h(N) \geq p]$ is not a constant random variable when $N \sim \binomial(c, \mu)$. Therefore, the inequality is strict. 
\end{proof}

We describe a simple example that provides intuition on the \textsc{CoNF} established in Theorem \ref{theorem:negative_bias}, which is also illustrated in Figure \ref{fig:example_conf}.
\begin{example}\label{rmk:example_transitions_neg_bias}
Suppose  $\mu = \frac{1}{2}$, $\mathcal{F} = \mathcal{U}[0,1]$, $c =a =b =1$, $h(\Phi_t) = \mathbb{E}[\Phi_t] =  \frac{N+1}{3}$, and $p = 1$. Under $\sigma^{\textsc{random}}$ the probability that the review shown is positive is~$0.5$ (see Figure \ref{subfig:random}).
Under $\newest$ the purchase probability is $\frac{2}{3}$ when the review is positive and $\frac{1}{3}$ when it is negative. Thus, under $\sigma^{\textsc{newest}}$, transitioning from a positive to a negative review is twice as likely as transitioning from a negative to a positive review (see Figure \ref{subfig:newest}). Hence, a negative rating is twice as likely as positive rating. This leads to lower revenue in steady state for $\sigma^{\textsc{newest}}$ compared to $\sigma^{\textsc{random}}$. 
\end{example}

\begin{figure*}[t!]
\centering
\begin{subfigure}[t]{0.35\textwidth}
    \centering
    \begin{tikzpicture}[scale=0.7] 
        \node[draw, rectangle, minimum width=1.3cm, minimum height=0.6cm] (rectangle1) at (2,0) {$(1)$ (positive)};
        \node[draw, rectangle, minimum width=1.3cm, minimum height=0.6cm] (rectangle2) at (6,0) {$(0)$ (negative)};
        \draw[->, black!60!green, ultra thick] (rectangle1) to[bend left] node[midway, above, sloped, text=black!60!green] {$\nicefrac{1}{2}$} (rectangle2);
        \draw[->, black!60!green, ultra thick] (rectangle2) to[bend left] node[midway, below, sloped, text=black!60!green] {$\nicefrac{1}{2}$} (rectangle1);  \draw[->, loop above, looseness=8, in=60, out=120, black!30!red, ultra thick] (rectangle1) edge node[above] {$\nicefrac{1}{2}$} (rectangle1);
        \draw[->, loop above, looseness=8, in=60, out=120, black!30!red, ultra thick] (rectangle2) edge node[above] {$\nicefrac{1}{2}$} (rectangle2);
    \end{tikzpicture}
    \caption{$\sigma^{\textsc{random}}$ transitions.} \label{subfig:random}
\end{subfigure}
\hspace{2cm} 
\begin{subfigure}[t]{0.35\textwidth}
    \centering
    \begin{tikzpicture}[scale=0.7] 
        \node[draw, rectangle, minimum width=1.3cm, minimum height=0.6cm] (rectangle1) at (0,0) {$(1)$ (positive)};
        \node[draw, rectangle, minimum width=1.3cm, minimum height=0.6cm] (rectangle2) at (4,0) {$(0)$ (negative)};
        \draw[->, black!60!green, ultra thick] (rectangle1) to[bend left] node[midway, above, sloped, text=black!60!green] {$\nicefrac{1}{3}$} (rectangle2);
        \draw[->, black!60!green, ultra thick] (rectangle2) to[bend left] node[midway, below, sloped, text=black!60!green] {$\nicefrac{1}{6}$} (rectangle1);
        
        \draw[->, loop above, looseness=8, in=60, out=120, black!30!red, ultra thick] (rectangle1) edge node[above] {$\nicefrac{2}{3}$} (rectangle1);
        \draw[->, loop above, looseness=8, in=60, out=120, black!30!red, ultra thick] (rectangle2) edge node[above] {$\nicefrac{5}{6}$} (rectangle2);
    \end{tikzpicture}    \caption{$\sigma^{\textsc{newest}}$ transitions.}
    \label{subfig:newest}
\end{subfigure}
\caption{The state transitions in the instance of the \cref{rmk:example_transitions_neg_bias} under $\sigma^{\textsc{random}}$ and $\sigma^{\textsc{newest}}$.}
\label{fig:example_conf}
\end{figure*}

Lastly, we generalize our result beyond revenue loss by analyzing the \textit{number of positive reviews} among the $c$ reviews. 
We show that the expected number of positive reviews is smaller under $\newest$ compared to $\random$ (proof in Appendix \ref{appendix:avg_newest_smaller_avg_random}). Let $\pi_{n}^{\textsc{newest}}$ (respectively $\pi_{n}^{\textsc{random}}$) denote the steady-state probability of observing $n$ positive reviews under $\newest$ (respectively $\random$).

\begin{proposition}
\label{thm_neg_bias_avg_review_rating}
     For any price $p$ satisfying Assumption~\ref{assumption:non_abs_non_degen}, the average number of positive reviews under $\newest$ is smaller than under $\random$. Formally, $\expect_{N \sim \pi^{\textsc{newest}}_n}[N] < \expect_{N \sim \pi^{\textsc{random}}_n}[N]$. 
\end{proposition}

\subsection{Characterization of Revenue under Newest (Proof of Proposition~\ref{theorem:most_recent_C_revenue})}\label{subsec:proof_most_recent_revenue}
We provide the proof of \cref{theorem:most_recent_C_revenue}, which contains the main technical crux of this section.
We first introduce some notation. Recalling that $Z_{t,i}$ denotes the rating of the $i$-th most recent review at round $t$, we refer to the $c$ most recent reviews by the vector $\bm{Z}_t = (Z_{t,1}, \ldots, Z_{t,c})$. We note that $\bm{Z}_t$ is a time-homogeneous Markov chain with a finite state space $\{0,1\}^{c}$. Given that we assume an infinite pool of initial reviews, $\bm{Z}_1=(X_{-1},\ldots, X_{-c})$ where $X_i \sim \bern(\mu)$ for $i \in \{-1, \ldots, -c\}$.

With respect to the transition dynamics of this Markov chain, for the state $\bm{Z}_t = (z_1, \ldots, z_{c}) \in \{0,1\}^{c}$,
the purchase probability is 
$\prob_{\Theta \sim \mathcal{F}}\Big[\Theta + h(\sum_{i=1}^{c} z_i) \geq p\Big]$.
If there is no purchase, no review is given and the state remains $\bm{Z}_{t+1} =(z_1, \ldots, z_{c})$. If there is a purchase, $\bm{Z}_{t}$ transitions to the state $\bm{Z}_{t+1} = (1, z_1, \ldots, z_{c-1})$ if the review is positive (with probability $\mu$) and to the state $\bm{Z}_{t+1} = (0,z_1, \ldots, z_{c-1})$ if the review is negative (with probability $1-\mu$).

Because the price $p$ is non-absorbing (Assumption~\ref{assumption:non_abs_non_degen}), for every state of reviews $(z_1, \ldots, z_{c}) \in \{0,1\}^{c}$,
the purchase probability is positive, and the probability that a new review is positive is strictly greater than zero (since $\mu \in (0, 1)$).  Then, $\bm{Z}_t$ can reach every state from every other state with positive probability (i.e. it is a single-recurrence-class Markov Chain with no transient states), and hence $\bm{Z}_t$ has a unique stationary distribution, which we denote by $\pi$. Our next lemma exactly characterizes the form of this stationary distribution.

\begin{lemma}\label{lemma:stationary_state_distribiton_newest_first}
    The stationary distribution $\pi$ of $\bm{Z}_t$ under any price $p$ satisfying Assumption \ref{assumption:non_abs_non_degen} is    
    \begin{equation*} \label{eq:stationary_distribution_most_recent}
         \pi_{(z_1, \ldots, z_{c})} = \kappa \cdot  \frac{\mu^{\sum_{i=1}^{c} z_i} (1-\mu)^{c-\sum_{i=1}^{c} z_i}}{\prob_{\Theta \sim \mathcal{F}}\Big[\Theta + h(\sum_{i=1}^{c} z_i) \geq p\Big] },
    \end{equation*}
    where $\kappa = 1/\expect_{N \sim \binomial(c, \mu)}\Big[\frac{1}{\prob_{\Theta \sim \mathcal{F}}[\Theta + h(N) \geq p]} \Big]$ is a normalizing constant.
\end{lemma}

\begin{proof}[Proof sketch.]
If the reviews were drawn i.i.d. at each round, the probability of state $(z_1, \ldots, z_{c})$ would be exactly $\mu^{\sum_{i=1}^{c} z_i} (1-\mu)^{c-\sum_{i=1}^{c} z_i}$, which is the numerator.
However, the set of newest reviews is only updated when there is a purchase, which occurs with probability $\prob_{\Theta \sim \mathcal{F}}\Big[\Theta + h(\sum_{i=1}^{c} z_i) \geq p\Big]$.
Hence, we multiply the numerator by $1/\prob_{\Theta \sim \mathcal{F}}\Big[\Theta + h(\sum_{i=1}^{c} z_i) \geq p\Big]$, which is the expected number of rounds until there is a new review under state $(z_1, \dots, z_c)$; in fact, we show such a property holds for general Markov chains (\cref{lemma: general_theorem_markov_chains_stationary}). A formal proof is provided in Appendix~\ref{appendx:proof_lemma:stationary_state_distribiton_newest_first}.
\end{proof}

Equipped with Lemma \ref{lemma:stationary_state_distribiton_newest_first}, we now prove \cref{theorem:most_recent_C_revenue}.
\begin{proof}[Proof of \cref{theorem:most_recent_C_revenue}.]
Using Eq.~\eqref{equation:reveneue_def}, the revenue of $\sigma^{\textsc{newest}}$ can be written as
\begin{align*}
    \textsc{Rev}(\sigma^{\textsc{newest}}, p) 
    &= p \liminf_{T \to \infty} \frac{\expect\left[\sum_{t=1}^{T} \prob_{\Theta \sim \mathcal{F}}\left[\Theta + h(\sum_{i=1}^c Z_{t,i}) \geq p\right]\right]}{T} \\
    &= p\sum_{(z_1, \ldots, z_c) \in \{0,1\}^{c}} \pi_{(z_1, \ldots, z_{c})} \prob_{\Theta \sim \mathcal{F}}\left[\Theta + h(\sum_{i=1}^c z_i) \geq p \right],
\end{align*}
where the second step expresses the revenue of the stationary distribution via the Ergodic theorem. 
Expanding $\pi_{(z_1, \ldots, z_{c})}$ based on Lemma~\ref{lemma:stationary_state_distribiton_newest_first}, the $\prob_{\Theta \sim \mathcal{F}}[\Theta + h(\sum_{i=1}^c z_i) \geq p]$ term cancels out and: 
    \begin{align*}
    \textsc{Rev}(\sigma^{\textsc{newest}}, p) 
       &= p \cdot \kappa \cdot \Big(\sum_{(z_1, \ldots, z_c) \in \{0,1\}^{c}} \mu^{\sum_{i=1}^{c} z_i} (1-\mu)^{c-\sum_{i=1}^{c}z_i} \Big).
    \end{align*}
    Note that the term in the parenthesis equals 1, since it is a sum over all probabilities of $\binomial(c, \mu)$.
    This yields $\rev(\sigma^{\textsc{newest}}, p) = p \cdot \kappa$, which gives the expression in the theorem.
\end{proof}

\subsection{Cost of Newest First can be arbitrarily bad}\label{subsec:quantify_cost_of_newest_first}
\cref{theorem:negative_bias} implies that, for any price $p$ satisfying Assumption \ref{assumption:non_abs_non_degen}, the \textsc{CoNF} is strictly greater than $1$, i.e., $\chi(p)> 1$.\footnote{For any price $p$ satisfying Assumption \ref{assumption:non_abs_non_degen}, the denominator $\rev(\newest, p)$ of $\chi(p)$ is strictly positive.} We now show that it can be arbitrarily large. We first provide a closed-form expression for $\chi(p)$ by dividing the expressions in Propositions \ref{theorem:random_C_revenue} and \ref{theorem:most_recent_C_revenue} (see Appendix \ref{appendix_subsec:CoNF_ratio_closed_form_general_c} for proof details).

\begin{lemma}\label{lemma: CoNF_ratio_closed_form_general_c}
For any price $p$ satisfying Assumption \ref{assumption:non_abs_non_degen}, the CoNF is given by:
\begin{equation*}\chi(p) = \sum_{i,j \in \{0, \ldots,c\}} \mu^{i+j}(1-\mu)^{2c-i-j} \binom{c}{i} \binom{c}{j} \frac{ \prob_{\Theta \sim \mathcal{F}}[\Theta + h(i) \geq p]}{ \prob_{\Theta \sim \mathcal{F}}[\Theta + h(j) \geq p]}.
\end{equation*}
\end{lemma}

\begin{theorem}
\label{thm:conf_arbitrarily_bad}
For any continuous value distribution $\mathcal{F}$ with positive mass on a bounded support, and any $M> 0$, there exists a price $p$ satisfying Assumption \ref{assumption:non_abs_non_degen} such that $\chi(p) > M$. 
\end{theorem}

\begin{proof}[Proof sketch.]
One summand in the right hand side of \cref{lemma: CoNF_ratio_closed_form_general_c} contains the ratio of the purchase probability of all reviews being positive compared to all reviews being negative:

\begin{equation}\label{eq:ratio_all_positive_all_negative}
    \beta(p) \coloneqq \frac{ \prob_{\Theta \sim \mathcal{F}}[\Theta + h(c) \geq p]}{ \prob_{\Theta \sim \mathcal{F}}[\Theta + h(0) \geq p]},
\end{equation}
which quantifies how much the reviews affect the purchase probability. Since all other terms are non-negative, the \textsc{CoNF} is lower bounded by this summand, i.e., $\chi(p) \geq \mu^c(1-\mu)^c \beta(p)$. The full proof (provided in Appendix \ref{appendix_subsec_conf_arb_bad_proof}) shows that, for any $M$ the latter quantity can become equal to $M$ by an appropriate selection of price that makes Newest to be mostly ``stuck'' (purchase probability when all reviews are negative is close to $0$ while the purchase probability when all reviews are positive is at least a constant).
\end{proof}

We complement \cref{thm:conf_arbitrarily_bad} by showing that if $\mathcal{F}$ has Monotone Hazard Rate (MHR), i.e., letting $F$ and $f$ be its cumulative and density functions, its hazard rate function $\frac{f(u)}{1-F(u)}$ is non-decreasing in $u$, then $\chi(p)$ is non-decreasing in $p$. The intuition behind this result is that the MHR property implies that each purchase probability ratio, $\frac{ \prob_{\Theta \sim \mathcal{F}}[\Theta + h(i) \geq p]}{ \prob_{\Theta \sim \mathcal{F}}[\Theta + h(j) \geq p]}$ for $i > j$, is non-decreasing in~$p$. Given that the CoNF is a linear combination of such ratios and their reciprocals (Lemma~\ref{lemma: CoNF_ratio_closed_form_general_c}), it is non-decreasing in~$p$. The formal proof is provided in Appendix \ref{appendix_subsec:conf_monotonicity_under_mhr}.

\begin{proposition}\label{thm:mhr_conf_increasing_in_price}
   Suppose that $\mathcal{F}$ is a continuous distribution with support $[\underline{\theta}, \overline{\theta}]$ and has MHR. Then $\chi(p)$ is non-decreasing for $p \in (\underline{\theta} + h(c), \overline{\theta} + h(0))$.
\end{proposition}

Lastly, we show in the following proposition (proof in Appendix~\ref{appendix:when_is_conf_small}) that CoNF is upper bounded by the ratio $\beta(p)$ defined in \eqref{eq:ratio_all_positive_all_negative}. As a result, when $\beta(p)$ is small, the \textsc{Cost of Newest First} is also small. This occurs when review ratings have small impact on purchases. For example, when $\mathcal{F} = \mathcal{U}[0,\overline{\theta}]$ and $\overline{\theta} \to \infty$, the idiosyncratic variability dominates the variability from estimating $\mu$ through reviews, yielding $\beta(p)\to 1$ and thus $\chi(p) \to 1$. 

\begin{proposition}\label{prop:CoNF_less_beta}
    For all prices $p$ satisfying Assumption \ref{assumption:non_abs_non_degen}, the CoNF is at most $\chi(p) \leq \beta(p)$. 
\end{proposition}

\subsection{Limited Attention: Main Driver for the Persistence of Negative Reviews }\label{subsec:driver_persistence_of_negative_reviews}

We further investigate the key drivers of the persistence of negative reviews. We show that this phenomenon is not an artifact of the specific belief update rule of our model; rather the limited attention of customers is the main driver behind our result. 

\noindent \textbf{The phenomenon holds for generic belief updating rules.}
Consider customers who read the newest review $(c=1)$ and any updating rule where at any round $t$, the purchase probability when the last review is positive is strictly greater than when it is negative.
Specifically, let $q(r,t)$ be the purchase probability of the customer at round $t \in \mathbb{N}$ when the newest review is $r \in \{0, 1\}$, and suppose $q(0,t) < q(1,t)$ for every $t$.
Suppose that any new review is equally likely to be positive or negative ($\mu = 1/2$).
We show that at any round, the newest review is more likely to be negative than positive (\cref{prop:conf_exists_general_update_rule},
Appendix~\ref{app:driver_generic_updating}).

\noindent\textbf{The phenomenon disappears without limited attention.} To isolate the effect of limited attention, we contrast to the setting where customers do not have limited attention and they can take the full history into account ($c\rightarrow \infty$). We show that there exist instances where, under any limited attention~$\tilde{c}$, the CoNF can be arbitrarily large, while it goes to $1$ when $c\rightarrow \infty$   (\cref{prop:for_any_c_M_instance_1_bad_conf_2_conf_to_one_without_bd_rationality}, Appendix \ref{appendix_subsec_conf_driver_bounded_rationality}); the latter result holds more broadly under mild assumptions on the purchase probability (\cref{prop:more_broadly_conf_to_one_without_bd_rationality}, Appendix \ref{appendix_subsec_conf_driver_bounded_rationality}).

\section{Dynamic Pricing Mitigates the Impact of CoNF}
\label{sec:pricing}
In this section, we allow the platform to optimize the pricing policy $\rho$, while the review ordering policy is either $\sigma^{\textsc{newest}}$ or $\sigma^{\textsc{random}}$. We assume that the platform knows the true underlying quality~$\mu$.\footnote{We assume $\mu$ is fixed over time and  the platform has access to enough reviews to estimate $\mu$ arbitrarily well. } 
Recall that $\Pi^{\textsc{static}}$ and  $\Pi^{\textsc{dynamic}}$ are the classes of static and dynamic pricing policies respectively, and that the revenue and CoNF for a class $\Pi$ are defined respectively as
$$\textsc{Rev}(\sigma, \Pi) \coloneqq \sup_{\rho \in \Pi} \textsc{Rev}(\sigma, \rho) \qquad \text{ and } \qquad \chi(\Pi)  \coloneqq \frac{ \rev(\sigma^{\textsc{random}}, \Pi)}{\rev(\sigma^{\textsc{newest}}, \Pi)}.$$

The main results of this section (\cref{subsec:recency_bias_ratio_static_optimal}) establish that the \textsc{CoNF} can be arbitrarily large for the optimal static pricing policy (\cref{thm:CoNF_opt_static_arbitrarily_bad}) but that it is bounded by a small constant for the optimal dynamic pricing policy (\cref{theorem:dynamic_pricing_CoNF_bound}). The main technical challenge of this section is in proving \cref{theorem:dynamic_pricing_CoNF_bound}.
To do this, we first \textit{characterize} the optimal dynamic pricing policies under both $\sigma^{\textsc{newest}}$ and $\sigma^{\textsc{random}}$ and derive exact expressions for their long-term revenue (\cref{subsec: optimize_dynamic_price}). 
In doing so, we derive a structural property of the optimal dynamic pricing policy under $\sigma^{\textsc{newest}}$: the prices ensure that the purchase probability is always equal regardless of the state of reviews.

\subsection{Cost of Newest First under Optimal Static and Dynamic Pricing}\label{subsec:recency_bias_ratio_static_optimal}
We first establish that when optimizing over static prices, the \textsc{CoNF} can be arbitrarily large for any number of reviews $c$.
Note that this is not implied by \cref{thm:conf_arbitrarily_bad}, since here we assume the platform always chooses the \emph{optimal} static price for a given instance.

\begin{theorem}\label{thm:CoNF_opt_static_arbitrarily_bad}
   For any instance where the support of $\mathcal{F}$ is $[0, \overline{\theta}]$, it holds that $\chi(\Pi^{\textsc{static}}) \geq \frac{\mu^ch(c)}{h(0)+\overline{\theta}}$. 
\end{theorem} 
This implies that $\chi(\Pi^{\textsc{static}}) \to +\infty$ if $h(c)$ is held constant and $\overline{\theta} \to 0$ and $h(0) \to 0$. Intuitively, this means that when the variability in the customer's idiosyncratic valuation $\Theta_t$ is negligible compared to the variability in review-inferred quality estimates, 
$\sigma^{\textsc{newest}}$ spends a disproportionate time in the state with no positive reviews, which leads to unbounded \textsc{CoNF}.

\begin{proof}[Proof of \cref{thm:CoNF_opt_static_arbitrarily_bad}.]
The optimal revenue under $\newest$ is at most $h(0) + \overline{\theta}$. This is because any price $p > h(0) + \overline{\theta}$ induces a purchase probability of zero when all reviews are negative and thus $\newest$ gets ``stuck" in a zero revenue state. Thus, $\max_{p \in \mathbb{R}}\rev(\sigma^{\textsc{newest}},p) \leq h(0) + \overline{\theta}$.

Under $\sigma^{\textsc{random}}$, 
if all reviews are positive, the non-negativity of the value distribution implies that a price $p = h(c)$ induces a purchase with probability one.
The probability of this event is~$\mu^c$, which implies that $\max_{p \in \mathbb{R}}\rev(\sigma^{\textsc{random}},p) \geq \mu^c h(c) $. Combining the two inequalities we obtain  
\begin{equation*}\chi(\Pi^{\textsc{static}}) = \frac{\max_{p \in \mathbb{R}}\rev(\sigma^{\textsc{random}},p)}{\max_{p \in \mathbb{R}}\rev(\sigma^{\textsc{newest}},p)}  \geq \frac{\mu^ch(c)}{h(0)+\overline{\theta}}. \end{equation*}
 \end{proof}

Next, we show an upper bound under dynamic pricing, which is the main result of this section. 
\begin{theorem}\label{theorem:dynamic_pricing_CoNF_bound}
For any instance, it holds that $\chi(\Pi^{\textsc{dynamic}})  \leq \frac{2}{\prob_{\Theta \sim \mathcal{F}}[\Theta \geq 0]}$.
\end{theorem}

In contrast to static pricing where the \textsc{CoNF} can be arbitrarily bad, \cref{theorem:dynamic_pricing_CoNF_bound} shows that its negative impact is uniformly bounded under dynamic pricing.
If the idiosyncratic valuation $\Theta_t$ is always non-negative, the upper bound on $\chi(\Pi^{\textsc{dynamic}})$ is 2.
Theorem \ref{theorem:dynamic_pricing_CoNF_bound} applies even when $\Theta_t$ can be negative: for example if it is non-negative with probability 1/2, then $\chi(\Pi^{\textsc{dynamic}}) \leq 4$. 

 The proof of \cref{theorem:dynamic_pricing_CoNF_bound} (\cref{subsec:proof_conf_bound_dynamic_pricing}) relies on characterizing the optimal dynamic pricing policies  under both $\sigma^{\textsc{newest}}$ and $\sigma^{\textsc{random}}$ and their corresponding revenues (\cref{subsec: optimize_dynamic_price}). Recall that with static pricing, $\newest$ spends a disproportionate amount of time in a negative review state compared to $\random$. In contrast, the optimal dynamic pricing sets prices so that the purchase probability is equal across all review states, leading to $\newest$ and $\random$ spending the same amount of time in each review state (Section~\ref{subsec: optimize_dynamic_price}). This allows us to bound the ratio of demands under $\newest$ and $\random$ by a factor of $\prob_{\Theta \sim \mathcal{F}}[\Theta \geq 0]$. Finally, we bound the ratio of the optimal prices under $\newest$ and $\random$, which we decompose in two terms stemming from the customer's belief about $\mu$ and customer specific valuation; each term is bounded by 1 (Section~\ref{subsec:proof_conf_bound_dynamic_pricing}). 

We complement this result by a lower bound (proof in Appendix \ref{appendix_proof_CoNF_example_at_least_alpha}) which shows that the \textsc{Cost of Newest First} still exists even under optimal dynamic pricing.
\begin{proposition}\label{prop:conf_still_exits_opt_dynamic_pricing}
    For any $\alpha<4/3$, there exists an instance such that $\chi(\Pi^{\textsc{dynamic}}) > \alpha$.  
\end{proposition}

\begin{remark}\label{rmk:CoNF_dynamic_at_least_one}
Even under optimal dynamic pricing, it is still the case that $\random$ induces no smaller revenue than $\newest$, i.e., $\chi(\Pi^{\textsc{dynamic}}) \geq 1$ (see Appendix \ref{appendix_subsec:dynamic_pricing_CoNF_geq_1_proof}).
\end{remark}

\begin{remark}\label{rmk_result_imporvement_CoNF_dynamic_pricing} If we have the additional knowledge of $h(n) \leq u$ for some $u \geq 0$, then we can improve the result of \cref{theorem:dynamic_pricing_CoNF_bound}  to $\chi(\Pi^{\textsc{dynamic}})  \leq \frac{2\prob_{\Theta \sim \mathcal{F}}[\Theta \geq -u]}{\prob_{\Theta \sim \mathcal{F}}[\Theta \geq 0]}$ (see Appendix \ref{appendix:result_imporvement_CoNF_dynamic_pricing}).
\end{remark}

\subsection{Characterization of Optimal Dynamic Pricing under Newest First}\label{subsec: optimize_dynamic_price}
Our key technical contribution characterizes the optimal dynamic pricing policy under $\sigma^{\textsc{newest}}$. We show that it satisfies a structural property: the purchase probability is equal regardless of the review state. We define the policies that satisfy this property as \textit{review-offsetting} policies. For a state of reviews $\bm{z} = (z_1, \ldots, z_c) \in \{0,1\}^c$, we denote by $N_{\bm{z}} = \sum_{i=1}^c z_i$ the number of positive review ratings. 

\begin{definition}
    A dynamic pricing policy $\rho$ is \textit{review-offsetting} if there exists an \textit{offset} $a \in \mathbb{R}$ such that $\rho(\bm{z}) = h(N_{\bm{z}}) +a$ for all  $\bm{z} \in \{0,1 \}^c$. 
\end{definition}

Note that for a review-offsetting policy $\rho$, the purchase probability at any state $\bm{z}\in \{0,1 \}^c$ is $\prob_{\Theta \sim \mathcal{F}}\big[\Theta + h(N_{\bm{z}}) \geq h(N_{\bm{z}}) + a\big] = \prob_{\Theta \sim \mathcal{F}}[\Theta \geq a]$, where the last term does not depend on $\bm{z}$.
Hence, review-offsetting policies induce equal purchase probability regardless of the state of reviews.

The main result of this section establishes that under $\sigma^{\textsc{newest}}$, there is a review-offsetting dynamic pricing policy that maximizes revenue and characterizes the corresponding offset. We note that this characterization is the only place where we require the platform to know the true quality $\mu$. 

\begin{theorem}\label{theorem:newest_first_dynamic_optimal_policy}
Let $p^{\star} \in \argmax_{p \in \mathbb{R}}p \prob_{\Theta \sim \mathcal{F}}\big[\Theta + \expect_{N \sim \binomial(c, \mu)}[h(N)] \geq p\big]$. 
Under $\newest$, 
the review-offsetting policy with offset $a^{\star}=p^{\star} -\expect_{N \sim \binomial(c, \mu)}[h(N)]$ is an optimal dynamic pricing policy.
\end{theorem}

To prove \cref{theorem:newest_first_dynamic_optimal_policy}, we need a characterization of the revenue for a dynamic pricing policy $\rho$ (similar to \cref{theorem:most_recent_C_revenue}). This is established in the following lemma (proof in Appendix~\ref{appendix:dynamic_pricing_revenue_formula}).

\begin{lemma}\label{lemma:revenue_recent_dynamic_formula_lemma_improving_review_offsetting_main_body}
For any dynamic pricing policy $\rho$ with positive purchase probabilities in all states,\footnote{In this lemma and the following proofs we again use the notation $\bm{Y} = (Y_1, \ldots, Y_c)$ and $N_{\bm{Y}} = \sum_{i=1}^c Y_i$.}
\begin{equation*}\label{:revenue_recent_dynamic_formula_lemma_improving_review_offsetting_main_body}
 \textsc{Rev}(\sigma^{\textsc{newest}}, \rho) = \frac{\expect_{Y_1, \ldots, Y_c \sim_{i.i.d.} \bern(\mu)}[\rho(\bm{Y})]}{\expect_{Y_1, \ldots, Y_c \sim_{i.i.d.} \bern(\mu)}\Big[\frac{1}{\prob_{\Theta \sim \mathcal{F}} [\Theta + h(N_{\bm{Y}}) \geq \rho(\bm{Y}) ]}\Big]}.
 \end{equation*}
\end{lemma}

\begin{proof}[Proof of \cref{theorem:newest_first_dynamic_optimal_policy}.]
The high-level idea of the proof is that any dynamic pricing policy $\rho$ can be improved by a particular review-offsetting policy. Specifically,
for any dynamic pricing policy $\rho$ and any review state $\bm{z} \in \{0,1\}^c$,  
we define a policy $\Tilde{\rho}_{\bm{z}}$ to be the review-offsetting policy with offset $a_{\bf{z}} \coloneqq \rho(\bm{z}) -h(N_{\bm{z}})$ and show that the revenue of one of $\{\Tilde{\rho}_{\bm{z}}\}_{ \bm{z} \in \{0,1\}^c}$ is at least the revenue of $\rho$.

To establish this we express the revenue of any review-offsetting policy $\tilde{\rho}$ with offset $a$, which satisfies that $\Tilde{\rho}(\bm{z}) = a+h(N_{\bm{z}})$ for all states $\bm{z}$. As a result, (a) its expected price is $\expect[\Tilde{\rho}(\bm{Y})] = a + \overline{h}$ where $\overline{h} = \expect_{Y_1, \ldots, Y_c \sim_{i.i.d.} \bern(\mu)}[h(N_{\bf{Y}})]$ and (b) its purchase probability is always equal to $\prob_{\Theta \sim \mathcal{F}}\big[\Theta + h(N_{\bm{Y}})\geq \Tilde{\rho}(\bm{Y})\big] = 
\prob_{\Theta \sim \mathcal{F}}\big[\Theta + \overline{h} \geq a + \overline{h}\big]$. By \cref{lemma:revenue_recent_dynamic_formula_lemma_improving_review_offsetting_main_body}, its revenue is equal to
\begin{equation}\label{eq:revenue_review_offsetting}
    \rev(\sigma^{\textsc{newest}},\Tilde{\rho}) =\big(a + \overline{h}\big) \prob_{\Theta \sim \mathcal{F}}\big[\Theta + \overline{h} \geq a + \overline{h}\big].
\end{equation}
Note that, across all review-offsetting policies, the offset that maximizes the above revenue is $a^{\star}=p^{\star}-\overline{h}$; we refer to the corresponding review-offsetting policy as $\rho^{\star}$. We next show that the revenue of any dynamic pricing policy is upper bounded by the revenue of $\rho^{\star}$, concluding the proof.
Given that the policy $\Tilde{\rho}_{\bm{z}}$ is review-offsetting with offset $a_{\bf{z}}$, the above implies 
\begin{equation*}
    \rev(\sigma^{\textsc{newest}}, \Tilde{\rho}_{\bm{z}}) = \underbrace{\Big(\rho(\bm{z})-h(N_{\bm{z}})+ \overline{h} \Big)}_{A_{\bm{z}}} \cdot \underbrace{\prob_{\Theta \sim \mathcal{F}}\big[\Theta + h(N_{\bm{z}}) \geq \rho(\bm{z}) \big]}_{(B_{\bm{z}})^{-1}}.
\end{equation*}

Let $\alpha_{\bm{z}} = \mu^{N_z}(1-\mu)^{c-N_{z}}$ be the probability that $c$ i.i.d. $\bern(\mu)$ trials result in $\bm{z}$. Hence $\expect_{Y_1, \ldots, Y_c \sim_{i.i.d.} \bern(\mu)}[\rho(\bm{Y})] = \expect_{Y_1, \ldots, Y_c \sim_{i.i.d.} \bern(\mu)}[\rho(\bm{Y})-h(N_{\bm{Y}}) + \overline{h}] = \sum_{\bm{z} \in \{0,1\}^c} \alpha_{\bm{z}} A_{\bm{z}}$. Applying \cref{lemma:revenue_recent_dynamic_formula_lemma_improving_review_offsetting_main_body} we can thus express the revenue of $\rho$ as
\begin{equation*}
\textsc{Rev}(\sigma^{\textsc{newest}}, \rho) =\frac{\sum_{\bm{z} \in \{0,1\}^c} \alpha_{\bm{z}} A_{\bm{z}}}{\sum_{\bm{z} \in \{0,1\}^c}\alpha_{\bm{z}} B_{\bm{z}}} \leq \max_{\bm{z} \in \mathcal{S}} \frac{A_{\bm{z}}}{B_{\bm{z}}}=\max_{z\in \{0,1\}^c}  \rev(\sigma^{\textsc{newest}}, \Tilde{\rho}_{\bm{z}})\leq \rev(\newest,\rho^{\star})
\end{equation*}
where the first inequality follows from the following natural convexity property (the proof can be found in Appendix \ref{appendix_proof_lemma_ratio_average_less_max}): for any $\{\tilde{\alpha}_{i},\tilde{A}_i,\tilde{B}_i\}_{i \in \mathcal{S}}$ with $\tilde{\alpha}_i,\tilde{B}_i >0 $ for all $i$,  $\frac{\sum_{i\in \mathcal{S}} \tilde{\alpha}_i \tilde{A}_i}{\sum_{i \in \mathcal{S}} \tilde{\alpha}_i \tilde{B}_i } \leq \max\limits_{\substack{i \in \mathcal{S}}}  \frac{\tilde{A}_i}{\tilde{B}_i}$.  
\end{proof}

\paragraph{Intuition behind the optimal dynamic pricing.}

Intuitively, the term $p^{\star}$ in \cref{theorem:newest_first_dynamic_optimal_policy} is the optimal price when facing a \textit{single} customer with a random valuation $\Theta + \expect_{N \sim \binomial(c, \mu)}[h(N)]$ for $\Theta \sim \mathcal{F}$. The selected offset makes the purchase probability equal to the purchase probability under the ``single customer'' setting with the optimal price~$p^{\star}$. This intuition enables us to characterize the optimal revenue of dynamic policies. 
To do this, a useful quantity in our characterization results is the optimal revenue for a given valuation distribution represented by a random variable $V$, i.e., $r^{\star}(V) \coloneqq \max_{p \in \mathbb{R}}p\prob_{V}[V \geq p]$. For the distributions we consider, $V$ has bounded support, and hence $r^{\star}(V) < \infty$. The proof of the following corollary can be found in Appendix~\ref{appendix:characterization_optimal_dynamic_pricing_rev_newest}.

\begin{corollary}\label{thm:most_recent_dynamic_opt_rev} The revenue of the optimal dynamic pricing policy under $\sigma^{\textsc{newest}}$ equals the optimal revenue from selling to a single customer with valuation $\Theta + \expect_{N \sim \binomial(c, \mu)}[h(N)]$. That is,
    $$\max_{\rho \in \Pi^{\textsc{dynamic}}} \rev(\sigma^{\textsc{newest}},\rho) = r^{\star}\big(\Theta + \expect_{N \sim \binomial(c, \mu)}[h(N)]\big).$$
\end{corollary}

\begin{remark}\label{rmk_opt_dynamic_generalization}
In Appendix~\ref{appendix_generalization_newest_first_all_optimal_dynamic_pricing_policies}, we also characterize the \textit{complete} set of optimal policies for $\sigma^{\textsc{newest}}$; the stated policy is the \emph{unique} optimal policy under mild regularity conditions (Appendix~\ref{appendix: unique_optimal_dynamic_pricing}). 
\end{remark}

\paragraph{Comparison between optimal dynamic pricing for Newest and Random.}
The following proposition (proof in Appendix \ref{app:proof_random_optimal_pricing}) derives the optimal dynamic pricing policy for $\sigma^{\textsc{random}}$. 
Specifically, for every state of reviews, the optimal price is the revenue-maximizing price for that state.

\begin{proposition}\label{theorem_char_dynamic_random}
   For every review state $\bm{z} \in \{0,1\}^c$, any optimal dynamic pricing policy under $\sigma^{\textsc{random}}$ sets $\rho^{\textsc{random}}(\bm{z}) \in \argmax_{p}p \prob_{\Theta \sim \mathcal{F}}\big[\Theta + h(N_{\bm{z}}) \geq p\big] $. This implies that $$\rev(\sigma^{\textsc{random}}, \Pi^{\textsc{dynamic}})  = \expect_{N \sim \binomial(c,\mu)}\Big[r^{\star}(\Theta + h(N)) \Big].$$
\end{proposition}

We next compare the optimal dynamic pricing policies under $\newest$ and $\random$. The following proposition (proof in \cref{appendix:comparison_opt_dynamic_newest_random}) shows that under mild regularity conditions, states $\bm{z}$ with $h(N_{\bm{z}}) > \overline{h}$ result in higher price under $\newest$, while states $\bm{z}$ which $h(N_{\bm{z}}) < \overline{h}$ result in lower price under $\newest$. Intuitively, $\rho^{\textsc{newest}}$ charges higher prices in review states $\bm{z}$ with ``high" ratings and lower prices in review states $\bm{z}$ with ``low" ratings compared to $\rho^{\textsc{random}}$ in order to induce the same purchase probability in every review state. 

\begin{proposition}\label{prop: comparison_prices_dynamic_newest_random_main_body} Let $\overline{h} =  \expect_{N \sim \binomial(c, \mu)}[h(N)]$. Under mild regularity conditions (see Appendix \ref{appendix:comparison_opt_dynamic_newest_random}), the unique dynamic pricing policies $\rho^{\textsc{newest}}$ and $\rho^{\textsc{random}}$ satisfy:
    \begin{itemize}
        \item $\rho^{\textsc{newest}}(\bm{z}) \geq \rho^{\textsc{random}}(\bm{z})$ for review states $\bm{z} \in \{0,1\}^c$ with  $h(N_{\bm{z}}) > \overline{h}$,
        \item $\rho^{\textsc{newest}}(\bm{z}) \leq \rho^{\textsc{random}}(\bm{z})$ for review states $\bm{z} \in \{0,1\}^c$ with $h(N_{\bm{z}}) < \overline{h}$,
        \item $\rho^{\textsc{newest}}(\bm{z}) = \rho^{\textsc{random}}(\bm{z})$ review states $\bm{z} \in \{0,1\}^c$ with $h(N_{\bm{z}}) = \overline{h}$.
    \end{itemize}
\end{proposition}

\subsection{Cost of Newest First is Bounded under Dynamic Pricing (Theorem \ref{theorem:dynamic_pricing_CoNF_bound})}\label{subsec:proof_conf_bound_dynamic_pricing}

We now prove Theorem \ref{theorem:dynamic_pricing_CoNF_bound}, leveraging the results of Section \ref{subsec: optimize_dynamic_price} that characterize the optimal dynamic pricing policies.
For convenience, we denote $\expect_{N \sim \binomial(c,\mu)}$ by $\expect_{N}$ and $\overline{h} = \expect_{N \sim \binomial(c, \mu)}[h(N)]$. By  \cref{theorem_char_dynamic_random} and \cref{thm:most_recent_dynamic_opt_rev}, we can express the \textsc{Cost of Newest First} as:
\begin{equation}\label{eq:move_inside_expectation}
\chi(\Pi^{\textsc{dynamic}}) = \frac{\expect_{N \sim \binomial(c,\mu)}[r^{\star}\big(\Theta +h(N)\big)]}{r^{\star}(\Theta + \overline{h})} = \expect_{N \sim \binomial(c, \mu)}\Big[\frac{r^{\star}(\Theta + h(N))}{r^{\star}(\Theta + \overline{h})} \Big],
\end{equation}
where the denominator does not depend on $N$ and can thus move inside the expectation. We now focus on the quantity inside the expectation for a particular realization of $N$.
For any price $p_n > 0$,
\begin{align}
\frac{r^{\star}(\Theta + h(n))}{r^{\star}(\Theta + \overline{h})} &= \frac{p^{\star}(\Theta + h(n)) \cdot \prob_{\Theta \sim \mathcal{F}}
[\Theta + h(n) \geq p^{\star}(\Theta + h(n))]}{p^{\star}(\Theta + \overline{h}) \cdot \prob_{\Theta \sim \mathcal{F}}
[\Theta + \overline{h} \geq p^{\star}(\Theta + \overline{h})]} \nonumber \\
&\leq \underbrace{\frac{p^{\star}(\Theta + h(n))}{{p}_n}}_{\textsc{Price Ratio}(p_{n},n)} \cdot \underbrace{\frac{\prob_{\Theta \sim \mathcal{F}}
[\Theta + h(n) \geq p^{\star}(\Theta + h(n))]}{\prob_{\Theta \sim \mathcal{F}}
[\Theta + \overline{h} \geq {p}_n]}}_{\textsc{Demand Ratio}(p_n,n)}. \label{ineq: price_ratio_times_demand_ratio_upp_bd}
\end{align}
The inequality replaces the revenue-maximizing price $p^{\star}(\Theta+\overline{h})$ by another price $p_n$, which can only increase the ratio. Given that we operate with dynamic prices, we are allowed to select a different price for any number of positive reviews $N$. 

In particular, a price of $\overline{h} +p^{\star}(\Theta + h(n))-h(n)$ ensures that the demand ratio is one. However, if $p^{\star}(\Theta + h(n))< h(n)$, the denominator in the price ratio can be unboundedly small. To simultaneously bound the expected price and demand ratios, we select
$\Tilde{p}_n = \overline{h} + \max(p^{\star}(\Theta + h(n))-h(n), 0)$.

\begin{lemma}\label{lem:price_ratio}
The expected price ratio is at most $\expect\Big[\textsc{Price Ratio}(\Tilde{p}_{N},N)\Big] \leq 2$.
\end{lemma}

\begin{lemma}\label{lem:demand_ratio}
For any $n \in \{0,1, \ldots, c\}$, the demand ratio is $\textsc{Demand Ratio}(\Tilde{p}_n,n) \leq \frac{1}{\prob_{\Theta \sim \mathcal{F}}[\Theta \geq 0]}$.
\end{lemma}

\begin{proof}[Proof of Theorem~\ref{theorem:dynamic_pricing_CoNF_bound}.]
The proof directly combines (\ref{eq:move_inside_expectation}), (\ref{ineq: price_ratio_times_demand_ratio_upp_bd}), and Lemmas \ref{lem:price_ratio} and \ref{lem:demand_ratio}.
\end{proof}

What is left is to prove the lemmas that bound the expected price ratio and the demand ratio. 

\begin{proof}[Proof of \cref{lem:price_ratio}.]
Given that customer valuations are additive with a belief and an idiosyncratic component, the optimal price (numerator of price ratio) can be similarly decomposed as:
$$p^{\star}(\Theta + h(n)) = \underbrace{h(n)}_{\text{Belief about $\mu$}} + \underbrace{p^{\star}(\Theta + h(n))-h(n)}_{\text{Idiosyncratic valuation}}$$
The expected price ratio $\expect\Big[\textsc{Price Ratio}(\Tilde{p}_{N},N)\Big]$ for $\Tilde{p}_n = \overline{h} + \max(p^{\star}(\Theta + h(n))-h(n), 0)$ is:
\begin{align*}
\expect_{N}\Bigg[\frac{h(N)}{\overline{h} +  \max(p^{\star}(\Theta + h(N))-h(N),0) } \Bigg] + \expect_{N}\Bigg[\frac{p^{\star}(\Theta + h(N))-h(N)}{\overline{h} +  \max(p^{\star}(\Theta + h(N))-h(N),0) } \Bigg].
\end{align*}
Given that the denominators in both terms are positive and $\overline{h} = \expect_{N}[h(N)]$, each of those terms can be upper bounded by $1$, concluding the proof.
\end{proof}

\begin{proof}[Proof of \cref{lem:demand_ratio}.]
For every number of positive reviews $n$, we distinguish two cases based on where the maximum in $\tilde{p}_n$ lies. If $p^{\star}(\Theta + h(n)) \geq h(n)$ the demand ratio is equal to $1$. Otherwise, 
$$\textsc{Demand Ratio}(\Tilde{p}_n,n) = \frac{\prob_{\Theta \sim \mathcal{F}}[\Theta \geq p^{\star}(\Theta + h(n))-h(n) ]}{ \prob_{\Theta \sim \mathcal{F}}[\Theta \geq 0] } \leq \frac{1}{\prob_{\Theta \sim \mathcal{F}}[\Theta \geq 0]}. $$
 \end{proof}

\subsection{Broader Implication: Cost of Ignoring State-dependent Customer Behavior}

Suppose a platform uses $\sigma^{\textsc{newest}}$, and they are \emph{not aware} of the phenomenon that customer's purchase decisions depend on the state of 
reviews --- they instead assume the purchase behavior is constant (i.e., purchase probability does not depend on the state of the newest $c$ reviews). Then, if the platform uses a standard data-driven approach to optimize prices (e.g., do price experimentation and estimate demand from data), the optimal revenue 
is $\rev(\sigma^{\textsc{newest}},\Pi^{\textsc{static}})$. In contrast, a platform can 
estimate separate demands for each state of reviews and employ a dynamic pricing policy to earn $\rev(\sigma^{\textsc{newest}},\Pi^{\textsc{dynamic}})$. By comparing $\rev(\sigma^{\textsc{newest}},\Pi^{\textsc{dynamic}})$ with $\rev(\sigma^{\textsc{newest}},\Pi^{\textsc{static}})$, we show that the revenue loss from not accounting for this state-dependent behavior can be arbitrarily large. The next result (proof in Appendix~\ref{appendix:what_if_platforms_are_not_aware}) follows from Theorems \ref{thm:CoNF_opt_static_arbitrarily_bad} and  \ref{theorem:dynamic_pricing_CoNF_bound}.

\begin{corollary}
 \label{thm:dynamic_vs_static}
For any $M > 0$, there exists an instance such that $\frac{\rev(\sigma^{\textsc{newest}},\Pi^{\textsc{dynamic}})}{\rev(\sigma^{\textsc{newest}},\Pi^{\textsc{static}})} > M$.
\end{corollary} 

\begin{remark}
    Our work assumes that customers are myopic and not strategic (e.g., they do not write negative reviews to drive the price down). This assumption is most appropriate in settings where customers purchase the product once or very infrequently, so that they are unlikely to learn or exploit the platform's pricing rule over repeated interactions. Examples include vacation rentals or hotels for occasional travelers, one-time local services, durable goods such as appliances or furniture, and products purchased for a particular event or need.    The myopic-customer assumption is  canonical in the dynamic pricing literature. We note that several papers tackle non-myopic customers \cite{levin2010optimal, besbes2015intertemporal,papanastasiou2017dynamic, chen2018robust,haghtalab2022learning} but we view this challenge as orthogonal to the main focus of our work and thus we do not investigate further how to handle non-myopic agents.
\end{remark}

\section{Cost of Newest First in the Presence of Non-Stationarity}
\label{sec:dynamic_quality}
As discussed in the introduction, customers prefer to read more recent reviews; this explains the practical popularity of $\newest$. 
One potential reason for the popularity of Newest First is that customers prefer to read newer reviews because of their belief that newer reviews are more representative of the current product quality compared to older reviews. This might lead one to believe that Newest First is better than Random when the product quality $\mu$ changes over time. To study this phenomenon, we focus on a simple model where just one review is displayed $(c=1)$ and the product quality $\mu^{(t)}$ evolves according to a Markov Chain with a high value $\mu_H$ and a low value $\mu_L$. The transition dynamics of the product quality Markov Chain $\mu^{(t)}$ are such that from any state with probability $\xi \in (0,1]$ the Markov Chain $\mu^{(t)}$ transitions to a new state which is equally likely to be $\mu_H$ or $\mu_L$ and otherwise it remains at the current state. The customer's purchase behavior given reviews is the same as in Section \ref{subsec:model_customer_purchase_behavior}. Let $q_0 = \prob_{\Theta \sim \mathcal{F}}[\Theta + h(0) \geq p]$ and $q_1 = \prob_{\Theta \sim \mathcal{F}}[\Theta + h(1) \geq p]$ be the purchase probabilities for a negative and a positive review respectively. 

Our first result in this setting shows that, for any qualities $\mu_L$ and $\mu_H$, the Cost of Newest First continues to arise, i.e., the revenue under the Newest First ordering policy is strictly worse than the revenue under a random ordering policy $\random$. To define the latter, notice that, given that the switching probability $\xi$ is independent of the product quality, the steady-state distribution for the quality is $\mu_H$ with probability $1/2$ and $\mu_L$ with probability $1/2$. As a result, we define $\random$ as the review from a random time period, which is thus drawn from the distribution $Bernoulli(\frac{\mu_H + \mu_L}{2})$.  

\begin{proposition}\label{prop: random_better_than_newest_nonstationarity}
    For any probability of change $\xi \in (0,1]$, product qualities $\mu_L$ and $\mu_H$ with $\mu_L < \mu_H$, and any price $p$ satisfying Assumption \ref{assumption:non_abs_non_degen}, the revenue of $\newest$ is strictly smaller than the revenue of $\random$, i.e, $\rev(\random, p) > \rev(\newest, p)$. 
\end{proposition}

\begin{figure}[H]
    \centering
    \captionsetup{font=small}
    \begin{subfigure}{0.45\textwidth}
        \centering
        \includegraphics[width=\linewidth]{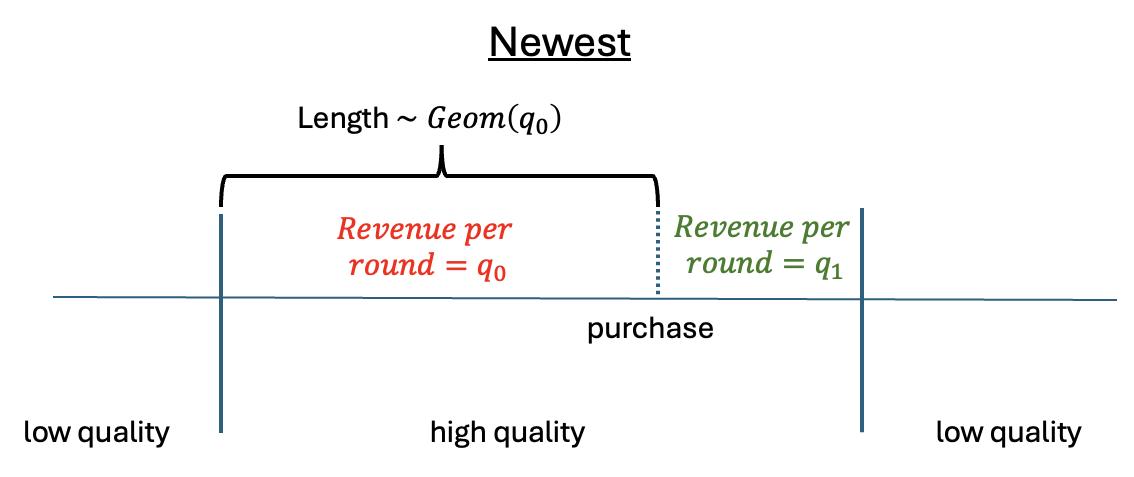}
        \caption{Newest review is negative and high quality}
        \label{fig:high_quality_epoch}
    \end{subfigure}
    \hfill
    \begin{subfigure}{0.45\textwidth}
        \centering
        \includegraphics[width=\linewidth]{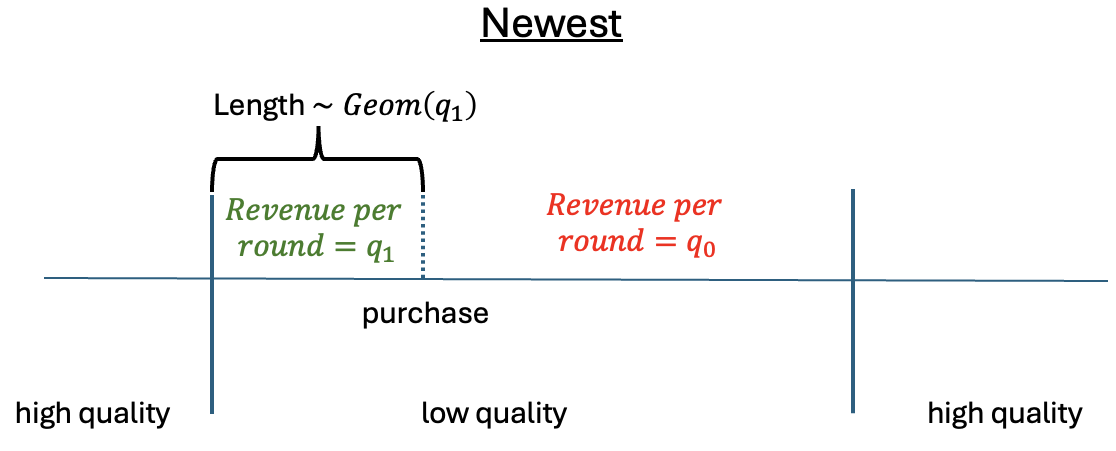}
        \caption{Newest review is positive and low quality}
        \label{fig:low_quality_epoch}
    \end{subfigure}
    \vskip\baselineskip
    \caption{When the newest review is negative and the state transitions from low to high, it takes $\textit{Geometric}(q_0)$ rounds to obtain a new review (Figure \ref{fig:non-stationarity-revenue-intuition} (a)). On the other hand, when the newest review is positive and the state transitions from high to low it takes $\textit{Geometric}(q_1)$ rounds to obtain a new review (Figure \ref{fig:non-stationarity-revenue-intuition} (b)).}  
    \label{fig:non-stationarity-revenue-intuition}
\end{figure}

\begin{proof}[Proof sketch.] For simplicity, we here assume that $\mu_L=0$ and $\mu_H=1$; the full proof in Appendix~\ref{appendix_subsec:proof_prop_random_better_than_newest_nonstationarity} proves the theorem for general product qualities. The crux of the proof is that, when transitioning from a low to high (high to low) state, $\newest$ needs more (less) time than $\random$ to obtain a purchase and thus update the newest review; see Figure~\ref{fig:non-stationarity-revenue-intuition} for an illustration. For periods in which the newest review has the same quality as the environment, $\newest$ and $\random$ operate identically if one couples the qualities for the two processes.
\end{proof}

We next formalize the idea that Newest First is better than Random in representing the product quality. To do this, for any review ordering policy $\sigma$, pricing policy $\rho$, and probability of change $\xi \in (0,1]$, we define the belief error as a measure of difference between the customer's posterior and the product quality at the current round in steady state.
$$\textsc{BeliefError}(\sigma, \rho; \xi) \coloneqq \liminf_{T \to \infty} \expect \Big[\frac{\sum_{t=1}^{T} (\expect[\Phi_t]-\mu_t)^2}{T} \Big].$$

The following proposition establishes that, when the posterior after a positive (respectively, negative) review matches $\mu_H$ (respectively, $\mu_L$), $\newest$ induces lower belief error than $\random$ for any probability of change $\xi\in(0,1)$. This explains the intuitive appeal of Newest First.

\begin{proposition}\label{prop: nonstationarity-belieferror-newest-better-than-random}
Let $\mu_L = \expect[\mathrm{Beta}(a,b+1)]$ and $\mu_H = \expect[\mathrm{Beta}(a+1,b)]$ and price $p$ satisfying Assumption \ref{assumption:non_abs_non_degen}. For any $\xi \in (0,1)$, $\textsc{BeliefError}(\newest,p; \xi) < \textsc{BeliefError}(\random,p; \xi)$.
Moreover, $\textsc{BeliefError}(\sigma, p; \xi)$ increases with $\xi$ for $\sigma=\newest$ and is constant in~$\xi$ for $\sigma=\random$. 
\end{proposition}

Intuitively, more frequent quality switches (larger $\xi$) implies that the newest review comes from the previous product quality (similar to Figure \ref{fig:non-stationarity-revenue-intuition}) with higher probability in steady state. The formal proof is provided in Appendix \ref{appendix_subsec:non-stationarity-belief-error}.

\section{Robustness of Cost of Newest First via Numerical Simulations} \label{sec:numerics}
We conduct numerical simulations on synthetic data to validate and extend our main results under variants of our theoretical model. Specifically, we examine the sensitivity of the CoNF to the limited-attention parameter~$c$ (Section~\ref{subsec:numerics_limited_attention}), its robustness under a time-varying prior (Section~\ref{subsec:numerics_time_varying_prior}), its persistence under increasing product quality (Section~\ref{subsec:numerics_nonstationarity}) and self-selection bias (Section~\ref{subsec:conf_with_self_selection_bias}), and the effectiveness of dynamic pricing in mitigating it (Section~\ref{subsec:numerics_dynamic_vs_static_conf}).

\subsection{Impact of the limited-attention parameter }\label{subsec:numerics_limited_attention}
We evaluate how the CoNF varies with the limited attention parameter $c$. We consider instances with $c \in \{1, 2, \ldots, 50\}$, $\mu \in \{0.1, 0.5\}$, $\mathcal{F} = \mathcal{U}[-1,1]$, prior $(a,b) = (\mu, 1-\mu)$, $h$ is the posterior mean, and $p = 1$. Figure~\ref{fig:limited_attention_parameter_vs_revenue} shows the revenue under $\newest$ and $\random$ as a function of $c$.

\begin{figure}[!htbp]
    \centering
    \includegraphics[width=0.9\textwidth]{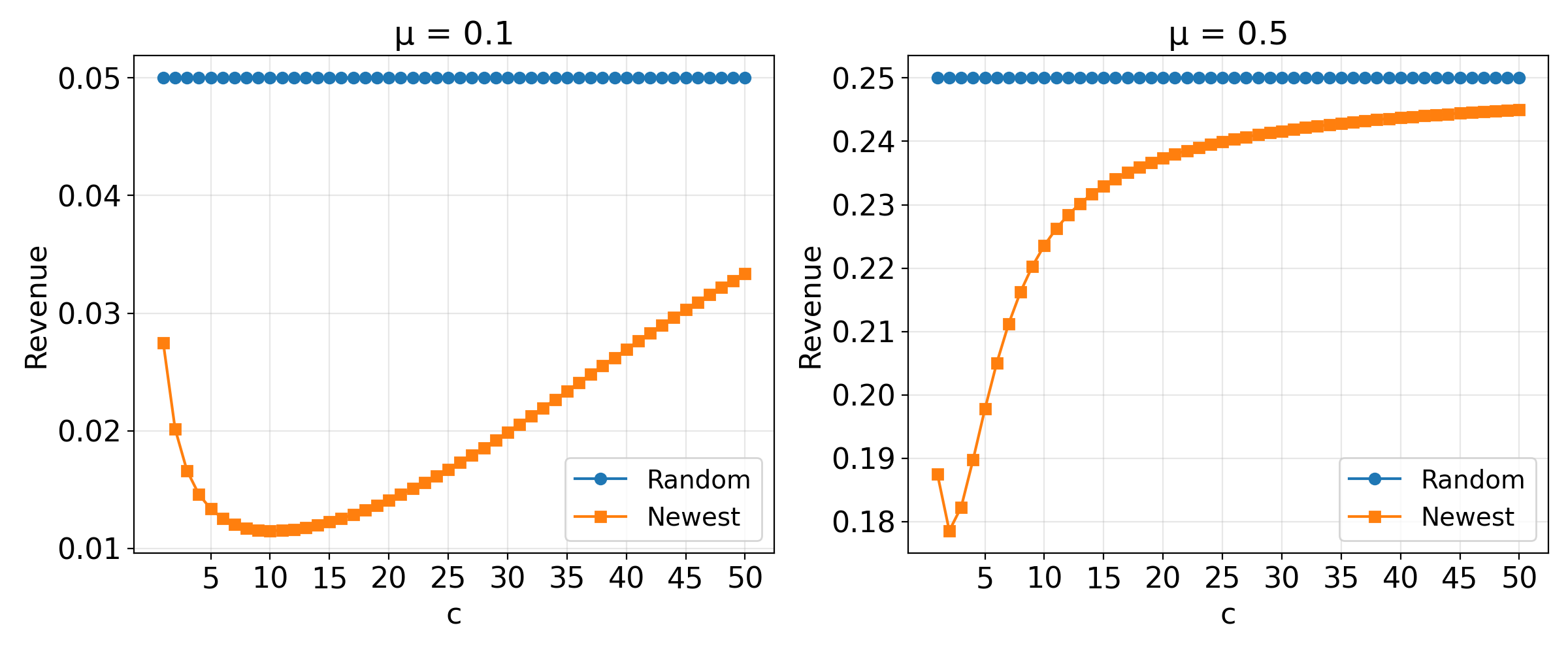}
    \caption{Revenue under $\newest$ and $\random$ as a function of the limited attention parameter $c$, for $\mu \in \{0.1, 0.5\}$. Revenue under $\random$ is constant in $c$, while revenue under $\newest$ is non-monotone: it first decreases and then increases, converging to $\random$ as $c \to \infty$.
    }
    \label{fig:limited_attention_parameter_vs_revenue}
\end{figure}

The revenue under $\random$ is constant in $c$, because by Proposition~\ref{theorem:random_C_revenue}, it equals $$p \expect_{N \sim \binomial(c, \mu)}\left[\prob_{\Theta \sim \mathcal{F}}\left[\Theta + \frac{\mu + N}{1 + c} \geq p \right] \right] = \expect_{N \sim \binomial(c, \mu)}\left[ \frac{\mu + N}{2(1 + c)}\right] = \frac{\mu (1+ c)}{2(1 + c)} = \frac{\mu}{2}$$
where the first equality uses that $\mathcal{F} = \mathcal{U}[-1,1]$, $p = 1$ and the second equality uses linearity of expectation.
The revenue under $\newest$ is non-monotone: it first decreases and then increases with $c$, converging to $\random$ as $c \to \infty$.
The convergence as $c \to \infty$ follows from the theory (\cref{prop:more_broadly_conf_to_one_without_bd_rationality},
Appendix~\ref{app:driver_generic_updating}): with sufficiently many reviews, the posterior concentrates near $\mu$ regardless of the ordering policy.
The initial decrease is more subtle.
When $c = 1$, the worst-case state is a single negative review, from which a single purchase is enough to recover from.
When $c = 10$, the worst-case state is ten consecutive negative reviews: in this state, the purchase probability becomes very small, and multiple purchases are needed to escape, so the system can remain stuck in this bad state for much longer.
This is why the CoNF can worsen as $c$ grows from small values.
For very large $c$, however, concentration ensures that an all-negative state is essentially unreachable, and the CoNF diminishes.

\subsection{Time-varying prior with cold start} \label{subsec:numerics_time_varying_prior}
Recall that our baseline model assumes a fixed prior $\mathrm{Beta}(a, b)$, which reflects aggregate information such as the average rating (e.g., 4.5/5). This assumption is appropriate when the average rating remains stable over time, as might be the case when there are many reviews. However, in reality, the number of available reviews is finite, and the average rating displayed to customers evolves endogenously as new reviews are added. To assess how robust our main results are to this consideration, we simulate an alternative model in which the prior is updated dynamically based on the evolving pool of reviews, beginning from a cold start with no reviews.

We maintain a review pool initialized with a $c$ reviews drawn from $\bern(\mu)$ at $t = 1$. At time $t$, let $P_t$ and $N_t$ denote the total positive and negative reviews in the pool.
The prior is set to $\mathrm{Beta}(a + \gamma \cdot P_t,\; b + \gamma \cdot N_t)$, where $\gamma \geq 0$ controls how strongly the aggregate history enters the prior; $\gamma = 0$ recovers the original model.
In this formulation, the influence of the prior grows as more reviews accumulate, which implies that a large history dominates any signal from a small number of displayed reviews.  Note that each of the $c$ reviews that are read effectively carries a weight of $1 + \gamma$ in the customer's posterior. We consider small values of $\gamma$ to reflect the fact that a review in the pool carries a small weight compared to a review that is read.

We consider three policies: $\newest$, $\randomfinite$ (draws $c$ reviews at random from the pool), and $\random$ (draws $c$ reviews i.i.d.\ from $\bern(\mu)$). We run these policies for a set of instances where $\mu \in \{0.1, 0.5\}$, $c = 1$, prior $\mathrm{Beta}(\mu, 1-\mu)$, $\gamma \in \{0, 0.01, 0.1\}$, $\mathcal{F} = \mathcal{U}[0,1]$, $h$ is the posterior mean, and $p = 1$, over 10,000 rounds averaged over 100,000 runs (Figure~\ref{fig:endogeneous_prior}). 

\begin{figure}[!htbp]
    \centering
    \includegraphics[width=0.9\textwidth]{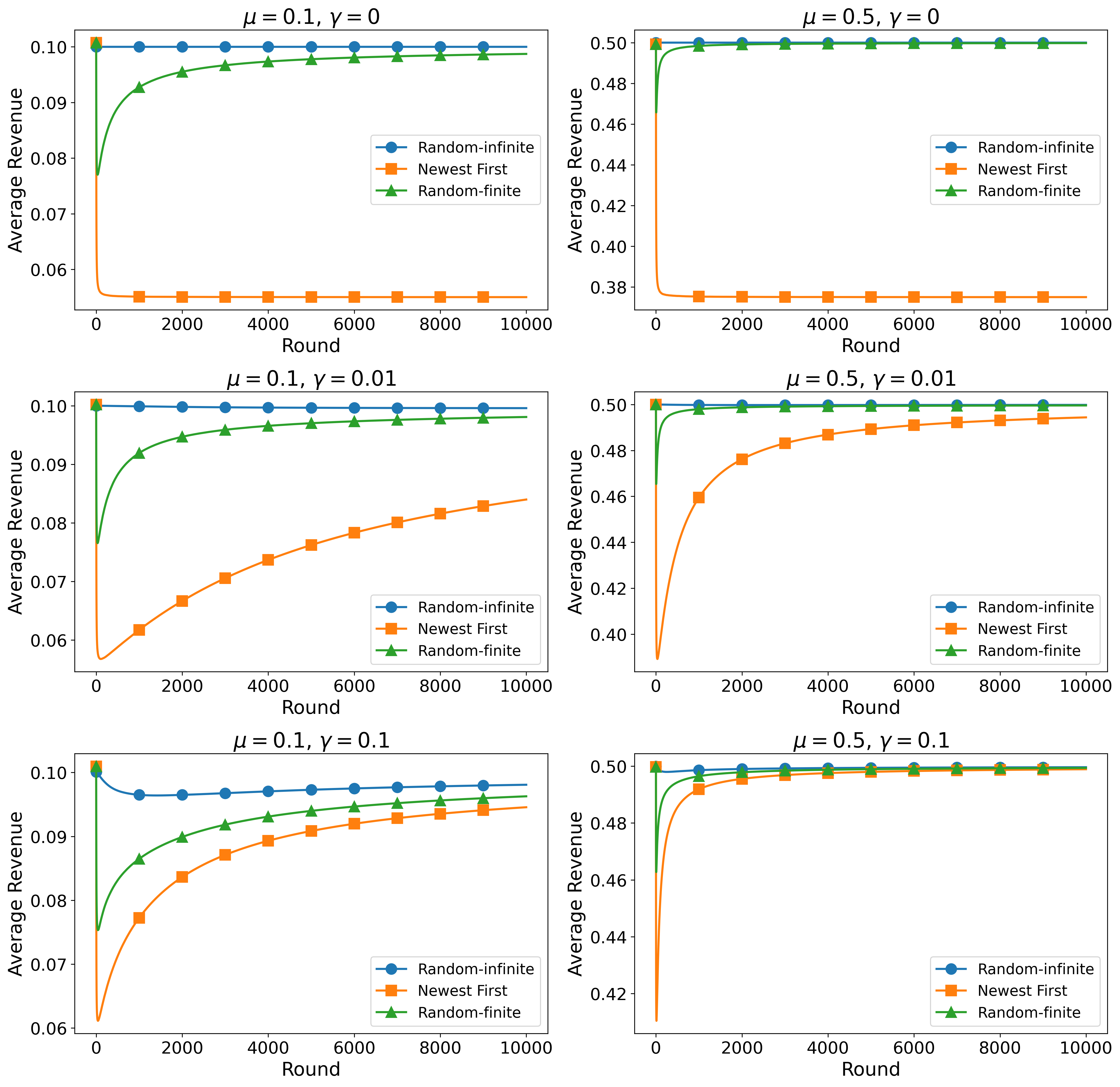}
    \caption{Average revenue per round for $\newest$, $\randomfinite$, and $\random$ under an endogenous time-varying prior $\mathrm{Beta}(a + \gamma \cdot P_t, b + \gamma \cdot N_t)$, for $\mu \in \{0.1, 0.5\}$ and $\gamma \in \{0, 0.01, 0.1\}$.
    }
    \label{fig:endogeneous_prior}
\end{figure}

Figure~\ref{fig:endogeneous_prior} confirms the robustness of our main insight.
Throughout, there is a clear and substantial gap between both random policies and $\newest$, consistent with the CoNF established in Section~\ref{sec:ordering}.
As $t$ grows and the prior strengthens, the impact of the $c$ reviews decreases, and all policies converge to the same revenue level.
A larger $\gamma$ accelerates this convergence since the prior strength is higher
and a smaller $\mu$ converges more slowly due to the slower accumulation of reviews.

We note that this setting is somewhat extreme: as $t \to \infty$, the prior $\mathrm{Beta}(a + \gamma \cdot P_t, b + \gamma \cdot N_t)$ concentrates to a point mass at $\mu$, so the product quality becomes exactly known and all policies trivially converge. In practice, however, customers continue to read individual reviews even when a large aggregate rating is available, suggesting that the informational content of the prior may be naturally bounded.  We show that even under this setting, the CoNF emerges when the prior has not fully converged. 

Next, we see that in all cases, the revenue of $\randomfinite$ quickly converges to the revenue of $\random$: once the pool is large enough, random sampling from it approximates i.i.d.\ draws. In early rounds, $\randomfinite$ falls below $\random$ because a small pool may be dominated by negative reviews, echoing the same self-reinforcing mechanism as CoNF; this effect is more pronounced for $\mu = 0.1$, where fewer purchases slow pool growth.

In summary, CoNF is not an artifact of the fixed-prior assumption: it emerges precisely when the prior is weak and individual reviews are most influential.

\subsection{Non-stationarity with increasing product quality}\label{subsec:numerics_nonstationarity}
Section~\ref{sec:dynamic_quality} establishes that CoNF persists when product quality changes according to a Markov Chain. 
Here, we consider a more extreme form of non-stationarity where product quality increases monotonically over time.
This setting is the most favorable possible for $\newest$, since newer reviews are systematically drawn from a higher-quality distribution and carry strictly more information about the current product than older reviews.
Despite this, we show that CoNF can still arise.

For a horizon of $T$ rounds, the product quality increases linearly from $\mu^L$ to $\mu^H > \mu^L$: specifically, $\mu_t = \mu^L + \frac{t-1}{T-1}(\mu^H - \mu^L)$.
If a customer purchases at round $t$, their review is drawn from $\bern(\mu_t)$.
Because newer reviews have higher mean quality, the optimal ordering policy should prioritize recent reviews, but may still benefit from some degree of randomization among them due to the effect of CoNF.
To explore this trade-off, we simulate a family of \emph{window-random} policies $\sigma^{\textsc{random}(w)}$, which selects $c$ reviews uniformly at random from the $w$ most recent reviews in the pool.
If $w$ exceeds the total number of reviews collected so far, then the window is capped by this number. 
When $w = c$, $\sigma^{\textsc{random}(w)}$ reduces to $\newest$; as $w \to \infty$, it approaches $\random$.
Note that, unlike the stationary setting where $\random$ is optimal among rating-agnostic policies (Appendix~\ref{appendix:why_random_right_benchmark}), this property no longer holds here: since newer reviews reflect a higher product quality, an intermediate window $w > c$ may outperform both $\newest$ and $\random$. We simulate the following instance: $c = 2$, $\mathcal{F} = \mathcal{U}[0,1]$, prior $(a,b) = (0.1, 0.9)$, $h$ is the posterior mean, $\mu^L = 0.1$, $\mu^H = 0.9$, and $T = 1000$. Note that
$w = 2$ represents $\newest$, while $w = 1000$ represents $\randomfinite$ (as it draws $c$ reviews at random from the current pool of all reviews).

Figure~\ref{fig:window_vs_revenue_price_0.75_1} shows revenue as a function of the window size $w$ for two representative prices $p \in \{0.75, 1\}$, where the leftmost point ($w = c = 2$) corresponds to $\newest$ and the rightmost ($w = 1000$) represents $\randomfinite$.
At $p = 0.75$, $\newest$ outperforms $\randomfinite$, since the most recent reviews carry higher-quality signals and the low price makes purchases easy to trigger.
At $p = 1$, however, $\randomfinite$ outperforms $\newest$ and CoNF re-emerges: at high prices, the effect of CoNF when the quality is low in the earlier rounds is more pronounced, as Newest gets stuck with no purchases for a longer duration.
In both cases, the optimal window is strictly intermediate, and even a modest increase from $w = 2$ to $w = 5$ yields a substantial revenue gain, showing that a small amount of randomization among recent reviews is sufficient to mitigate CoNF.

These results reflect two competing forces, illustrated in Figure~\ref{fig:round_vs_average_rating_shown_price_0.75_1}, which plots the average rating of the displayed reviews over time.
Non-stationarity favors $\newest$: in later rounds, $\newest$ displays higher-rated reviews than $\randomfinite$, since it tracks the improving product quality more closely.
CoNF works in the opposite direction and favors $\randomfinite$: in early rounds, when quality is low, $\newest$ gets stuck showing negative reviews, suppressing purchases and delaying the accumulation of newer, better reviews.
At $p = 0.75$, the non-stationarity effect dominates and $\newest$ yields higher revenue; at $p = 1$, the high purchase threshold amplifies CoNF and $\randomfinite$ is superior.

\begin{figure}[!htbp]
    \centering
    \includegraphics[width=0.9\textwidth]{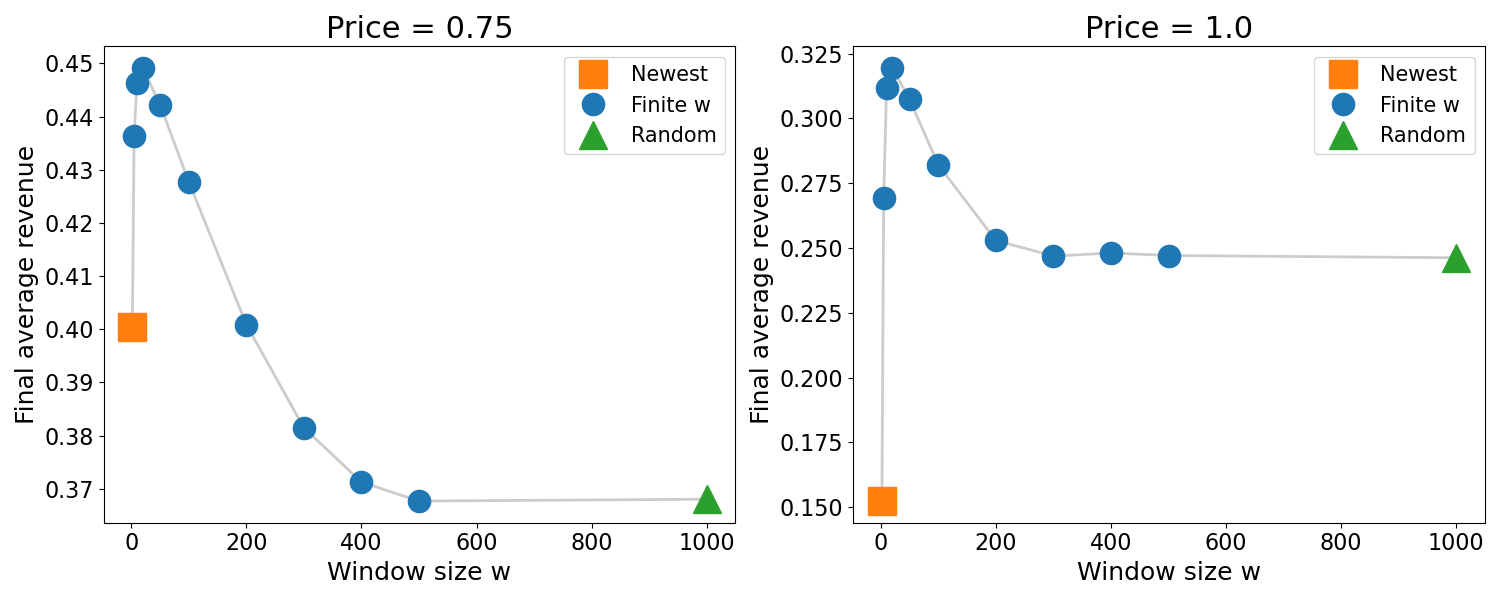}
    \caption{Average revenue at round $T=1000$ as a function of window size $w \in \{2, 5, 10, 20, 50, 100, 200,300,400, 500, 1000\}$, for $p \in \{0.75, 1\}$. In both cases, the optimal window is intermediate: neither $\newest$ ($w = 2$) nor $\randomfinite$ ($w =1000$) achieves the highest revenue.
    }
    \label{fig:window_vs_revenue_price_0.75_1}
\end{figure}

\begin{figure}[!htbp]
    \centering
    \includegraphics[width=0.9\textwidth]{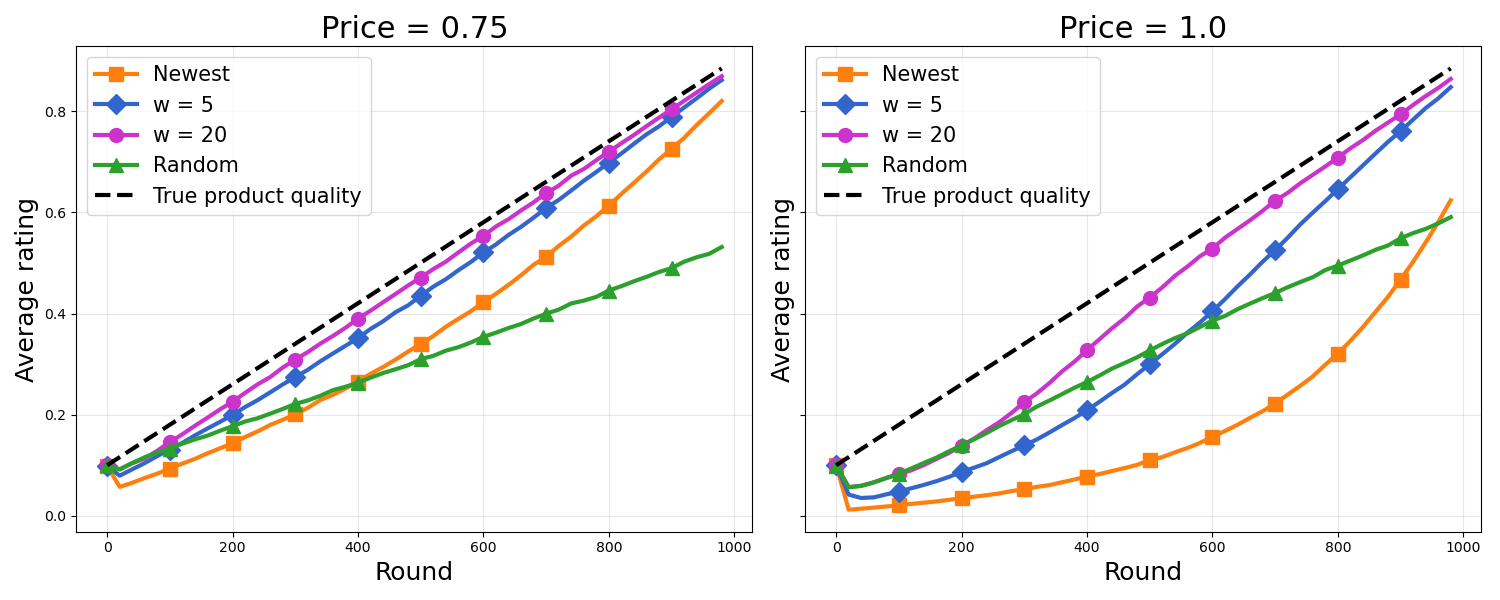}
    \caption{Average rating of the $c = 2$ displayed reviews as a function of round, for $p \in \{0.75, 1\}$. $\newest$ displays higher-rated reviews in later rounds due to increasing product quality, yet this does not prevent CoNF at $p = 1$.
    }
    \label{fig:round_vs_average_rating_shown_price_0.75_1}
\end{figure}

\subsection{CoNF in the presence of self-selection bias}\label{subsec:conf_with_self_selection_bias}
Recall that customer $t$'s realized valuation is $X_t+\Theta_t$, where $X_t$ and $\Theta_t$ represent the contribution from the product's unobservable and observable parts respectively. 
In our original model, we assume that the review reveals $X_t$, which is realized independently after a customer purchase. 
Here, we consider what happens when the review reveals $X_t+\Theta_t$. 
Note that in this case, since $\Theta_t$ is used in the purchase decision, there will be an upward ``self-selection'' bias in the reviews (as those with higher values of $\Theta_t$ are more likely to purchase and write a review). 

We assume that customers read one review ($c=1$) and when customer $t$ with idiosyncratic valuation $\Theta_t$ reads a review $R_s = \Theta_s + X_s$, they make a purchase if $\Theta_t + R_s > p$. The customer's idiosyncratic valuation is drawn from $\mathcal{F} = \mathcal{N}(0,1)$, the product quality is $\mu = 0.5$, and each $X_s$ is drawn from $\bern(\mu)$. Unlike our original model in Section~\ref{sec:model}, where customers can observe $X_s$ by reading reviews in detail, customers in this model are only able to use the entire rating when forming their belief for the product quality. As a result, we refer to this model as the \textit{coarse-ratings model}. \cite{bs18} show that under an analogous model, self-selection bias exists.  We compare this model to a baseline model where customer $t$ reading $R_s$ can observe $X_s$ and purchases if $\Theta_t + X_s \geq p$; this corresponds to our model (when $h$ is the mean and $a,b \to 0$). Figure~\ref{fig:self-selection-bias} shows the revenues of Random and Newest under the coarse-ratings model and the baseline model for prices $p \in [0,4]$. The revenue of Newest is smaller than the revenue of Random for any price under the coarse-ratings model. This constitutes evidence that CoNF holds even when customers are only able to observe $\Theta_s + X_s$. 

The gap between Newest and Random is larger under the coarse-ratings model than under the baseline model. Intuitively, this holds due to the higher downside in ratings under the coarse-ratings model as the review features an additional idiosyncratic component. The revenue of Random under the coarse-ratings model is larger than under the baseline due to effect of self-selection bias. The revenue of Newest under the coarse-ratings model is smaller than the revenue of Random under the coarse-rating model due to CoNF. 

Interestignly, for prices $p \in (0,2.3)$ the revenue of Newest under the coarse-ratings model is smaller than the revenue of Random under the baseline. As a result, the revenue loss due to CoNF outweighs the revenue gain due to self-selection bias. Intuitively, the lowest-rating states under the coarse-ratings model yield lower revenue than the lowest-rating states under the baseline due to additional downside of the idiosyncratic component. For low prices, entering these states is highly likely under the coarse-rating model and leads to lower revenue compared to the baseline. 

\begin{remark}
    Newest yields smaller revenue under the coarse-ratings model than under the baseline for prices $p \in (0,1)$ and larger revenue under the coarse-ratings model for prices $p \in (1, 4)$. The lowest-rating states under the coarse-ratings model yield lower revenue than the lowest-rating states under the baseline due to the additional downside of the idiosyncratic component. For low prices, entering these states is highly likely under the coarse-ratings model and leads to lower revenue than the baseline. For high prices, entering these states is highly unlikely because a purchase under the same low idiosyncratic valuation is needed; this leads to higher revenue than the baseline. 
\end{remark}

\begin{figure}[!htbp]
\centering
\includegraphics[width=0.8\textwidth]{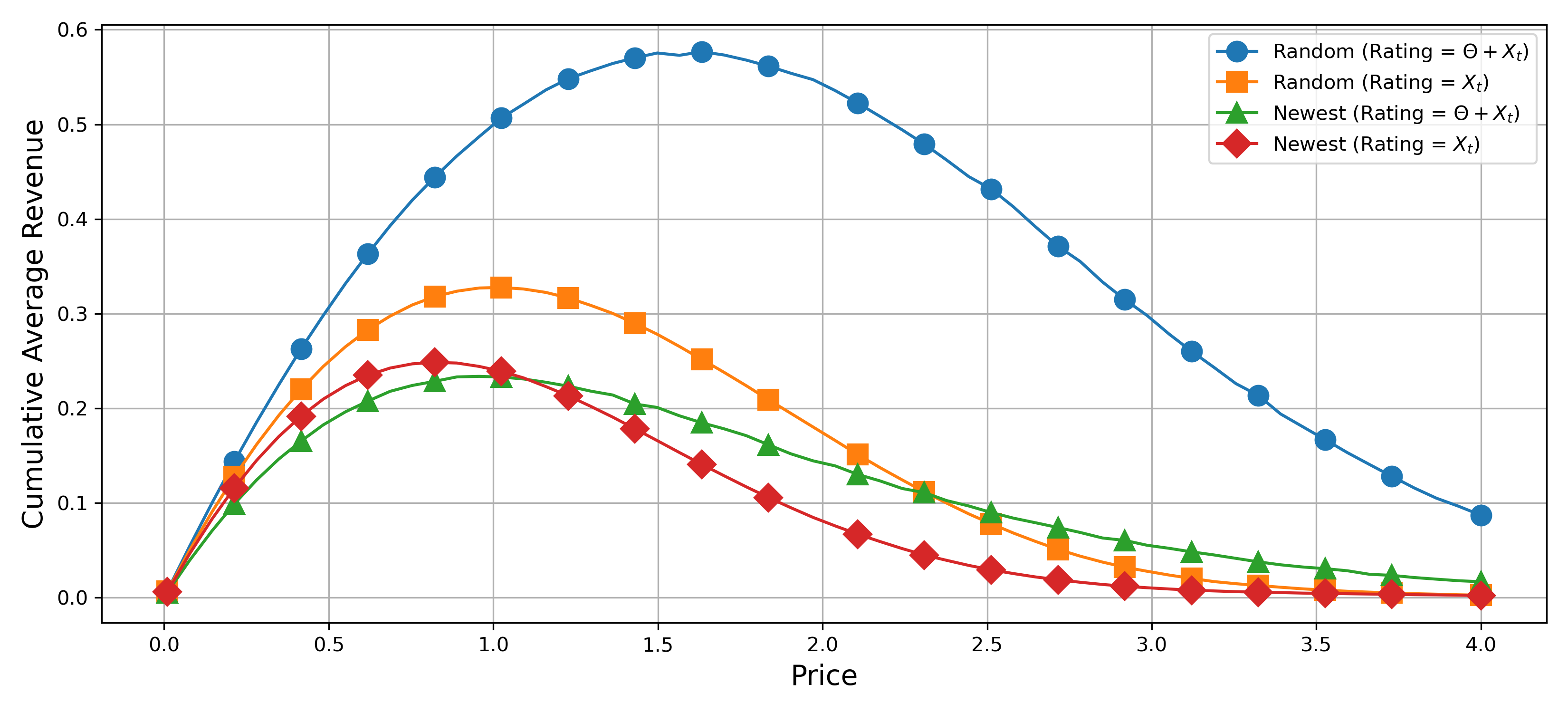}

\caption{Cost of Newest First in the persence of self-selection bias.} 
\label{fig:self-selection-bias}
\end{figure}

\subsection{Empirical validation that dynamic pricing mitigates the impact of CoNF}
\label{subsec:numerics_dynamic_vs_static_conf}
Recall that our main result of Section~\ref{sec:pricing} was that dynamic pricing mitigates the impact of CoNF. Under dynamic pricing, the worst case value for the CoNF, $\chi(\Pi^{\textsc{dynamic}})$, is at most 2 (Theorem~\ref{theorem:dynamic_pricing_CoNF_bound}), while $\chi(\Pi^{\textsc{static}})$ can be unbounded (Theorem \ref{thm:CoNF_opt_static_arbitrarily_bad}). 
Since these results are worst-case bounds, in this section, we compare the values of $\chi(\Pi^{\textsc{static}})$ and $\chi(\Pi^{\textsc{dynamic}})$ on a class of instances.

We consider instances where the customer-specific valuation is distributed as $\mathcal{F} = \mathcal{U}[-\epsilon, \epsilon]$, $c = 1$, $\mu = 0.5$,  the prior belief is $\text{Beta}(a, a)$, and $h(\cdot)$ is the mean of the posterior.
We vary $\epsilon$, which corresponds to the variability in the customer-specific valuation, and we also vary the magnitude of $a$, which represents the strength of the prior belief. 
A low prior strength, corresponding to lower values for $a$, implies that a review has a large impact on a customer's belief, whereas a higher prior strength implies that the customer's belief is minimally impacted by reviews.
We vary $\epsilon$ from 0 to 3, while we select $a \in \{0.05, 0.5, 5\}$. We plot the values of $\chi(\Pi^{\textsc{static}})$ and $\chi(\Pi^{\textsc{dynamic}})$ in Figure~\ref{fig:chi_pi_dynamic_chi_pi_static_comparison}. We plot the values of $\rev(\newest, \Pi^{\textsc{static}})$ and $\rev(\newest, \Pi^{\textsc{dynamic}})$ in Figure~\ref{fig:revenue_newest_static_dynamic}.

\begin{figure}[!htbp]
\centering
\includegraphics[width=1\textwidth]{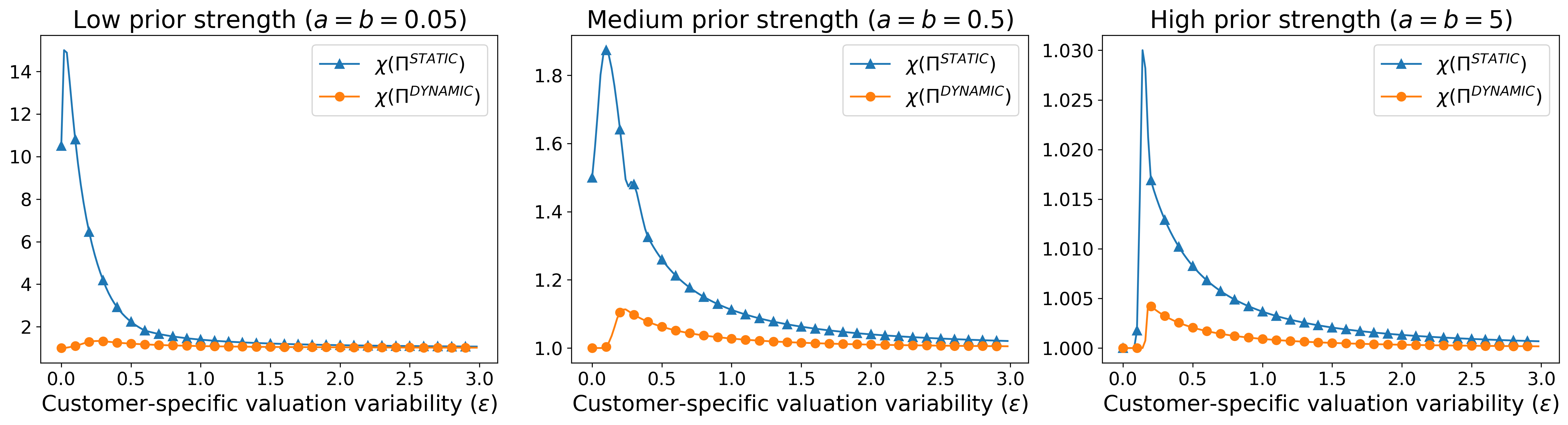}

\caption{Comparison of $\chi(\Pi^{\textsc{static}})$ and $\chi(\Pi^{\textsc{dynamic}})$.} 
\label{fig:chi_pi_dynamic_chi_pi_static_comparison}
\end{figure}

\begin{figure}[!htbp]
\centering
\includegraphics[width=1\textwidth]{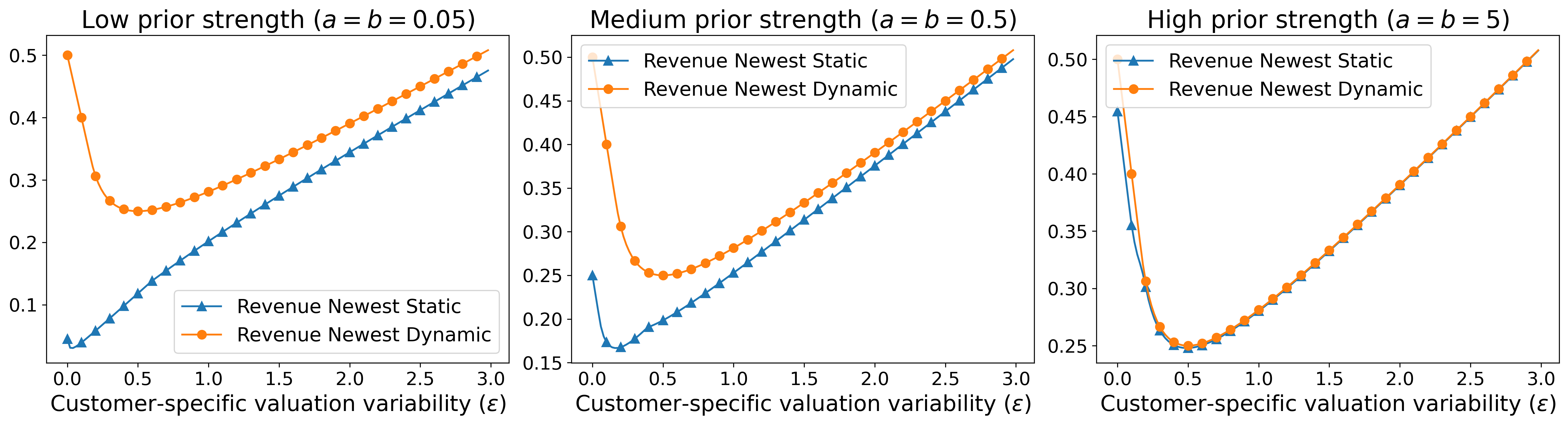}

\caption{Comparison of $\rev(\newest, \Pi^{\textsc{static}})$ and $\rev(\newest, \Pi^{\textsc{dynamic}})$.} 
\label{fig:revenue_newest_static_dynamic}
\end{figure}

We observe that the CoNF is small (less than $1.1$)  under dynamic pricing under all instances.
For static pricing, the CoNF can be large, especially when the prior strength is low and $\epsilon$ is small. 
This behavior is expected from the theory, as the instances that were constructed to show that the CoNF can be arbitrarily bad are instances where the prior strength is low (and thus the variability in review-inferred quality estimates is high) while the variability in customer-specific valuation is also low (\cref{thm:CoNF_opt_static_arbitrarily_bad}).
Importantly, we observe that $\chi(\Pi^{\textsc{dynamic}})$ is strictly smaller than $\chi(\Pi^{\textsc{static}})$ for all instances that were tested. In Figure~\ref{fig:revenue_newest_static_dynamic}, we see that the revenue of Newest under dynamic pricing can be significantly larger than the revenue of Newest under static pricing, especially when the prior strength is low and $\epsilon$ is small. The difference between Newest's revenue under dynamic and static pricing decreases as $\epsilon$ increases. These results corroborate the message that dynamic pricing mitigates the impact of the Cost of Newest First. 

\begin{remark}\label{rem:exponential}
    We note that the insights from this section are not tailored to the uniform distribution. We can recreate similar plots as in Figure~\ref{fig:chi_pi_dynamic_chi_pi_static_comparison} and Figure~\ref{fig:revenue_newest_static_dynamic} with exponentially distributed idiosyncratic valuations (see Appendix~\ref{appendix:numerics_dynamic_pricing_exponential_idiosyncratic}).
\end{remark}

\section{Conclusions}
\label{sec:conclusions}
In this paper, we model the idea that customers read only a small number of reviews before making purchase decisions. This model gives rise to the \textsc{Cost of Newest First}, the idea that, when reviews are ordered by Newest First, negative reviews will persist as the newest review longer than positive reviews.
This phenomenon does not arise in models from the existing literature, since prior works assume that customers incorporate either all reviews or a summary statistic of all reviews into their beliefs. 
We show that incorporating randomness into the review ordering or using dynamic pricing can alleviate the negative impact of the \textsc{Cost of Newest First}.

Our work opens up a number of intriguing avenues for future research. First, existing literature on social learning studies the self-selection bias (which we do not consider in our model but numerically study in Section~\ref{subsec:conf_with_self_selection_bias}) -- how does this self-selection bias broadly interact with the \textsc{Cost of Newest First}? 
Second, in terms of operational decisions, a platform contains multiple products --- should it take the state of reviews into consideration when making display or ranking decisions?
Third, given this limited attention behavior, are there alternative methods of disseminating relevant information from reviews? For example, one could succinctly summarize information from all reviews (via, e.g, generative AI) to be the most helpful for each customer. Lastly, on the theoretical side, our analysis fully characterizes the steady state of a stochastic process whose state remains unchanged with some state-dependent probability (\cref{lemma: general_theorem_markov_chains_stationary}  which is the crux in the analysis of \cref{lemma:stationary_state_distribiton_newest_first}). It would be interesting to apply this result to other settings that exhibit a similar structure.

\subsection*{Acknowledgements.} We thank the anonymous reviewers from the 25th ACM Conference on Economics and Computation (EC 2024) for their thorough feedback that greatly improved the presentation of the paper. We are also grateful to the Simons Institute for the Theory of Computing as this work started during the Fall'22 semester-long program on \emph{Data Driven Decision Processes}.


%
%
%




\bibliographystyle{alpha} 
\bibliography{References} 





  



\newpage

\appendix



\section{Generalizing beyond Beta-Bernoulli distributions (Remark~\ref{remark_model_generalization})}\label{appendix: model_generalization}
We consider a generalization of the customer behavior model in \cref{sec:model}. We only point to the modeling assumptions that change; everything else remains as in Section \ref{sec:model}:
\begin{enumerate}
    \item With respect to the customer valuation, the product's unobservable part $\mu_t$ is drawn from an arbitrary distribution $\mathcal{D}$ with mean $\mu$ and finite support $R = \{r_1, \ldots, r_s\} \in \mathbb{R}^{s}$ where $r_1 < \ldots < r_s$ and $s \geq 2$. That is, $\mathcal{D}$ is no longer restricted to  be Bernoulli with $S=\{0,1\}$ and $\mu$ in $(0,1)$. This can capture a system where $\mu_t$ is the number of stars ($S=\{1,\ldots,5\}$) which is common in online platforms such as Tripadvisor, Airbnb, and Amazon.
    \item With respect to the customer purchase behavior, when presented with a vector of $c$ reviews $\bm{Z}_t = (Z_{t,1}, \ldots, Z_{t,c}) \in S^{c}$, the customer maps them to an estimated valuation $\hat{V}_t = \Theta_t+\hat{h}(\bm{Z}_t)$, where $\hat{h}: S^{c} \to \mathbb{R}$ is an arbitrary fixed mapping. The model of \cref{subsec:model_customer_purchase_behavior} is a special case where the estimate $\hat{h}$ is created via a two-stage process: the customer initially creates a posterior belief $\Phi_t = \mathrm{Beta}\big(a + \sum_{i=1}^{c} Z_{t,i}, b + c-\sum_{i=1}^{c} Z_{t,i}\big)$ and maps this posterior to an estimate $\hat{h}(\bm{Z}_t) = h(\Phi_t)$ via a mapping $h$. Unlike this special case (where the estimate can only depend on the number of positive reviews in the $c$ displayed reviews), our generalization here allows an arbitrary mapping from $\bm{Z}_t$ that can also take the order of reviews into consideration. 
    \item With respect to the customer review generation, the review is $X_t$ in the event of a purchase and $X_t=\perp$ otherwise. The difference to Section \ref{sec:model} is that $\mu_t$ is not restricted to be Bernoulli.
    \item We assume that the estimator $\hat{h}(\bm{z})$ is strictly increasing in each coordinate of the review ratings $\bm{z}$. This extends Assumption \ref{assumption:h_monotonicity_positive_reviews} in a way that can capture the order of the reviews. 
    \item We assume that the customer's smallest estimated valuation $\underline{\hat{V}} = \Theta+\underline{h}$ where $\underline{h} = \hat{h}(r_1, \ldots, r_1)$ has nonzero mass on non-negative numbers, i.e,  $\prob_{\Theta \sim \mathcal{F}}[\underline{\hat{V}}\geq 0] > 0$. This extends Assumption \ref{assumption:non-negative_mass_idiosyncratic} and ensures that there exist non-negative prices inducing positive purchase probability. 
\end{enumerate}
The main driver behind all results in Sections \ref{sec:ordering} and \ref{sec:pricing} is the characterization of the stationary distribution of the Markov chain $\bm{Z}_t = (Z_{t,1}, \ldots, Z_{t,c})$ of the newest $c$ reviews.
\begin{itemize}
\item Estimator $\hat{h}$ generalizes the customer's purchase behavior. The purchase probability at a state with review ratings $\bm{z} \in R^c$ is $\prob_{\Theta \sim \mathcal{F}}[\Theta + h(\bm{z}) \geq \rho(\bm{z})]$ where $\rho$ is the pricing policy.  
\item The review generation process changes the transition dynamics of $\bm{Z}_t$ upon a purchase, i.e., a new review takes one of the values $\{r_1, \ldots, r_s\}$ (as opposed to $\{0,1\}$ in the original model). 
\item Letting $g_{\mathcal{D}}$ be the probability mass function of $\mathcal{D}$, the stationary distribution of $\bm{Z}_t $ is $$\pi_{\bm{z}} = \kappa \cdot \frac{\prod_{i=1}^c g_{\mathcal{D}}(r_i)}{\prob_{\Theta \sim \mathcal{F}}[\Theta + \hat{h}(\bm{z}) \geq \rho(\bm{z}) ]} \quad \text{where } \kappa = \frac{1}{\expect_{Z_1, \ldots, Z_c \sim_{i.i.d.} \mathcal{D}}\Big[\frac{1}{\prob_{\Theta \sim \mathcal{F}}[\Theta + \hat{h}(Z_1, \ldots, Z_c) \geq \rho(\bm{z}) ]}\Big]}.$$
\item This is analogous to \cref{lemma:stationary_state_distribiton_newest_first} and \cref{lemma_dynamic_stationary_distribution} for static and dynamic pricing in the the Beta-Bernoulli model. In particular, \cref{lemma_dynamic_stationary_distribution} replaces $g_{\mathcal{D}}(r_i)$ and $\hat{h}(\bm{z})$ by $\mu^{z_i}(1-\mu)^{1-z_i}$ and $h(N_{\bm{z}})$ in the expression for $\pi$, and the expectation $\expect_{Z_1, \ldots, Z_c \sim_{i.i.d.} \mathcal{D}}$ and $\hat{h}(Z_1, \ldots, Z_c)$ with $\expect_{Y_1, \ldots, Y_c \sim \bern(\mu)}$ and $h(\sum_{i=1}^c Y_i)$ in the expression for $\kappa$. The proof is completely analogous.
\end{itemize}
Using this stationary distribution and similar steps as in the proofs of  \cref{theorem:most_recent_C_revenue} and \cref{lemma:revenue_recent_dynamic_formula_lemma_improving_review_offsetting_main_body} the revenue of $\newest$ in the generalized model becomes:
$$ \textsc{Rev}(\sigma^{\textsc{newest}}, \rho) = \frac{\expect_{Z_1, \ldots, z_c \sim_{i.i.d.} \mathcal{D}}[\rho(\bm{Z})]}{\expect_{Z_1, \ldots, Z_c \sim_{i.i.d.} \mathcal{D}}\Big[\frac{1}{\prob_{\Theta \sim \mathcal{F}} [\Theta + \hat{h}(\bm{Z}) \geq \rho(\bm{Z}) ]}\Big]}.$$
The revenue of $\random$ (which shows $c$ i.i.d. reviews from $\mathcal{D}$) directly follows by adapting the purchase probability in Proposition \ref{theorem:random_C_revenue} and \cref{theorem_char_dynamic_random}:
$$\rev(\random, \rho) =  \expect_{Z_1, \ldots, Z_c \sim_{i.i.d.} \mathcal{D}}\Big[\rho(\bm{Z}) \cdot \prob_{\Theta \sim \mathcal{F}}[\Theta + \hat{h}(Z_1, \ldots, Z_c) \geq \rho(\bm{Z})] \Big].$$
Theorem \ref{theorem:negative_bias} and \ref{theorem:dynamic_pricing_CoNF_bound} then extend to our generalized model by analogously adapting their proofs. With respect to the negative results (\cref{thm:conf_arbitrarily_bad} and \cref{thm:CoNF_opt_static_arbitrarily_bad}), given that the model in the main body is a special case of the generalized model, they also directly extend.

\begin{remark}
The generalized model allows for customers who are fully Bayesian in estimating the fixed valuation given the ordering policy $\sigma$ and use an estimator mapping $\hat{h} = \hat{h}_{\textsc{FullyBayesian}}$. Such customers could be reactive to $\sigma$ and account for the effect of \textsc{CoNF}; the extension of the \textsc{CoNF} result may thus seem surprising. The reason why this occurs is that \textsc{CoNF} evaluates the customers assuming that they follow the same behavioral model under $\random$ and under $\newest$ (and does not consider the setting where customers are reactive to the ordering policy).
\end{remark}

\section{Supplementary material for Section~\ref{sec:ordering}}\label{app:ordering}
\subsection{Average rating with Newest is smaller than with Random (Proposition~\ref{thm_neg_bias_avg_review_rating})}\label{appendix:avg_newest_smaller_avg_random}
\cref{thm_neg_bias_avg_review_rating} states that the average review rating of the $c$ reviews displayed by $\newest$ is strictly smaller than the average rating of the $c$ reviews displayed by $\random$. To prove the proposition we analyze the stationary distribution of the number of positive reviews under each of $\newest$ and $\random$. \cref{lemma:stationary_state_distribiton_newest_first} implies that the stationary distribution of seeing $n$ positive reviews under $\newest$ is \begin{align*}\pi_{n}^{\textsc{newest}} = \kappa \binom{c}{n} \frac{\mu^{n}(1-\mu)^{c-n}}{\prob_{\Theta \sim \mathcal{F}}[\Theta + h(n) \geq p]}  \qquad\text{where} \qquad \kappa = 1/\expect_{N \sim \binomial(c, \mu)}\Big[\frac{1}{\prob_{\Theta \sim \mathcal{F}}[\Theta + h(N) \geq p]} \Big].\end{align*} By \cref{theorem:most_recent_C_revenue}, $\rev(\newest,p) = \kappa \cdot p$, so $\kappa$ can be interpreted as the rate at which customers purchase under $\newest$. In contrast,  the corresponding stationary distribution under $\random$ is \begin{align*}\pi_{n}^{\textsc{random}}= \binom{c}{n} \mu^n (1-\mu)^{c-n}.\end{align*} 

To prove the proposition, we first compare the behavior of $\newest$ and $\random$ based on $n^{\star} = \max\{ n| \prob_{\Theta \sim \mathcal{F}}[\Theta + h(n) \geq p] \leq \kappa \}$ (i.e., the largest number of positive review ratings where the purchase probability is at most the average purchase rate $\kappa = 1/\expect_{N \sim \binomial(c, \mu)}\Big[\frac{1}{\prob_{\Theta \sim \mathcal{F}}[\Theta + h(N) \geq p]} \Big]$). 

\begin{lemma}\label{lemma:comparison_dist_newest_random_general}
For any price $p$ satisfying Assumption \ref{assumption:non_abs_non_degen}, $\pi^{\textsc{newest}}_{n} \leq \pi^{\textsc{random}}_{n}$ if $n >n^{\star}$ and $\pi^{\textsc{newest}}_{n} \geq \pi^{\textsc{random}}_{n}$ if $n \leq n^{\star}$.
\end{lemma}
\begin{proof}[Proof of \cref{lemma:comparison_dist_newest_random_general}.]
By \cref{lemma:stationary_state_distribiton_newest_first}, it holds that $\pi^{\textsc{random}}_{n}  = \pi^{\textsc{newest}}_{n} \cdot \frac{1}{\kappa} \cdot \prob_{\Theta \sim \mathcal{F}}[\Theta + h(n) \geq p] $. By the definition of $n^{\star}$ and the monotonicity of $h(n)$, if $n \leq n^{\star}$, $\prob_{\Theta \sim \mathcal{F}}[\Theta + h(n) \geq p] \leq \kappa$ and thus $\pi^{\textsc{newest}}_{n} \geq \pi^{\textsc{random}}_{n}$. The other case is analogous. 
\end{proof}

\begin{proof}[Proof of \cref{thm_neg_bias_avg_review_rating}.]
 To show that $\newest$  has smaller average rating compared to
 $\random$, 
    \begin{align*}
        \expect_{N \sim \pi^{\textsc{random}}_n}[N]- \expect_{N \sim \pi^{\textsc{newest}}_n}[N]&=\sum_{m=0}^c m( \pi^{\textsc{random}}_m-\pi^{\textsc{newest}}_m)\\
        &= \underbrace{\sum_{m =n^{\star} + 1 }^{c}m( \pi^{\textsc{random}}_m-\pi^{\textsc{newest}}_m)}_{(1)} -\underbrace{\sum_{m=0}^{n^{\star}}m( \pi^{\textsc{newest}}_m-\pi^{\textsc{random}}_m)}_{(2)}
    \end{align*}
    \cref{lemma:comparison_dist_newest_random_general} yields $( \pi^{\textsc{random}}_m-\pi^{\textsc{newest}}_m) \geq 0$ for $m \geq n^{\star} +1$ and thus
    \begin{align*}
     (1) \geq (n^{\star}+1) \sum_{m =n^{\star} + 1 }^{c}( \pi^{\textsc{random}}_m-\pi^{\textsc{newest}}_m) &= (n^{\star}+1) \big(\prob_{N \sim \pi_{n}^{\textsc{random}}}[N > n^{\star}]-\prob_{N \sim \pi_{n}^{\textsc{newest}}}[N > n^{\star}]\big)\\
     &= (n^{\star}+1) \big(\prob_{N \sim \pi_{n}^{\textsc{newest}}}[N \leq n^{\star}]-\prob_{N \sim \pi_{n}^{\textsc{random}}}[N \leq n^{\star}]\big)
    \end{align*}
    where in the last equality we used that $\prob[X > n^{\star}] = 1-\prob[X \leq n^{\star}]$ for any random variable $X$. Using \cref{lemma:comparison_dist_newest_random_general} again yields  $( \pi^{\textsc{newest}}_m-\pi^{\textsc{random}}_m) \geq 0$ for $m \leq n^{\star}$ and thus
    $$(2) \leq n^{\star} \sum_{m =0 }^{n^{\star}}( \pi^{\textsc{newest}}_m-\pi^{\textsc{random}}_m) = n^{\star} \big(\prob_{N \sim \pi_{n}^{\textsc{newest}}}[N \leq n^{\star}]-\prob_{N \sim \pi_{n}^{\textsc{random}}}[N \leq n^{\star}]\big).$$
    Therefore,
    $$(1)-(2) \geq \prob_{N \sim \pi_{n}^{\textsc{newest}}}[N \leq n^{\star}]-\prob_{N \sim \pi_{n}^{\textsc{random}}}[N \leq n^{\star}].$$

To conclude the proof it remains to show that $\prob_{N \sim \pi_{n}^{\textsc{newest}}}[N \leq n^{\star}] > \prob_{N \sim \pi_{n}^{\textsc{random}}}[N \leq n^{\star}]$. By \cref{lemma:comparison_dist_newest_random_general}, $\pi^{\textsc{newest}}_{n} \geq \pi^{\textsc{random}}_{n}$ for $n \leq n^{\star}$. To show the claim it is enough to show that there is some $n \leq n^{\star}$ such that $\pi^{\textsc{newest}}_{n} > \pi^{\textsc{random}}_{n}$. We show that the purchase probability when all reviews are negative is strictly greater under $\newest$ than under $\random$ i.e. $\pi^{\textsc{newest}}_{0} > \pi^{\textsc{random}}_{0}$. By the monotonicity of $h$,  $\prob_{\Theta \sim \mathcal{F}}[\Theta + h(0) \geq p] \leq \prob_{\Theta \sim \mathcal{F}}[\Theta + h(n) \geq p]$ for all $n \in \{0,1 \ldots, c\}$. Since $p$ satisfies Assumption \ref{assumption:non_abs_non_degen} (and is thus non-degenerate), the inequality is strict when $n = c$. Taking the reciprocal and expectation of the last inequality yields the result.  
\end{proof}

\subsection{Stationary Distribution under Newest First (Lemma \ref{lemma:stationary_state_distribiton_newest_first})}\label{appendx:proof_lemma:stationary_state_distribiton_newest_first}
We first state a general property characterizing how the stationary distribution of a Markov chain changes if we modify it so that in every state, the process remains there with some probability. 

\begin{lemma}\label{lemma: general_theorem_markov_chains_stationary}
    Let $\mathcal{M}$ be a Markov chain on a finite state space $\mathcal{S}$ ($|\mathcal{S}| = m$) with a transition probability matrix $M \in \mathbb{R}^{m \times m}$ and a stationary distribution $\pi \in \mathbb{R}^m$. For any function $f:\mathcal{S} \to (0,1]$, we define a new Markov chain $\mathcal{M}_f$ on the same state space $\mathcal{S}$. At state $s \in \mathcal{S}$, $\mathcal{M}_f$ transitions according to the matrix $M$ with probability $f(s)$ and remains in $s$ with probability $1-f(s)$.
    Then $\pi_{f}(s) = \kappa \cdot\frac{\pi(s)}{f(s)}$ is a stationary distribution of $\mathcal{M}_f$ where $\kappa = 1/\sum_{s \in \mathcal{S}}\frac{\pi(s)}{f(s)}$ is a normalizing constant
\end{lemma}
\begin{proof}[Proof of \cref{lemma: general_theorem_markov_chains_stationary}.]
For any state $s \in \mathcal{S}$, the probability of a self-transition under $\mathcal{M}_f$ is $M_f(s,s) = f(s) M(s,s) + 1-f(s)$ since there are two ways to transition from $s$ back to itself: (1) the $f(s)$ transition followed by a transition back to $s$ via $\mathcal{M}$ and (2) the $1-f(s)$ transition that does not alter the current state. 
For all states $s \neq s'$, $M_f(s,s') = f(s) M(s,s')$ since the only way for $\mathcal{M}_f$ to transition from $s$ to $s'$ is to take a $f(s)$ transition at $s$ and follow the transitions of $\mathcal{M}$ to get to $s'$. 

The distribution $\Big\{\pi_{f}(s) = \kappa \cdot \frac{\pi(s)}{f(s)}\Big\}_{s\in\mathcal{S}}$ is a probability distribution by the definition of the normalizing constant $\kappa = 1/\sum_{s \in \mathcal{S}}\frac{\pi(s)}{f(s)}$. Using the transitions of $\mathcal{M}_f$, it holds that:
\begin{align*}
    \sum_{s' \in \mathcal{S}}\pi_{f}(s') M_f(s',s) &= \pi_f(s) M_{f}(s,s) + \sum_{s' \neq s}\pi_{f}(s') M_f(s',s) \\
    &=\pi_f(s) \cdot \big(1-f(s) + f(s) M(s,s) \big) +  \sum_{s' \neq s}\pi_f(s') \cdot\big( f(s') M(s',s) \big) \\
    &=\kappa \cdot \frac{\pi(s)}{f(s)} \cdot \big(1-f(s) + f(s) M(s,s) \big) +  \sum_{s' \neq s}\kappa \cdot \frac{\pi(s')}{f(s')} \cdot\big( f(s') M(s',s) \big) \\
    &= \kappa \cdot  \frac{\pi(s)}{f(s)} - \kappa \cdot \pi(s) + \underbrace{\kappa \cdot \pi(s) M(s,s) + \sum_{s' \neq s} \kappa \cdot \pi(s')M(s',s)}_{= \kappa \cdot \pi(s) \text{ (as $\pi$ is stationary for $\mathcal{M}$)}} = \pi_f(s).
\end{align*}
Hence, $\pi_f$ is a stationary distribution of $M_f$ as $\pi_f(s) = \sum_{s' \in \mathcal{S}}\pi_{f}(s') M_f(s',s)$ for all states $s$. 
\end{proof}

\begin{proof}[Proof of Lemma \ref{lemma:stationary_state_distribiton_newest_first}.]
We will show that the stationary distribution of the newest $c$ reviews $\bm{Z}_t$ is 
$$\pi_{(z_1, \ldots, z_{c})} = \kappa \cdot  \frac{\mu^{\sum_{i=1}^{c} z_i} (1-\mu)^{c-\sum_{i=1}^{c} z_i}}{\prob_{\Theta \sim \mathcal{F}}\Big[\Theta + h(\sum_{i=1}^{c} z_i) \geq p\Big] }$$
where $\kappa = 1/\expect_{N \sim \binomial(c, \mu)}\big[\frac{1}{\prob_{\Theta \sim \mathcal{F}}[\Theta + h(N) \geq p]} \big]$ is the normalizing constant. In the language of \cref{lemma: general_theorem_markov_chains_stationary}, $\bm{Z}_t$ corresponds to $\mathcal{M}_f$, the state space $\mathcal{S}$ to $\{0,1\}^c$, and $f$ is a function that expresses the purchase 
probability at a given state, i.e., $f(z_1, \ldots, z_c) = \prob_{\Theta \sim \mathcal{F}}\Big[\Theta + h(\sum_{i=1}^{c} z_i) \geq p\Big]$. Note that with probability $1-f(z_1, \ldots, z_c)$, $\bm{Z}_t$ remains at the same state (as there is no purchase). 

To apply \cref{lemma: general_theorem_markov_chains_stationary}, we need to show that whenever there is  purchase, $\bm{Z}_t$ transitions according to a Markov chain with stationary distribution $\mu^{\sum_{i=1}^{c} z_i} (1-\mu)^{c-\sum_{i=1}^{c} z_i}$. Consider the Markov chain $\mathcal{M}$ which always replaces the $c$-th last review with a new $\bern(\mu)$ review. This process has stationary distribution equal to the above numerator and $\bm{Z}_t$ transitions according to $\mathcal{M}$ upon a purchase, i.e., with probability $f(z_1, \ldots, z_c)$. As a result, by \cref{lemma: general_theorem_markov_chains_stationary}, $\pi$ is a stationary distribution for $\bm{Z}_t$.  As $\bm{Z}_t$ is irreducible and aperiodic, this is the unique stationary distribution. 
\end{proof}

\subsection{Closed-form expression for Cost of Newest First (Lemma~\ref{lemma: CoNF_ratio_closed_form_general_c})}\label{appendix_subsec:CoNF_ratio_closed_form_general_c}
\begin{proof}[Proof of \cref{lemma: CoNF_ratio_closed_form_general_c}.]
Dividing the expressions given in Propositions \ref{theorem:random_C_revenue} and \ref{theorem:most_recent_C_revenue}, the price term $p$ cancels out and the CoNF can be expressed as:
    \begin{align*}
        \chi(p) 
        &= \expect_{N \sim \binomial(c,\mu)}\big[\prob_{\Theta \sim \mathcal{F}}[\Theta + h(N) \geq p]\big] \cdot \expect_{N \sim \binomial(c,\mu)}\Big[\frac{1}{\prob_{\Theta \sim \mathcal{F}}[\Theta + h(N) \geq p]} \Big] \\
        &= \Big(\sum_{i=0}^{c} \mu^{i} (1-\mu)^{c-i}\binom{c}{i} \prob_{\Theta \sim \mathcal{F}}[\Theta + h(i) \geq p] \Big) \cdot \Big( \sum_{j=0}^{c} \mu^{j} (1-\mu)^{c-j}\binom{c}{j} \frac{1}{\prob_{\Theta \sim \mathcal{F}}[\Theta + h(j) \geq p]} \Big)\\
        &= \sum_{i,j \in \{0, \ldots,c\}} \mu^{i+j}(1-\mu)^{2c-i-j} \binom{c}{i} \binom{c}{j} \frac{ \prob_{\Theta \sim \mathcal{F}}[\Theta + h(i) \geq p]}{ \prob_{\Theta \sim \mathcal{F}}[\Theta + h(j) \geq p]}. 
    \end{align*}
 \end{proof}

\subsection{CoNF can be arbitrarily large under static pricing (Theorem~\ref{thm:conf_arbitrarily_bad})}\label{appendix_subsec_conf_arb_bad_proof}

\begin{proof}[Proof of \cref{thm:conf_arbitrarily_bad}]
One summand in the right hand side of \cref{lemma: CoNF_ratio_closed_form_general_c} contains the ratio of the purchase probability of all reviews being positive compared to all reviews being negative:
$$\beta(p) \coloneqq \frac{ \prob_{\Theta \sim \mathcal{F}}[\Theta + h(c) \geq p]}{ \prob_{\Theta \sim \mathcal{F}}[\Theta + h(0) \geq p]},$$
which quantifies how much the reviews affect the purchase probability. Since all other terms are non-negative, the \textsc{CoNF} is lower bounded by this summand, i.e., $\chi(p) \geq \mu^c(1-\mu)^c \beta(p)$. 

Since $\mathcal{F}$ is bounded, suppose that its support is $[\underline{\theta}, \overline{\theta}]$. 
When all reviews are negative, selecting a price of $h(0) + \overline{\theta}$ results in a purchase probability of $0$. 
Combined with the continuity of $\mathcal{F}$, this implies that, when $p \to h(0) + \overline{\theta}$, the purchase probability goes to 0.
If, on the other hand, all reviews were positive, then using $h(c) > h(0)$ and that $\mathcal{F}$ is continuous and has positive mass on its support, the purchase probability is positive; i.e., $\lim_{p  \to h(0) + \overline{\theta}} \prob_{\Theta \sim \mathcal{F}}[\Theta + h(c) \geq p] > 0$.
Therefore $\lim_{p \to h(0) + \overline{\theta}} \beta(p) = +\infty$, which implies that $\lim_{p \to h(0) + \overline{\theta}} \chi(p) = +\infty$ since $\chi(p) \geq \mu^c(1-\mu)^c \beta(p)$.

Lastly, for any price $p \in (\underline{\theta} + h(0), \overline{\theta} + h(0))$ it holds
$0 < \prob_{\Theta \sim \mathcal{F}}[\Theta + h(0) \geq p] < \prob_{\Theta \sim \mathcal{F}}[\Theta + h(c) \geq p]$ and thus it is non-degenerate and non-absorbing. Combing this with the fact that $\overline{\theta} > 0$ (by Assumption \ref{assumption:non-negative_mass_idiosyncratic}), for every $M > 0$ there exists $\epsilon(M) > 0$ such that $p = \overline{\theta} + h(0)-\epsilon(M) $ satisfies Assumption \ref{assumption:non_abs_non_degen} and has CoNF $\chi(p) > M$. \end{proof}

\subsection{Monotonicity of CoNF under monotone hazard rate (Proposition~\ref{thm:mhr_conf_increasing_in_price})}\label{appendix_subsec:conf_monotonicity_under_mhr} 

\begin{proof}[Proof of \cref{thm:mhr_conf_increasing_in_price}. ]
    By \cref{lemma: CoNF_ratio_closed_form_general_c}, the CoNF is given by
    \begin{align*}
        \chi(p) &= \sum_{i,j \in \{0,\ldots, c\}} \mu^{i+j}(1-\mu)^{2c-i-j} \binom{c}{i} \binom{c}{j} \frac{\prob_{\Theta \sim \mathcal{F}} [\Theta + h(i) \geq p]}{\prob_{\Theta \sim \mathcal{F}} [\Theta + h(j) \geq p]}\\
        &= \sum_{i=0}^{c} \mu^{2i}(1-\mu)^{2(c-i)} \binom{c}{i}^2 \\
        &\quad+ \sum_{i < j} \mu^{i+j}(1-\mu)^{2c-i-j} \binom{c}{i} \binom{c}{j}\Bigg( \frac{\prob_{\Theta \sim \mathcal{F}} [\Theta + h(i) \geq p]}{\prob_{\Theta \sim \mathcal{F}} [\Theta + h(j) \geq p]} +  \frac{\prob_{\Theta \sim \mathcal{F}} [\Theta + h(j) \geq p]}{\prob_{\Theta \sim \mathcal{F}} [\Theta + h(i) \geq p]} \Bigg).
    \end{align*}
Letting $F$ be the cumulative density function of $\mathcal{F}$ and $\overline{F}(u) \coloneqq 1-F(u)$ be the survival function:
$$\chi(p) = \sum_{i=0}^{c} \mu^{2i}(1-\mu)^{2(c-i)} \binom{c}{i}^2 +  \sum_{i < j} \mu^{i+j}(1-\mu)^{2c-i-j} \binom{c}{i} \binom{c}{j} \Bigg( \frac{\overline{F}(p-h(i))}{\overline{F}(p-h(j))} +  \frac{\overline{F}(p-h(j))}{\overline{F}(p-h(i))} \Bigg).$$
We denote the ratio of the purchase probability with $j$ and $i$ positive reviews by $u_{i,j}(p) \coloneqq \frac{\overline{F}(p-h(j))}{\overline{F}(p-h(i))}$. For $i < j$, this ratio is $u_{i,j}(p) \geq 1$ because $p-h(i) > p-h(j)$ due to the monotonicity of $h$ (Assumption \ref{assumption:h_monotonicity_positive_reviews}). Given that $\mathcal{F}$ is continuous, $\overline{F}(x)$ is differentiable for any $x$ such that $\overline{F}(x) \in (0,1)$. Furthermore for $p \in (\underline{\theta} + h(c), \overline{\theta} + h(0))$, $\overline{F}(p-h(i)) \in (0,1)$ for all $i \in \{0,1, \ldots, c\}$, and thus $u_{i,j}(p)$ is differentiable at $p$.

We now show that $u_{i,j}(p)$ is non-decreasing in $p$. Letting $f$ be the probability density function of $\mathcal{F}$ and taking the derivative of $u_{i,j}(p)$ with respect to $p$:
\begin{align*}
\frac{d}{dp} u_{i,j}(p) &= \frac{f(p-h(i)) \overline{F}(p-h(j)) - f(p-h(j)) \overline{F}(p-h(i))}{\overline{F}(p-h(j))^2} \\
&=\frac{\Big(\frac{f(p-h(i))}{\overline{F}(p-h(i))} - \frac{f(p-h(j))}{\overline{F}(p-h(j))} \Big) \overline{F}(p-h(j)) \overline{F}(p-h(i))}{\overline{F}(p-h(j))^2}.
\end{align*}
Observe that $p-h(i) > p-h(j)$ since $i < j$ and the strict monotonicity of $h$. By the MHR property of $\mathcal{F}$: $\frac{f(p-h(i))}{\overline{F}(p-h(i))} \geq  \frac{f(p-h(j))}{\overline{F}(p-h(j))}$, implying that $\frac{d}{dp} u_{i,j}(p) \geq 0$, and thus $u_{i,j}(p)$ is non-decreasing in $p$. 

To finish the proof of the theorem, we rewrite $\chi(p)$ as a function of $\{u_{i,j}(p)\}_{i <j }$ as
$$\chi(p) = \sum_{i=0}^{c} \mu^{2i}(1-\mu)^{2(c-i)} \binom{c}{i}^2 +  \sum_{i < j} \mu^{i+j}(1-\mu)^{2c-i-j} \binom{c}{i} \binom{c}{j} \Bigg( u_{i,j}(p) +  \frac{1}{u_{i,j}(p)} \Bigg)$$
Note that the function $u + \frac{1}{u}$ is non-decreasing for $u \geq 1$. Since $u_{i,j}(p) \geq 1$ is non-decreasing in $p$, then $u_{i,j}(p) +  \frac{1}{u_{i,j}(p)}$ is monotonically increasing in $p$ for every $i < j$. Thus, $\chi(p)$ is non-decreasing in $p$. 
\end{proof}
\begin{remark}
    The same proof also extends to cases when $\underline{\theta} = -\infty$ and/or $\overline{\theta} = +\infty$. 
\end{remark}

\subsection{Bounding CoNF by the sensitivity of purchases in reviews (Proposition~\ref{prop:CoNF_less_beta})}\label{appendix:when_is_conf_small}

\begin{proof}[Proof of Proposition~\ref{prop:CoNF_less_beta}.]
    The monotonicity of $h$, implies that for all $n \in \{0,1, \ldots, c\}$:
    $$\prob_{\Theta \sim \mathcal{F}}[\Theta + h(c) \geq p] \geq \prob_{\Theta \sim \mathcal{F}}[\Theta + h(n) \geq p] \geq \prob_{\Theta \sim \mathcal{F}}[\Theta + h(0) \geq p].$$
    Taking expectation over the number of positive reviews $N \sim \binomial(c, \mu)$, the first inequality implies
\begin{equation}\label{ineq:rev_random_ineq_prob}
        \expect_{N \sim \binomial(c, \mu)}\big[\prob_{\Theta \sim \mathcal{F}}[\Theta + h(N) \geq p]\big] \leq \prob_{\Theta \sim \mathcal{F}}[\Theta + h(c) \geq p].
    \end{equation}
    Similarly, taking expectation of the the reciprocal of the second inequality (which is well-defined as $p$ satisfies Assumption \ref{assumption:non_abs_non_degen} and is thus non-absorbing) implies that
\begin{equation}\label{ineq:rev_newest_ineq_prob}
        \expect_{N \sim \binomial(c, \mu)}\Big[\frac{1}{\prob_{\Theta \sim \mathcal{F}}[\Theta + h(N) \geq p]}\Big] \leq \frac{1}{\prob_{\Theta \sim \mathcal{F}}[\Theta + h(0) \geq p]}.
    \end{equation}
    Expressing $\chi(\pi)$ as the ratio of the expressions in Propositions \ref{theorem:random_C_revenue} and \ref{theorem:most_recent_C_revenue} and using (\ref{ineq:rev_random_ineq_prob}) and (\ref{ineq:rev_newest_ineq_prob}):
    \begin{align*}
    \chi(p) &= \expect_{N \sim \binomial(c, \mu)}[\prob_{\Theta \sim \mathcal{F}}[\Theta + h(N) \geq p]] \cdot \expect_{N \sim \binomial(c, \mu)}\Big[\frac{1}{\prob_{\Theta \sim \mathcal{F}}[\Theta + h(N) \geq p]}\Big] \\
    &\leq \frac{\prob_{\Theta \sim \mathcal{F}}[\Theta + h(c) \geq p]}{\prob_{\Theta \sim \mathcal{F}}[\Theta + h(0) \geq p]} = \beta(p)
    \end{align*}
    which concludes the proof. 
\end{proof}

\subsection{Phenomenon holds under a generic update rule (Section \ref{subsec:driver_persistence_of_negative_reviews})}\label{app:driver_generic_updating}

Consider an item with product quality $\mu = 1/2$ and customers who read a single review ($c = 1$).
Let $q(r, t)$ be the probability that the customer at round $t$ purchases the item when the newest review is $r \in \{0, 1\}$.
We consider any customer purchase model that satisfies the intuitive condition that the purchase probability given a negative review is strictly smaller than the purchase probability given a positive review. We also require that the purchase probability is always strictly smaller than $1$.

\begin{assumption}\label{assumption:negative_reviews_lower_purchase}
For every round $t \geq 1$, $q(0, t) < q(1, t) < 1$.
\end{assumption}
We assume that the review at time $t = 0$ is randomly initialized and is positive with probability~$1/2$.

\begin{proposition}\label{prop:conf_exists_general_update_rule}
For any round $t \geq 2$, the newest review is negative with probability strictly greater than $1/2$.
\end{proposition}

\begin{proof}
Fix any round $t$. Suppose the newest review $r_t$ came from the previous round $t-1$. Then, the review is negative with probability exactly equal to $1/2$.
Next, suppose the newest review came from round $t-k$, where $k \geq 2$.
Then, the probability that the review is negative is
\begin{align*}
   \Pr(r_t=0\;|\; \text{review $r_t$ came from time $t-k$}) 
   &= \frac{\Pr(r_t=0, \text{review $r_t$ came from time $t-k$})}{\Pr(\text{review $r_t$ came from time $t-k$})} \\
   =&\; \frac{\frac{1}{2} \prod_{\ell=1}^{k-1} (1-q(0, t-\ell))}
   {\frac{1}{2} \prod_{\ell=1}^{k-1} (1-q(0, t-\ell)) + \frac{1}{2} \prod_{\ell=1}^{k-1} (1-q(1, t-\ell))} 
   > \frac{1}{2}.
\end{align*}
The last inequality holds since $(1-q(1, t-\ell)) < (1-q(0, t-\ell))$ for all $\ell \in [1, k-1]$ by Assumption~\ref{assumption:negative_reviews_lower_purchase}.
Therefore, for any $k \geq 2$, if the review came from time $t-k$, then the review is negative with probability strictly greater than $1/2$. 
Since the purchase probability at any round is strictly less than 1 (by  Assumption~\ref{assumption:negative_reviews_lower_purchase}), there is a non-zero probability that the review came from round $t-2$ or before, which implies that the  newest review is negative with probability strictly less than $1/2$.
\end{proof}

\subsection{CoNF exists due to limited attention (Section \ref{subsec:driver_persistence_of_negative_reviews})}\label{appendix_subsec_conf_driver_bounded_rationality}

Suppose that the customer's product quality estimator $h$ is the mean, $h[\Phi] = \expect[\Phi]$ for all distributions $\Phi$. We refer to an instance as tuple $(\mathcal{F},\mu, a,b,p)$ where $\mathcal{F}$ is the customer's idiosyncratic value distribution, $\mu$ is the product quality, $a$ and $b$ are the prior parameters, and $p$ is the price the platform sets. For any limited attention window $c$, we define the CoNF as $\chi_c(p) = \frac{\rev_c(\random,p)}{\rev_c(\newest,p)}$ where  $\rev_c(\sigma,\rho)$ is the revenue under review ordering policy $\sigma$ and pricing policy $\rho$ when customers read $c$ reviews. Proposition \ref{prop:for_any_c_M_instance_1_bad_conf_2_conf_to_one_without_bd_rationality} establishes that the existence of CoNF is indeed due to the limited attention of the customers. In particular, for any limited attention window $\tilde{c} $ and any $M> 0$ there exists an instance so that the $\chi_{\tilde{c} }(p) \geq M$, but $\lim_{c  \to \infty} \chi_{c}(p) = 1$. Proposition \ref{prop:more_broadly_conf_to_one_without_bd_rationality} then establishes that the latter result holds broadly under mild assumptions. As a result, CoNF disappears when customers observe full history.

\begin{proposition}\label{prop:for_any_c_M_instance_1_bad_conf_2_conf_to_one_without_bd_rationality}
    For any $\tilde{c} \geq 1$ and any $M > 0$ there exists an instance $\mathcal{E}_{\tilde{c},M}=(\mathcal{F},a,b, p)$ such that for the instance $\mathcal{E}_{\tilde{c},M}$, it holds that a)  $\chi_{\tilde{c} }(p) \geq M$ and b)
  $\lim_{c \to \infty} \chi_{c}(p) = 1$.
\end{proposition}
\begin{proof}For any $\tilde{c} \geq 1$ and $M > 0$ we define the instance $\mathcal{E}_{\tilde{c},M}$ to have a customer value distribution of $\mathcal{F} = Exponential(\lambda_{\tilde{c}, M})$ where $\lambda_{\tilde{c}, M} = 3 \ln(2^{2 \tilde{c}} M)$, a product quality $\mu = \frac{1}{2}$, prior parameters $a = b= 1$, and a price $p = 1-\frac{1}{2+\tilde{c} }$.

To prove (a), note that, given $\mu = \frac{1}{2}$ we can express $\chi_{\tilde{c}}(p)$ using Lemma \ref{lemma: CoNF_ratio_closed_form_general_c} as
\begin{equation}\label{inequality:conf_lower_bound_specific_instance_in_conf_due_to_bd_rationality_app}
\chi_{\tilde{c}}(p) = \sum_{i,j \in \{0,1, \ldots,\tilde{c}\} } \frac{1}{2^{2 \tilde{c}}} \binom{\tilde{c}}{i}\binom{\tilde{c}}{j} \frac{\prob_{\Theta \sim \mathcal{F}}[\Theta + h(i) \geq p]}{\prob_{\Theta \sim \mathcal{F}}[\Theta + h(j) \geq p]} \geq \frac{1}{2^{2 \tilde{c}}} \frac{\prob_{\Theta \sim \mathcal{F}}[\Theta + h(\tilde{c}) \geq p]}{\prob_{\Theta \sim \mathcal{F}}[\Theta + h(0) \geq p]}.\end{equation}
The inequality holds as one of the terms is $\frac{1}{2^{2\tilde{c}}} \frac{\prob_{\Theta \sim \mathcal{F}}[\Theta + h(\tilde{c}) \geq p]}{\prob_{\Theta \sim \mathcal{F}}[\Theta + h(0) \geq p]}$ and all other terms are non-negative.

Given that $h$ is the mean and $a=b= 1$ then $h(\tilde{c}) = h(\mathrm{Beta}(1+\tilde{c}, 1)) = \frac{1+\tilde{c}}{2+\tilde{c}}$ and $h(0)= h(\mathrm{Beta}(1, 1+\tilde{c}))= \frac{1}{2+\tilde{c}}$. Combining this with the fact that $p = 1-\frac{1}{2+\tilde{c}}$, the purchase probability with all positive reviews equals
\begin{equation}\label{eq:probability_all_positive_reviews}
    \prob_{\Theta \sim \mathcal{F}}\Big[\Theta + \frac{\tilde{c}+1}{\tilde{c}+2} \geq 1-\frac{1}{\tilde{c}+2} \Big] = \prob_{\Theta \sim \mathcal{F}}[\Theta \geq 0]=1
\end{equation}
as the exponential distribution is always non-negative. The purchase probability for all negative reviews equals
\begin{equation}\label{eq:probability_all_negative_reviews}
    \prob_{\Theta \sim \mathcal{F}} \Big[\Theta + \frac{1}{\tilde{c}+1} \geq 1-\frac{1}{\tilde{c}+2} \Big] = \prob_{\Theta \sim \bern(\frac{1}{2^{2\tilde{c}}M})}\Big[\Theta \geq 1- \frac{2}{2+\tilde{c}} \Big] \leq \prob_{\Theta \sim \mathcal{F}}\Big[\Theta \geq \frac{1}{3} \Big] = \exp(-\frac{\lambda_{\tilde{c},M}}{3})
\end{equation}
where in the first equality we used that $\tilde{c} \geq 1$ and in the last equality we used that $\prob_{\Theta \sim \mathcal{F}}[\Theta \geq x] = \exp(-\lambda_{\tilde{c}, M} x)$. Using that $\lambda_{\tilde{c}, M} = 3 \ln(2^{2 \tilde{c}} M)$, \eqref{eq:probability_all_negative_reviews} yields  $ \prob_{\Theta \sim \mathcal{F}} \Big[\Theta + \frac{1}{\tilde{c}+1} \geq 1-\frac{1}{\tilde{c}+2} \Big] \leq \frac{1}{2^{2 \tilde{c}} M}$.
Combining \eqref{eq:probability_all_positive_reviews} and \eqref{eq:probability_all_negative_reviews} with \eqref{inequality:conf_lower_bound_specific_instance_in_conf_due_to_bd_rationality_app} yields 
$\chi_{c}(p) \geq \frac{1}{2^{2\tilde{c}}} 2^{2\tilde{c}}M = M$
as desired.

To prove (b), the formulas for the revenue of $\random$ and $\newest$ given by Propositions \ref{theorem:random_C_revenue} and \ref{theorem:most_recent_C_revenue} and the fact that $\mu = \frac{1}{2}$ and that $h$ is the mean imply that $\chi_{c}(p)$ can be expressed as 
\begin{equation*}
    \frac{\rev_{c}(\random,p)}{\rev_{c}( \newest,p)} = \expect_{N \sim \binomial(c, \frac{1}{2})}\Big[\prob_{\Theta \sim \mathcal{F}}[\Theta + \frac{1+N}{2+c} \geq p] \Big] \expect_{N \sim \binomial(c, \frac{1}{2})} \Bigg[ \frac{1}{\prob_{\Theta \sim \mathcal{F}}[\Theta + \frac{1+N}{2+c} \geq p]} \Bigg].
\end{equation*}
For any number of positive reviews $N \in \{0,1, \ldots, c\}$, the purchase probability equals
\begin{equation}\label{equation:formula_purchase_probability_instance_E_M_tilde_c}
    \prob_{\Theta \sim \mathcal{F}}\Big[\Theta + \frac{1+N}{2+c} \geq p \Big] = \prob_{\Theta \sim \mathcal{F}}\Big[\Theta \geq 1-\frac{1}{2+\tilde{c}}- \frac{1+N}{2+c}  \Big] = \exp(-\lambda_{\tilde{c},M}(1-\frac{1}{2+\tilde{c}}- \frac{1+N}{2+c} ))
\end{equation}

Using \eqref{equation:formula_purchase_probability_instance_E_M_tilde_c} in the expression for $\chi_{c}(p)$, the term $\exp(-\lambda_{\tilde{c},M}(1-\frac{1}{2+\tilde{c}}))$ cancels out yielding
\begin{align}\label{eq:expression_for_chi_p_MGF_binomail}
    \chi_{c}(p) &=\expect_{N \sim \binomial(c, \frac{1}{2})}\Big[\exp(-\lambda_{\tilde{c},M}(1-\frac{1}{2+\tilde{c}}- \frac{1+N}{2+c} )) \Big] \expect_{N \sim \binomial(c, \frac{1}{2})} \Bigg[ \exp(\lambda_{\tilde{c},M}(1-\frac{1}{2+\tilde{c}}- \frac{1+N}{2+c} )) \Bigg] \nonumber\\
    &=\expect_{N \sim \binomial(c, \frac{1}{2})}\Big[\exp(\lambda_{\tilde{c},M}\frac{1+N}{2+c} ) \Big] \expect_{N \sim \binomial(c, \frac{1}{2})} \Bigg[ \exp(-\lambda_{\tilde{c},M} \frac{1+N}{2+c} ) \Bigg].
\end{align}
Let $\mathcal{E}^{\textsc{Conc}} $ be the event that $\{N \in [\frac{1}{2} c-c^{\frac{2}{3}}, \frac{1}{2}c + c^{\frac{2}{3}}]\}$. The Chernoff bound implies that $\prob[\mathcal{E}^{\textsc{Conc}}] \geq 1-2e^{-\frac{c^{\frac{1}{3}}}{3}}$. The law of total expectation applied to the first expectation in \eqref{eq:expression_for_chi_p_MGF_binomail} yields 
\begin{align*}
    \expect_{N \sim \binomial(c, \frac{1}{2})}\Big[\exp(\lambda_{\tilde{c},M}\frac{1+N}{2+c} ) \Big]  &= \underbrace{\expect_{N \sim \binomial(c, \frac{1}{2})}\Big[\exp(\lambda_{\tilde{c},M}\frac{1+N}{2+c} )|\mathcal{E}^{\textsc{Conc}} \Big] \prob[\mathcal{E}^{\textsc{Conc}} ]}_{=(1)} \\
    &+ \underbrace{\expect_{N \sim \binomial(c, \frac{1}{2})}\Big[\exp(\lambda_{\tilde{c},M}\frac{1+N}{2+c} ) |(\mathcal{E}^{\textsc{Conc}} )^{\perp}\Big] \prob[(\mathcal{E}^{\textsc{Conc}})^{\perp}]}_{=(2)}
\end{align*}
Notice that $\exp(\lambda_{\tilde{c},M}(\frac{1+N}{2+c} )) \in [1, \exp(\lambda_{\tilde{c},M})]$ for all $N \in \{0,\ldots, c\}$. Combining this with $\lim_{c \to \infty}\prob[(\mathcal{E}^{\textsc{Conc}})^{\perp}] = 0$ yields that the term $(2)$ goes to $0$ as $ c \to \infty$. Using the definition of $\mathcal{E}^{\textsc{Conc}}$, we can upper and lower bound the term $(1)$ as 
$$\exp \Big(\lambda_{\tilde{c}, M} \frac{1+\frac{1}{2} c -c^{\frac{2}{3}}}{2+c} \Big) \prob[\mathcal{E}^{\textsc{Conc}}] \leq (1) \leq \exp\Big(\lambda_{\tilde{c}, M} \frac{1+\frac{1}{2} c +c^{\frac{2}{3}}}{2+c} \Big) \prob[\mathcal{E}^{\textsc{Conc}}].$$
As the $\lim_{c \to \infty} \prob[\mathcal{E}^{\textsc{Conc}}] = 1$, both right-hand side and left-hand side converge to $\exp\Big(\frac{\lambda_{\tilde{c}, M}}{2} \Big)$. Thus, $\lim_{c \to \infty} \expect_{N \sim \binomial(c, \frac{1}{2})}\Big[\exp(\lambda_{\tilde{c},M}\frac{1+N}{2+c} ) \Big]  = \exp\Big(\frac{\lambda_{\tilde{c}, M}}{2} \Big)$. Similarly, the law of total expectation applied to the second expectation in \eqref{eq:expression_for_chi_p_MGF_binomail} yields 
\begin{align*}
    \expect_{N \sim \binomial(c, \frac{1}{2})}\Big[\exp(-\lambda_{\tilde{c},M}\frac{1+N}{2+c} ) \Big]  &= \underbrace{\expect_{N \sim \binomial(c, \frac{1}{2})}\Big[\exp(-\lambda_{\tilde{c},M}\frac{1+N}{2+c} )|\mathcal{E}^{\textsc{Conc}} \Big] \prob[\mathcal{E}^{\textsc{Conc}} ]}_{=(3)} \\
    &+ \underbrace{\expect_{N \sim \binomial(c, \frac{1}{2})}\Big[\exp(-\lambda_{\tilde{c},M}\frac{1+N}{2+c} ) |(\mathcal{E}^{\textsc{Conc}} )^{\perp}\Big] \prob[(\mathcal{E}^{\textsc{Conc}})^{\perp}]}_{=(4)}
\end{align*}
Notice that $\exp(-\lambda_{\tilde{c},M}(\frac{1+N}{2+c} )) \in [\exp(-\lambda_{\tilde{c},M}),1]$ for all $N \in \{0,\ldots, c\}$. Combining this with $\lim_{c \to \infty}\prob[(\mathcal{E}^{\textsc{Conc}})^{\perp}] = 0$ yields that the term $(4)$ goes to $0$ as $ c \to \infty$. Using the definition of $\mathcal{E}^{\textsc{Conc}}$, we can upper and lower bound the term $(3)$ as 
$$\exp \Big(-\lambda_{\tilde{c}, M} \frac{1+\frac{1}{2} c +c^{\frac{2}{3}}}{2+c} \Big) \prob[\mathcal{E}^{\textsc{Conc}}] \leq (3) \leq \exp\Big(-\lambda_{\tilde{c}, M} \frac{1+\frac{1}{2} c -c^{\frac{2}{3}}}{2+c} \Big) \prob[\mathcal{E}^{\textsc{Conc}}].$$
As the $\lim_{c \to \infty} \prob[\mathcal{E}^{\textsc{Conc}}] = 1$, both right-hand side and left-hand side converge to $\exp\Big(-\frac{\lambda_{\tilde{c}, M}}{2} \Big)$. Thus, $\lim_{c \to \infty} \expect_{N \sim \binomial(c, \frac{1}{2})}\Big[\exp(-\lambda_{\tilde{c},M}\frac{1+N}{2+c} ) \Big]  = \exp\Big(-\frac{\lambda_{\tilde{c}, M}}{2} \Big)$. Thus, $\lim_{c \to \infty} \chi_{c}(p) = 1$.
\end{proof}

\begin{proposition}\label{prop:more_broadly_conf_to_one_without_bd_rationality}
    Fix an instance $(\mathcal{F}, \mu, a, b, p)$ such that $\mathcal{F}$ is continuous, and suppose that $c$ varies. 
    If the purchase probability is lower bounded by $\delta > 0$ for all review states and for all $c$, then $\lim_{c\rightarrow\infty}\chi_{c}(p)=1$.
\end{proposition}

The instance $\mathcal{E}_{\tilde{c},M}$ of Proposition~\ref{prop:for_any_c_M_instance_1_bad_conf_2_conf_to_one_without_bd_rationality} has a lower bound $\delta=\exp(-3\ln(2^{2 \tilde{c} } M)) = \frac{1}{2^{6\tilde{c}}M^{3}}$. This implies that the existence of a lower bound $\delta > 0$ does not imply a bounded CoNF with limited attention. We note that assuming a lower bound on the purchase probability is common in the literature on social learning with reviews \citep{cims17,bs18,imsz19}.

\begin{proof}[Proof of \cref{prop:more_broadly_conf_to_one_without_bd_rationality}]
Let $\mathcal{E} = (\mathcal{F}, \mu, a,b, p)$ be an instance such that the purchase probability is lower bounded by $\delta > 0$ for any limited attention window $c \geq 1$ and number of positive reviews $N \in \{0,1, \ldots, c\}$, i.e., 
\begin{equation}\label{ineq:lower_bound_purch_prob_delta_all_reviews_c}
\prob_{\Theta \sim \mathcal{F}}\Big[\Theta + \frac{N+a}{a+b+c} \geq p \Big] \geq \delta \text{ for all $c \geq 1, N \in \{0, \ldots, c\}$}.
\end{equation}

Using the formulas for the revenue of $\random$ and $\newest$ given by Propositions \ref{theorem:random_C_revenue} and \ref{theorem:most_recent_C_revenue}, the CoNF can be expressed as 
$$\frac{\rev_c(\random,p)}{\rev_c(\newest,p)} = \expect_{N \sim \binomial(c, \mu)}\Big[\prob_{\Theta \sim \mathcal{F}}\big[\Theta + \frac{N +a}{a+b+c} \geq p \big] \Big] \expect_{N \sim \binomial(c, \mu)}\Bigg[\frac{1}{\prob_{\Theta \sim \mathcal{F}}\big[\Theta + \frac{N +a}{a+b+c} \geq p\big]} \Bigg].$$
Let $\mathcal{E}^{\textsc{Conc}}$ be the event that $\{ N \in [\mu c - c^{\frac{2}{3}}, \mu c+ c^{\frac{2}{3}}]\}$. The Chernoff bound implies that $\prob[\mathcal{E}^{\textsc{Conc}}] \geq 1 - 2e^{-\frac{c^{\frac{1}{3}}}{3}}$. The law of total expectation applied to the first expectation yields
\begin{align*}
     \expect_{N \sim \binomial(c, \mu)}\Big[\prob_{\Theta \sim \mathcal{F}}\big[\Theta &+ \frac{N +a}{a+b+c} \geq p \big] \Big] = \underbrace{\expect_{N \sim \binomial(c, \mu)}\Big[\prob_{\Theta \sim \mathcal{F}}\big[\Theta + \frac{N +a}{a+b+c} \geq p \big] | \mathcal{E}^{\textsc{Conc}} \Big] \prob[\mathcal{E}^{\textsc{Conc}}]}_{=(1)}\\
     &+\underbrace{\expect_{N \sim \binomial(c, \mu)}\Big[\prob_{\Theta \sim \mathcal{F}}\big[\Theta + \frac{N +a}{a+b+c} \geq p \big] | (\mathcal{E}^{\textsc{Conc}})^{\perp} \Big] \prob[(\mathcal{E}^{\textsc{Conc}})^{\perp}]}_{=(2)}
\end{align*}
Notice that $\prob_{\Theta \sim \mathcal{F}}\big[\Theta + \frac{N +a}{a+b+c} \geq p \big] \in [0,1]$ for any $N \in \{0, \ldots, c\}$. Combining this with $\prob[(\mathcal{E}^{\textsc{Conc}})^{\perp}] \to_{c \to \infty} 0$, yields that $(2)$ converges to zero as $c \to \infty$. Using the definition of $\mathcal{E}^{\textsc{Conc}}$ we can lower and upper bound the term (1) as 
$$\prob_{\Theta \sim \mathcal{F}}\big[\Theta + \frac{\mu c +c^{\frac{2}{3}}+a}{a+b+c} \geq p \big] \prob[\mathcal{E}^{\textsc{Conc}}] \geq (1) \geq \prob_{\Theta \sim \mathcal{F}}\big[\Theta + \frac{\mu c -c^{\frac{2}{3}}+a}{a+b+c} \geq p \big]\prob[\mathcal{E}^{\textsc{Conc}}].
$$
Given that $\frac{\mu c +c^{\frac{2}{3}}+a}{a+b+c} , \frac{\mu c -c^{\frac{2}{3}}+a}{a+b+c} \to_{c \to \infty} \mu$ and $\mathcal{F}$ is continuous $\prob_{\Theta \sim \mathcal{F}}\big[\Theta + \frac{\mu c +c^{\frac{2}{3}}+a}{a+b+c} \geq p \big] , \prob_{\Theta \sim \mathcal{F}}\big[\Theta + \frac{\mu c -c^{\frac{2}{3}}+a}{a+b+c} \geq p \big] \to_{c \to \infty} \prob_{\Theta \sim \mathcal{F}}[\Theta +\mu \geq p]$. Combining this with $\lim_{c \to \infty}\prob[\mathcal{E}^{\textsc{Conc}}] = 1$ yields that $(1)$ converges to $\prob[\Theta + \mu \geq p]$ as $c \to \infty$. Thus, $\lim_{c \to \infty}  \expect_{N \sim \binomial(c, \mu)}\Big[\prob_{\Theta \sim \mathcal{F}}\big[\Theta + \frac{N +a}{a+b+c} \geq p \big] \Big] = \prob[\Theta + \mu \geq p]$. The law of total expectation applied to the second expectation yields
\begin{align*}
     \expect_{N \sim \binomial(c, \mu)}\Big[&\frac{1}{\prob_{\Theta \sim \mathcal{F}}\big[\Theta + \frac{N +a}{a+b+c} \geq p \big] }\Big] = \underbrace{\expect_{N \sim \binomial(c, \mu)}\Big[\frac{1}{\prob_{\Theta \sim \mathcal{F}}\big[\Theta + \frac{N +a}{a+b+c} \geq p \big]} | \mathcal{E}^{\textsc{Conc}} \Big] \prob[\mathcal{E}^{\textsc{Conc}}]}_{=(3)}\\
     &+\underbrace{\expect_{N \sim \binomial(c, \mu)}\Big[\frac{1}{\prob_{\Theta \sim \mathcal{F}}\big[\Theta + \frac{N +a}{a+b+c} \geq p \big]} | (\mathcal{E}^{\textsc{Conc}})^{\perp} \Big] \prob[(\mathcal{E}^{\textsc{Conc}})^{\perp}]}_{=(4)}
\end{align*}
Notice that $\frac{1}{\prob_{\Theta \sim \mathcal{F}}\big[\Theta + \frac{N +a}{a+b+c} \geq p \big]} \in [0,\frac{1}{\delta}]$ for any $N \in \{0, \ldots, c\}$ due to \eqref{ineq:lower_bound_purch_prob_delta_all_reviews_c}. Combining this with $\prob[(\mathcal{E}^{\textsc{Conc}})^{\perp}] \to_{c \to \infty} 0$, yields that $(4)$ converges to zero as $c \to \infty$. Using the definition of $\mathcal{E}^{\textsc{Conc}}$ we can lower and upper bound the term $(3)$ as 
$$\frac{1}{\prob_{\Theta \sim \mathcal{F}}\big[\Theta + \frac{\mu c -c^{\frac{2}{3}}+a}{a+b+c} \geq p \big]} \prob[\mathcal{E}^{\textsc{Conc}}] \geq (3) \geq \frac{1}{\prob_{\Theta \sim \mathcal{F}}\big[\Theta + \frac{\mu c +c^{\frac{2}{3}}+a}{a+b+c} \geq p \big]}\prob[\mathcal{E}^{\textsc{Conc}}].
$$
Given that $\frac{\mu c +c^{\frac{2}{3}}+a}{a+b+c} , \frac{\mu c -c^{\frac{2}{3}}+a}{a+b+c} \to_{c \to \infty} \mu$ and $\mathcal{F}$ is continuous $\prob_{\Theta \sim \mathcal{F}}\big[\Theta + \frac{\mu c +c^{\frac{2}{3}}+a}{a+b+c} \geq p \big] , \prob_{\Theta \sim \mathcal{F}}\big[\Theta + \frac{\mu c -c^{\frac{2}{3}}+a}{a+b+c} \geq p \big] \to_{c \to \infty} \prob_{\Theta \sim \mathcal{F}}[\Theta +\mu \geq p]$. Combining this with $\lim_{c \to \infty}\prob[\mathcal{E}^{\textsc{Conc}}] = 1$ yields that $(3)$ converges to $\frac{1}{\prob[\Theta + \mu \geq p]}$ as $c \to \infty$. Thus, $\lim_{c \to \infty}  \expect_{N \sim \binomial(c, \mu)}\Big[\prob_{\Theta \sim \mathcal{F}}\big[\Theta + \frac{N +a}{a+b+c} \geq p \big] \Big] = \frac{1}{\prob[\Theta + \mu \geq p]}$, concluding the proof.\end{proof}

\section{Supplementary material for Section~\ref{sec:pricing}}\label{app:pricing}
\subsection{\textsc{CoNF} still exists under optimal dynamic pricing (Proposition~\ref{prop:conf_still_exits_opt_dynamic_pricing}) }\label{appendix_proof_CoNF_example_at_least_alpha}
\begin{proof}[Proof of \cref{prop:conf_still_exits_opt_dynamic_pricing}.]
For any $\alpha<4/3$, we provide an instance with $\chi(\Pi^{\textsc{dynamic}})  = \alpha$. The instance consists of a single review $(c =1 )$, true quality $\mu < \frac{1}{3}$, estimate mappings $h(0) = \mu^2$, $h(1) = 1-\mu^2$, and customer-specific valuation $\mathcal{F} = \mathcal{U}[0,2 \overline{h}]$ where $\overline{h} \coloneqq \expect_{N}[h(N)]$. 

First, for any  $A > 0$ and any price $p$ the revenue of selling to a customer with valuation $V = Y + x$ where $Y \sim \mathcal{U}[0,A]$ is given by $p\cdot  \prob[V \geq p] = \frac{p(A+x-p)}{A}$ if $p \in [x, A+x]$. Thus, the revenue maximizing price is $p^{\star} = \max(\frac{A+x}{2},x) = \frac{\max(x,A) + x}{2}$ yielding an optimal revenue of 
\begin{equation}\label{eq:optimal_rev_uniform_plus_x}
\max_{p} \{p\cdot  \prob[V \geq p]\} = \frac{\max(x,A) + x}{2} \cdot \frac{2A+x-\max(A,x)}{2A} =\frac{(\max(A,x) + x)^2}{4 \max(A,x)}.
\end{equation}
Theorems \ref{theorem_char_dynamic_random} and \ref{thm:most_recent_dynamic_opt_rev} imply that
$\chi(\Pi^{\textsc{dynamic}}) = \frac{\expect_{N}[r^{\star}(\Theta + h(N))]}{r^{\star}(\Theta + \overline{h})}$. Using (\ref{eq:optimal_rev_uniform_plus_x}) with $A = 2\overline{h}$, we obtain
\begin{equation}\label{eq:revenue_uniform_plus_x_cases}
r^{\star}(\Theta + x) = \frac{(\max(2\overline{h},x) + x)^2}{4 \max(2\overline{h},x)} =
\begin{cases}
    \frac{9}{8} \overline{h} & \text{if } x = \overline{h} \\
    h(1)  & \text{if } x = h(1) \\
    \frac{(h(0) + 2\overline{h})^2}{8 \overline{h}}  & \text{if } x = h(0) 
\end{cases}.
\end{equation}
where $\max(2 \overline{h}, h(1)) = h(1)$ 
because $2\mu(1-\mu)(1+2\mu) = 2\mu(1+\mu-2\mu^2) < 2\mu(1+\mu) < \frac{8}{9} < 1-\mu^2$ for $\mu < \frac{1}{3}$ and thus 
$$2\overline{h} = 2(\mu h(1) + (1-\mu)h(0)) = 2\big(\mu(1-\mu^2) + (1-\mu)\mu^2\big) = 2\mu(1-\mu)(1+2\mu) < 1-\mu^2 = h(1).$$

Using (\ref{eq:revenue_uniform_plus_x_cases}) , the CoNF can be expressed as:
    \begin{align*}
        \chi(\Pi^{\textsc{dynamic}}) &= \frac{\expect_{N}[r^{\star}(\Theta + h(N))]}{r^{\star}(\Theta + \overline{h})} = \frac{\mu r^{\star}(\Theta + h(1)) + (1-\mu)r^{\star}(\Theta + h(0))}{r^{\star}(\Theta + \overline{h})}  \\
        &\overset{(1)}{=} \frac{ \mu h(1) + (1-\mu) \frac{(h(0) + 2\overline{h})^2}{8 \overline{h}}}{\frac{9}{8}\overline{h}}= \frac{8}{9} \cdot \frac{\mu  h(1)}{\overline{h}} + \frac{1-\mu}{9} \Big[
        \Big(\frac{h(0)}{\overline{h}}\Big)^2 + 4\frac{h(0)}{\overline{h}} + 4 \Big]\\
        &\overset{(2)}{=} \frac{8}{9} -\frac{8}{9}\cdot \frac{(1-\mu)h(0)}{\overline{h}}  + \frac{1-\mu}{9} \cdot 
        \Big(\frac{h(0)}{\overline{h}}\Big)^2 + \frac{1-\mu}{9} \cdot 4\frac{h(0)}{\overline{h}} + \frac{1-\mu}{9} \cdot 4 \\
        &\overset{(3)}{=} \frac{8}{9}  + \frac{1-\mu}{9} \Big[
        \Big(\frac{h(0)}{\overline{h}}\Big)^2 - 4\frac{h(0)}{\overline{h}} + 4 \Big]= \frac{8}{9}  + \frac{1-\mu}{9} \Big(
        \frac{h(0)}{\overline{h}} - 2 \Big)^2 \\
        &\overset{(4)}{=}  \frac{8}{9}+ \frac{1-\mu}{9}(\frac{\mu}{1+\mu-2\mu^2} -2)^2 
    \end{align*}
    where (1) follows by the expression for the optimal revenue (\ref{eq:revenue_uniform_plus_x_cases}), (2) follows by $\frac{\mu  h(1)}{\overline{h}} = 1-\frac{(1-\mu ) h(0)}{\overline{h}}$ which follows by the definition of $\overline{h}$, (3) follows by $-\frac{8}{9}\cdot \frac{(1-\mu)h(0)}{\overline{h}} + \frac{1-\mu}{9} \cdot 4\frac{h(0)}{\overline{h}} = -\frac{1-\mu}{9} \cdot 4\frac{h(0)}{\overline{h}}$, and (4) follows by $\frac{h(0)}{\overline{h}} = \frac{\mu^2}{\mu(1-\mu)(1+2\mu)} = \frac{\mu}{1+\mu-2\mu^2}$ because $h(0) = \mu^2$ and $ \overline{h} = \mu (1-\mu)(1+2\mu)$. 

    Thus, $\chi(\Pi^{\textsc{dynamic}}) = \frac{8}{9}+ \frac{1-\mu}{9}(\frac{\mu}{1+\mu-2\mu^2} -2)^2 $. The second term is $\frac{1-\mu}{9}(\frac{\mu}{1+\mu-2\mu^2} -2)^2 \to \frac{4}{9}$ as $\mu \to 0$, and thus $\chi(\Pi^{\textsc{dynamic}})\to \frac{8}{9} + \frac{4}{9} = \frac{4}{3}$ as $\mu \to 0$. Therefore, for any $\alpha < 4/3$ there is some sufficiently small $\mu > 0$ such that $\chi(\Pi^{\textsc{dynamic}})  > \alpha$. 
\end{proof}

\subsection{Random is no worse than Newest under dynamic pricing (Remark~\ref{rmk:CoNF_dynamic_at_least_one})}\label{appendix_subsec:dynamic_pricing_CoNF_geq_1_proof}
 \cref{theorem:negative_bias} established that $\chi(p) > 1$ for any price $p$ satisfying Assumption \ref{assumption:non_abs_non_degen}. Here we show that if the platform optimizes over dynamic prices we have $\chi(\Pi^{\textsc{dynamic}}) \geq 1$ i.e. $\random$ has no smaller revenue that $\newest$ under optimal dynamic prices. Recall that $\chi(\Pi^{\textsc{dynamic}}) = \frac{\max_{\rho \in \Pi^{\textsc{dynamic}}}\textsc{Rev}(\sigma^{\textsc{random}},\rho)}{\max_{\rho \in \Pi^{\textsc{dynamic}}}\textsc{Rev}(\sigma^{\textsc{newest}},\rho)}$. We show the following result.

\begin{proposition}\label{lemma:CoNF_dynamic_pricing} For any problem instance, $\chi(\Pi^{\textsc{dynamic}}) \geq 1$. 
\end{proposition}

\begin{proof}
By definition of $\chi(\Pi^{\textsc{dynamic}})$, it is sufficient to show that \begin{equation}\label{eq:conf_dynamic_bound}
\max_{\rho \in \Pi^{\textsc{dynamic}}}\textsc{Rev}(\sigma^{\textsc{newest}},\rho) \leq \max_{\rho \in \Pi^{\textsc{dynamic}}}\textsc{Rev}(\sigma^{\textsc{random}},\rho).
\end{equation}
Letting $\overline{h} = \expect_{N \sim \binomial(c, \mu)}[h(N)]$ and 
$p^{\star} \in \argmax_{p \in \mathbb{R}}p  \prob_{\Theta \sim \mathcal{F}}[\Theta + \overline{h} \geq p]$, Theorem \ref{theorem:newest_first_dynamic_optimal_policy} yields that the pricing policy given by $\rho^{\textsc{newest}}(\bm{z}) = h(N_{\bm{z}}) + p^{\star}-\overline{h}$ is optimal under $\sigma^{\textsc{newest}}$ and is thus the maximizer of the left-hand side of (\ref{eq:conf_dynamic_bound}). To prove (\ref{eq:conf_dynamic_bound}) it is sufficient to show that offering prices $\rho^{\textsc{newest}}(\bm{z})$ under $\sigma^{\textsc{random}}$ yields the same revenue as offering prices $\rho^{\textsc{newest}}(\bm{z})$ under $\sigma^{\textsc{newest}}$ i.e. $\textsc{Rev}(\sigma^{\textsc{random}},\rho^{\textsc{newest}}) = \textsc{Rev}(\sigma^{\textsc{newest}},\rho^{\textsc{newest}})$. 

As $\sigma^{\textsc{random}}$ shows $c$ i.i.d. $\bern(\mu)$ reviews, the revenue of any pricing policy $\rho$ is given by \begin{equation}\label{eq:revenue_random_any_pricing_policy}
\rev(\random, \rho) = \expect_{Z_1, \ldots, Z_c \sim_{i.i.d} \bern(\mu)}\Big[\rho(Z_1, \ldots, Z_c) \prob_{\Theta \sim \mathcal{F}}[\Theta + h(\sum_{i=1}^c Z_i) \geq  \rho(Z_1, \ldots, Z_c)]\Big].
\end{equation}
For any reviews $(Z_1, \ldots, Z_c)$ the purchase probability under $\rho^{\textsc{newest}}(Z_1, \ldots, Z_c)$ equals
\begin{equation}\label{eq:expected_prob_purchase_random_newest_price}
\prob_{\Theta \sim \mathcal{F}}\Big[\Theta + h(\sum_{i=1}^c Z_i) \geq  h(\sum_{i=1}^c Z_i) + p^{\star}-\overline{h}\Big] = \prob_{\Theta \sim \mathcal{F}}[\Theta +\overline{h}\geq   p^{\star}] 
\end{equation}
and is independent of $(Z_1, \ldots, Z_c)$. As $Z_1, \ldots, Z_c \sim_{i.i.d.} \bern(\mu)$, the number of positive review ratings $\sum_{i=1}^c Z_i$ is distributed as $N \sim \binomial(c,\mu)$; the expected price of $\rho^{\textsc{newest}}(Z_1, \ldots, Z_c)$ is thus
\begin{equation}\label{eq:expected_price_random_newest_price}
\expect_{Z_1, \ldots, Z_c \sim_{i.i.d} \bern(\mu)}\Big[\rho^{\textsc{newest}}(Z_1, \ldots, Z_c)\Big] = \expect_{N \sim \binomial(c, \mu)}[h(N) + p^{\star}-\overline{h}] =p^{\star}
\end{equation}
where the last equality uses the definition of $\overline{h} = \expect_{N \sim \binomial(c, \mu)}[h(N)]$. Combining (\ref{eq:revenue_random_any_pricing_policy}), (\ref{eq:expected_prob_purchase_random_newest_price}),and (\ref{eq:expected_price_random_newest_price}), we obtain $\rev(\random, \rho) = p^{\star}\prob_{\Theta \sim \mathcal{F}}[\Theta +\overline{h}\geq p^{\star}] = r^{\star}(\Theta + \overline{h})$. The last expression is exactly equal to $\max_{\rho \in \Pi^{\textsc{dynamic}}}\textsc{Rev}(\sigma^{\textsc{newest}},\rho)$ by \cref{thm:most_recent_dynamic_opt_rev} which concludes the proof. 
\end{proof}

\subsection{Stronger CoNF upper bound under estimator upper bound (Remark~\ref{rmk_result_imporvement_CoNF_dynamic_pricing})}\label{appendix:result_imporvement_CoNF_dynamic_pricing}

\begin{proposition}
    For any instance with $h(n) \leq u, \forall n \in \{0,1, \ldots, c\}$: $\chi(\Pi^{\textsc{dynamic}}) \leq \frac{2 \prob_{\Theta \sim \mathcal{F}}[\Theta \geq -u]}{\prob_{\Theta \sim \mathcal{F}}[\Theta \geq 0]}$.
\end{proposition}

\begin{proof}
Similar to the proof of \cref{theorem:dynamic_pricing_CoNF_bound}, we define 
$\Tilde{p}_n = \overline{h} + \max(p^{\star}(\Theta + h(n))-h(n),0)$ and let $p(n) = p^{\star}(\Theta + h(n))$ for convenience.  We refine the analysis of \cref{lem:demand_ratio}, to show that for any $n \in \{0,1, \ldots, c\}$, $\textsc{Demand Ratio}(\Tilde{p}_n,n) \leq \frac{\prob_{\Theta \sim \mathcal{F}}[\Theta \geq -u]}{\prob_{\Theta \sim \mathcal{F}}[\Theta \geq 0]}$. We consider two cases:
    \begin{itemize}
        \item If $p(n) \geq h(n)$, then  $\max(p(n)-h(n),0) = p(n)-h(n)$ and thus $$ \textsc{Demand Ratio}(\Tilde{p}_n,n)=\frac{\prob_{\Theta \sim \mathcal{F}}\big[\Theta \geq p(n)-h(n) \big]}{ \prob_{\Theta \sim \mathcal{F}}\big[\Theta \geq \max(p(n)-h(n),0)\big] } = 1 \leq \frac{\prob_{\Theta \sim \mathcal{F}}[\Theta \geq -u]}{\prob_{\Theta \sim \mathcal{F}}[\Theta \geq 0]}.$$
    since $\prob_{\Theta \sim \mathcal{F}}[\Theta \geq -u] \geq \prob_{\Theta \sim \mathcal{F}}[\Theta \geq 0]$ as $u > 0$. If $p(n) < h(n)$, then 
    \item If $p(n) < h(n)$, then $\max(p(n)-h(n),0) = 0$
    \begin{align*}
     \textsc{Demand Ratio}(\Tilde{p}_n,n) &= \frac{\prob_{\Theta \sim \mathcal{F}}[\Theta \geq p(n)-h(n) ]}{ \prob_{\Theta \sim \mathcal{F}}[\Theta \geq \max(p(n)-h(n),0)] } \\
     &=\frac{\prob_{\Theta \sim \mathcal{F}}[\Theta \geq p(n)-h(n) ]}{ \prob_{\Theta \sim \mathcal{F}}[\Theta \geq 0] }  \leq  \frac{\prob_{\Theta \sim \mathcal{F}}[\Theta \geq -h(n)]}{\prob_{\Theta \sim \mathcal{F}}[\Theta \geq 0]} \leq \frac{\prob_{\Theta \sim \mathcal{F}}[\Theta \geq -a]}{\prob_{\Theta \sim \mathcal{F}}[\Theta \geq 0]}.
    \end{align*}
The first inequality uses $\prob_{\Theta \sim \mathcal{F}}[\Theta \geq p(n)-h(n) ] \leq \prob_{\Theta \sim \mathcal{F}}[\Theta \geq -h(n) ]$ as the optimal price $p(n) \geq 0$. The second inequality uses $\prob_{\Theta \sim \mathcal{F}}[\Theta \geq -h(n) ] \leq \prob_{\Theta \sim \mathcal{F}}[\Theta \geq -u ]$ as $h(n) \leq u$.
    \end{itemize}
By Lemma \ref{lem:price_ratio}, the expected price ratio is upper bounded as $\expect\Big[\textsc{Price Ratio}(\Tilde{p}_{N},N)\Big] \leq 2$. Combining this with the aforementioned bound on the demand ratio the proof follows. 
\end{proof}

\subsection{Dynamic pricing revenue of Newest First (Lemma~\ref{lemma:revenue_recent_dynamic_formula_lemma_improving_review_offsetting_main_body})}\label{appendix:dynamic_pricing_revenue_formula}
As an analogue of \cref{lemma:stationary_state_distribiton_newest_first}, let $\bm{Z}_t = (Z_{t,1}, \ldots, Z_{t,c}) \in \{0,1\}^c$ denote the process of the newest $c$ reviews. We note that $\bm{Z}_t$ is a time-homogenous Markov chain on a finite state space $\{0,1\}^c$. 

If $\bm{Z}_t$ is at state $\bm{z}=(z_1, \ldots, z_c)$ it stays at that state if there is no purchase (with probability $1-\prob_{\Theta \sim \mathcal{F}}[\Theta +h(N_{\bm{z}}) \geq \rho(\bm{z})]$). If there is a purchase with a positive review, it transitions to $(1, z_1, \ldots, z_{c-1})$ (with probability $\mu\prob_{\Theta \sim \mathcal{F}}[\Theta +h(N_{\bm{z}}) \geq \rho(\bm{z})]$).  If there is a purchase with a negative review, it transitions to $(0, z_1, \ldots, z_{c-1})$ (with probability $(1-\mu)\prob_{\Theta \sim \mathcal{F}}[\Theta +h(N_{\bm{z}}) \geq \rho(\bm{z})]$). If $\mu \in (0,1)$ and $\rho$ induces positive purchase probabilities in all states, $\bm{Z}_t$ is a single-recurrence-class Markov chain with no transient states. Thus, $\bm{Z}_t$ admits a unique stationary distribution characterized in the following lemma. Recall that for a state of $c$ reviews $\bm{z} \in \{0,1\}^c$, $N_{\bm{z}} = \sum_{i=1}^c z_i$ denotes the number of positive review ratings. 

\begin{lemma}\label{lemma_dynamic_stationary_distribution}
The stationary distribution of $\bm{Z}_t \in \{0,1\}^c$ under any dynamic pricing policy $\rho$ with positive purchase probabilities in all states is 
$$\pi_{\bm{z}} = \kappa \cdot \frac{\mu^{N_{\bm{z}}}(1-\mu)^{c-N_{\bm{z}}}}{\prob_{\Theta \sim \mathcal{F}}\Big[\Theta + h(N_{\bm{z}}) \geq \rho(\bm{z}) \Big]}  \text{ for } \bm{z} \in \{0,1\}^c, $$
where $\kappa = 1/\expect_{Y_1, \ldots, Y_c \sim_{i.i.d.} \bern(\mu)}\Bigg[\frac{1}{\prob_{\Theta \sim \mathcal{F}} \big[\Theta + h(\sum_{i=1}^c Y_i) \geq \rho(Y_1, \ldots, Y_c) \big]}\Bigg]$ is a normalizing constant.
\end{lemma}

\begin{proof}[Proof of \cref{lemma_dynamic_stationary_distribution}.]
Similar to the proof of \cref{lemma:stationary_state_distribiton_newest_first}, we invoke 
 \cref{lemma: general_theorem_markov_chains_stationary}. In the language of \cref{lemma: general_theorem_markov_chains_stationary}, $\bm{Z}_t$ corresponds to $\mathcal{M}_f$, the state space $\mathcal{S}$ to $\{0,1\}^c$, and $f$ is a function that expresses the purchase probability at a given state, i.e., $f(z_1, \ldots, z_c) = \prob_{\Theta \sim \mathcal{F}}\Big[\Theta + h(N_{\bm{z}}) \geq \rho(\bm{z})\Big]$. Note that with probability $1-f(z_1, \ldots, z_c)$, $\bm{Z}_t$ remains at the same state (as there is no purchase). 

To apply \cref{lemma: general_theorem_markov_chains_stationary}, we need to show that whenever there is  purchase, $\bm{Z}_t$ transitions according to a Markov chain with stationary distribution $\mu^{N_{\bm{z}}} (1-\mu)^{c-N_{\bm{z}}}$. Consider the Markov chain $\mathcal{M}$ which always replaces the $c$-th last review with a new $\bern(\mu)$ review. This process has stationary distribution equal to the above numerator and $\bm{Z}_t$ transitions according to $\mathcal{M}$ upon a purchase, i.e., with probability $f(z_1, \ldots, z_c)$. As a result, by \cref{lemma: general_theorem_markov_chains_stationary}, $\pi$ is a stationary distribution for $\bm{Z}_t$.  As $\bm{Z}_t$ is irreducible and aperiodic, this is the unique stationary distribution. 
\end{proof}

\begin{proof}[Proof of \cref{lemma:revenue_recent_dynamic_formula_lemma_improving_review_offsetting_main_body}.]
    By Eq.~\eqref{equation:reveneue_def} and the Ergodic theorem, we can express the revenue as 
    \begin{align*}
        \textsc{Rev}(\sigma^{\textsc{newest}}, \rho) &=\liminf_{T \to \infty} \frac{\expect\left[\sum_{t=1}^{T} \rho(Z_{t,1},  \ldots,Z_{t,c}) \prob_{\Theta \sim \mathcal{F}}\left[\Theta + h(\sum_{i=1}^c Z_{t,i}) \geq \rho(Z_{t,1},  \ldots,Z_{t,c})\right]\right]}{T} \\
        &=\sum_{(z_1, \ldots, z_c) \in \{0,1\}^c} \pi_{(z_1, \ldots, z_c)} \rho(z_1, \ldots, z_c) \prob_{\Theta \sim \mathcal{F}} \big[\Theta + h(N_{\bm{z}}) \geq \rho(z_1, \ldots, z_c) \big]  \\
        &= \kappa \cdot \sum_{(z_1, \ldots, z_c) \in \{0,1\}^c} \mu^{\sum_{i=1}^c z_i}(1-\mu)^{c-\sum_{i=1}^c z_i} \rho(z_1, \ldots, z_c)\\
        &= \kappa \cdot \expect_{Y_1, \ldots, Y_c \sim_{i.i.d.} \bern(\mu)}[\rho(Y_1, \ldots, Y_c)]\\
        &= \frac{\expect_{Y_1, \ldots, Y_c \sim_{i.i.d.} \bern(\mu)}[\rho(Y_1, \ldots, Y_c)]}{\expect_{Y_1, \ldots, Y_c \sim_{i.i.d.} \bern(\mu)}\Big[\frac{1}{\prob_{\Theta \sim \mathcal{F}} \big[\Theta + h(\sum_{i=1}^c Y_i) \geq \rho(Y_1, \ldots, Y_c) \big]}\Big]}.
    \end{align*}
The third equality applies Lemma~\ref{lemma_dynamic_stationary_distribution} and cancels 
the term $\prob_{\Theta \sim \mathcal{F}} \big[\Theta + h(N_{\bm{z}}) \geq \rho(z_1, \ldots, z_c) \big]$.
\end{proof}

\subsection{Ratio of averages is bounded by maximum ratio (Proof of Theorem~\ref{theorem:newest_first_dynamic_optimal_policy})}\label{appendix_proof_lemma_ratio_average_less_max}

\begin{lemma}\label{lemma_ratio_average_less_max_main_body}
    Let $\{\tilde{\alpha}_{i},\tilde{A}_i,\tilde{B}_i\}_{i \in \mathcal{S}}$ be such that $\tilde{\alpha}_i,\tilde{B}_i >0 $ for all $i$. Then 
    $\frac{\sum_{i\in \mathcal{S}} \tilde{\alpha}_i \tilde{A}_i}{\sum_{i \in \mathcal{S}} \tilde{\alpha}_i \tilde{B}_i } \leq \max\limits_{\substack{i \in \mathcal{S}}}  \frac{\tilde{A}_i}{\tilde{B}_i}$. Equality is achieved if and only if $\frac{\tilde{A}_i}{\tilde{B}_i} = \frac{\tilde{A}_j}{\tilde{B}_j}$ for all $i, j \in \mathcal{S}$.
\end{lemma}

\begin{proof}[Proof of \cref{lemma_ratio_average_less_max_main_body}.]
We show the inequality by contradiction and assume that $\frac{\sum_{i\in \mathcal{S}} \tilde{\alpha}_i \tilde{A}_i}{\sum_{i\in \mathcal{S}} \tilde{\alpha}_i \tilde{B}_i } >  \frac{\tilde{A}_j}{\tilde{B}_j}$
for all $j \in \mathcal{S}$. Given that the denominators are positive, this implies that for any $j \in \mathcal{S}$
\begin{equation*}
\big(\sum_{i\in \mathcal{S}} \tilde{\alpha}_i \tilde{A}_i \big) \big(\tilde{\alpha}_j \tilde{B}_j \big) > \big(\tilde{\alpha}_j \tilde{A}_j \big) \big(\sum_{i\in \mathcal{S}} \tilde{\alpha}_i \tilde{B}_i \big).
\end{equation*}
 Summing over $j \in \mathcal{S}$ we obtain the following which is a contradiction:
$$\big(\sum_{i\in \mathcal{S}} \tilde{\alpha}_i \tilde{A}_i \big) \big(\sum_{j\in \mathcal{S}} \tilde{\alpha}_j \tilde{B}_j \big) > \big(\sum_{j\in \mathcal{S}} \tilde{\alpha}_j \tilde{A}_j \big) \big(\sum_{i\in \mathcal{S}} \tilde{\alpha}_i \tilde{B}_i \big)$$
As a result $\frac{\sum_{i\in \mathcal{S}} \tilde{\alpha}_i \tilde{A}_i}{\sum_{i\in \mathcal{S}} \tilde{\alpha}_i \tilde{B}_i } \leq \max_{i\in \mathcal{S}} \frac{\tilde{A}_i}{\tilde{B}_i}$. With respect to equality, let $j \in \argmax_{k} \frac{\tilde{A}_k}{\tilde{B}_k}$. Thus, 
\begin{equation}\label{eq:max_equality}
\frac{\sum_{i\in \mathcal{S}} \tilde{\alpha}_i \tilde{A}_i}{\sum_{i\in \mathcal{S}} \tilde{\alpha}_i \tilde{B}_i } = \frac{\tilde{A}_j}{\tilde{B}_j}. 
\end{equation}
Multiplying by the denominators and rearranging the above can be rewritten as
\begin{equation}\label{eq:max_eq_rewritten}
\sum_{i\in \mathcal{S}} \tilde{\alpha}_i \tilde{B}_j \tilde{B}_i \big(\frac{\tilde{A}_i}{\tilde{B}_i} - \frac{\tilde{A}_j}{\tilde{B}_j} \big) = 0.
\end{equation}
Since $\tilde{\alpha}_i > 0$, $\tilde{B}_i > 0$ for all $i \in \mathcal{S}$, and $\frac{\tilde{A}_i}{\tilde{B}_i} \leq \frac{\tilde{A}_j}{\tilde{B}_j}$ for all $i \in \mathcal{S}$ , \eqref{eq:max_eq_rewritten} holds only if $\frac{\tilde{A}_i}{\tilde{B}_i} = \frac{\tilde{A}_j}{\tilde{B}_j}$ for all $i \in \mathcal{S}$. For the ``if" direction, suppose that $\frac{\tilde{A}_i}{\tilde{B}_i} = \frac{\tilde{A}_j}{\tilde{B}_j}$ for all $i, j$. In particular this holds, when $j \in \argmax_{k} \frac{\tilde{A}_k}{\tilde{B}_k}$ is a maximizing index and $i \in \mathcal{S}$ is an arbitrary index. This implies that  \eqref{eq:max_eq_rewritten} holds for any $j \in \argmax_{k} \frac{\tilde{A}_k}{\tilde{B}_k}$ and therefore so does \eqref{eq:max_equality}, which concludes the proof.
\end{proof}

\subsection{Characterization of optimal revenue under Newest First (Corollary~\ref{thm:most_recent_dynamic_opt_rev})}\label{appendix:characterization_optimal_dynamic_pricing_rev_newest}
\begin{proof}[Proof of \cref{thm:most_recent_dynamic_opt_rev}.]
Letting $\overline{h} = \expect_{N \sim \binomial(c, \mu)}[h(N)]$, \cref{theorem:newest_first_dynamic_optimal_policy} establishes that the optimal dynamic pricing policy is review-offsetting with offset $p^{\star}-\overline{h}$. Using the expression for the revenue of review-offsetting policies with offset $a^{\star} = p^{\star} -\overline{h}$ given by \eqref{eq:revenue_review_offsetting}:
$$\max_{\rho \in \Pi^{\textsc{dynamic}}} \rev(\newest, \rho) = (a^{\star} + \overline{h}) \prob_{\Theta \sim \mathcal{F}}[\Theta + \overline{h} \geq a^{\star}+ \overline{h}] = p^{\star} \prob_{\Theta \sim \mathcal{F}}[\Theta + \overline{h} \geq p^{\star}] = r^{\star}(\Theta + \overline{h}). 
$$
\end{proof}

\subsection{All optimal dynamic policies under Newest First (Remark~\ref{rmk_opt_dynamic_generalization})}\label{appendix_generalization_newest_first_all_optimal_dynamic_pricing_policies}

\begin{proposition}\label{theorem:all_optimal_dynamic_pricing_policies_newest}
    Suppose that the platform uses $\sigma^{\textsc{newest}}$ as the review ordering policy. Let $p_{\bm{z}} \in \argmax_{p} p \prob_{\Theta \sim \mathcal{F}}\big[\Theta +\expect_{N \sim \binomial(c, \mu)}[h(N)] \geq p\big]$ be any revenue-maximizing price for $\bm{z} \in \{0,1\}^c$. A dynamic pricing policy $\rho^{\textsc{newest}}$ is optimal if and only if it has the form
    $$\rho^{\textsc{newest}}(\bm{z}) = h(N_{\bm{z}}) + p_{\bm{z}} - \expect_{N \sim \binomial(c, \mu)}[h(N)].$$
\end{proposition}

\noindent For an arbitrary dynamic pricing policy $\rho$, letting $\overline{h} = \expect_{N \sim \binomial(c, \mu)}[h(N)]$, the optimality condition of the theorem can be rewritten as\footnote{If the customer-specific distribution $\mathcal{F}$ is strictly regular, there is a unique optimal dynamic pricing policy.}
\begin{equation}\label{eq:optimality_condition}
\rho(\bm{z}) -h(N_{\bm{z}}) + \overline{h} \in \argmax_{p} p \prob_{\Theta \sim \mathcal{F}}\big[\Theta +\overline{h} \geq p\big] \quad \text{for all $\bm{z} \in \{0,1\}^c$}.
\end{equation}  

For any policy $\rho$ and any state $\bm{z} \in \{0,1\}^c$, we define a policy $\tilde{\rho}_{\bm{z}}$ to be the review-offsetting policy with offset $a_{\bm{z}} = \rho(\bm{z}) - h(N_{\bm{z}})$. The next lemma establishes that the revenue of $\rho$ is upper-bounded by the revenue of one of the policies $\Tilde{\rho_{\bm{z}}}$ and characterizes that equality is achieved if and only if each of $\Tilde{\rho}_{\bm{z}}$ have the same revenue. 

\begin{lemma}\label{lemma:review_offsetting_for_optimality} 
    The revenue of $\rho$ is at most the highest revenue of $\Tilde{\rho}_{\bm{z}}$ over $\bm{z} \in \{0,1\}^c$, i.e.,
    $$\rev(\newest,\rho) \leq \max_{\bm{z} \in \{0,1\}^c} \rev( \newest,\Tilde{\rho}_{\bm{z}}).$$
    Equality holds if and only if $\rev( \newest,\Tilde{\rho}_{\bm{z}}) = \rev( \newest,\Tilde{\rho}_{\bm{z}'})$ for all $\bm{z}, \bm{z}' \in \{0,1\}^c$. 
\end{lemma}

To characterize all optimal dynamic pricing policies, the next lemma shows that
$\Tilde{\rho}_{\bm{z}}$ is optimal if and only if $\rho(\bm{z})-h(N_{\bm{z}}) + \overline{h}$ is a revenue maximizing price when selling to a customer with valuation $\Theta + \overline{h}$. 

\begin{lemma}\label{lemma: optimality_condition_rho_tilda}The policy $\Tilde{\rho}_{\bm{z}}$ is optimal if and only if $\rho(\bm{z})-h(N_{\bm{z}}) + \overline{h}$ is a revenue maximizing price when selling to a customer with valuation $\Theta + \overline{h}$ i.e. $\rho(\bm{z})-h(N_{\bm{z}}) + \overline{h} \in \argmax_{p} p \prob_{\Theta \sim \mathcal{F}}[\Theta + \overline{h} \geq p]$. 
\end{lemma}

\begin{proof}[Proof of \cref{theorem:all_optimal_dynamic_pricing_policies_newest}.]
    For an optimal policy $\rho$, \cref{lemma:review_offsetting_for_optimality} implies that $$\rev(\newest,\rho) \leq \max_{\bm{z} \in \{0,1\}^c} \rev(\newest,\Tilde{\rho}_{\bm{z}})$$ where equality holds if and only if $\rev( \newest,\Tilde{\rho}_{\bm{z}}) = \rev(\newest,\Tilde{\rho}_{\bm{z}'})$ for all $\bm{z}, \bm{z}' \in \{0,1\}^c$. The optimality of $\rho$ thus implies that $\Tilde{\rho}_{\bm{z}}$ must be optimal for all $\bm{z} \in \{0,1\}^c$. By \cref{lemma: optimality_condition_rho_tilda}, $\rho(\bm{z})-h(N_{\bm{z}}) + \overline{h} \in \argmax_{p} p \prob_{\Theta \sim \mathcal{F}}[\Theta + \overline{h} \geq p]$ for all $\bm{z} \in \{0,1\}^c$ i.e. (\ref{eq:optimality_condition}) holds. 

    Conversely, suppose that $\rho$ satisfies (\ref{eq:optimality_condition}). \cref{lemma: optimality_condition_rho_tilda} implies that each of $\Tilde{\rho}_{\bm{z}}$ is optimal and thus yields the same revenue, i.e., 
    $\rev(\newest, \Tilde{\rho}_{\bm{z}}) = r^{\star}(\Theta + \overline{h})$
    for all $\bm{z}$. By the equality condition of \cref{lemma:review_offsetting_for_optimality}, $\rev(\newest,\rho) = \max_{\bm{z} \in \{0,1\}^c} \rev( \newest,\Tilde{\rho}_{\bm{z}}) = r^{\star}(\Theta + \overline{h})$ and thus $\rho$ is optimal. 
\end{proof}

To prove \cref{lemma:review_offsetting_for_optimality} and \cref{lemma: optimality_condition_rho_tilda}, we show a simplified expression for the revenue of any review offsetting policy.

\begin{lemma}\label{lemma:review_offsetting_to_single_expected_value}
Let $\overline{h} = \expect_{N \sim \binomial(c, \mu)}[h(N)]$. For a review-offsetting policy $\tilde{\rho}$
with offset $a \in \mathbb{R}$:
$$\rev(\sigma^{\textsc{newest}},\Tilde{\rho}) =\big(a + \overline{h}\big) \prob_{\Theta \sim \mathcal{F}}\big[\Theta + \overline{h} \geq a + \overline{h}\big].$$
\end{lemma}

\begin{proof}[Proof of Lemma \ref{lemma:review_offsetting_for_optimality}.]
Given that policy $\Tilde{\rho}_{\bm{z}}$ is review-offsetting with offset $a_{\bf{z}}$, it holds that
\begin{equation*}
    \rev(\sigma^{\textsc{newest}}, \Tilde{\rho}_{\bm{z}}) = \underbrace{\Big(\rho(\bm{z})-h(N_{\bm{z}})+ \overline{h} \Big)}_{A_{\bm{z}}} \cdot \underbrace{\prob_{\Theta \sim \mathcal{F}}\big[\Theta + h(N_{\bm{z}}) \geq \rho(\bm{z}) \big]}_{(B_{\bm{z}})^{-1}}.
\end{equation*}
Let $\alpha_{\bm{z}} = \mu^{N_z}(1-\mu)^{c-N_{z}}$ be the probability that $c$ i.i.d. $\bern(\mu)$ trials result in $\bm{z}$. Hence $\expect_{Y_1, \ldots, Y_c \sim_{i.i.d.} \bern(\mu)}[\rho(\bm{Y})] = \expect_{Y_1, \ldots, Y_c \sim_{i.i.d.} \bern(\mu)}[\rho(\bm{Y})-h(N_{\bm{Y}}) + \overline{h}] = \sum_{\bm{z} \in \{0,1\}^c} \alpha_{\bm{z}} A_{\bm{z}}$. Applying \cref{lemma:revenue_recent_dynamic_formula_lemma_improving_review_offsetting_main_body} we can thus express the revenue of $\rho$ as
\begin{equation*}
  \textsc{Rev}(\sigma^{\textsc{newest}}, \rho) =\frac{\sum_{\bm{z} \in \{0,1\}^c} \alpha_{\bm{z}} A_{\bm{z}}}{\sum_{\bm{z} \in \{0,1\}^c}\alpha_{\bm{z}} B_{\bm{z}}} \leq \max_{\bm{z} \in \mathcal{S}} \frac{A_{\bm{z}}}{B_{\bm{z}}}=\max_{z\in \{0,1\}^c}  \rev(\sigma^{\textsc{newest}}, \Tilde{\rho}_{\bm{z}})\leq \rev(\newest,\rho^{\star})
\end{equation*}
where the first inequality follows from the following natural convexity property (\cref{lemma_ratio_average_less_max_main_body}, Appendix \ref{appendix_proof_lemma_ratio_average_less_max}): for any $\{\tilde{\alpha}_{i},\tilde{A}_i,\tilde{B}_i\}_{i \in \mathcal{S}}$ with $\tilde{\alpha}_i,\tilde{B}_i >0 $ for all $i$,  $\frac{\sum_{i\in \mathcal{S}} \tilde{\alpha}_i \tilde{A}_i}{\sum_{i \in \mathcal{S}} \tilde{\alpha}_i \tilde{B}_i } \leq \max\limits_{\substack{i \in \mathcal{S}}}  \frac{\tilde{A}_i}{\tilde{B}_i}$. By \cref{lemma_ratio_average_less_max_main_body}, equality holds if and only if $\frac{A_{\bm{z}}}{B_{\bm{z}}} = \rev(\sigma^{\textsc{newest}}, \Tilde{\rho}_{\bm{z}})$ is the same for all $\bm{z} \in \{0,1\}^c$. \end{proof}

\begin{proof}[Proof of \cref{lemma: optimality_condition_rho_tilda}.]
Given that $\Tilde{\rho}_{\bm{z}}$ is a review-offsetting policy with offset $\rho(\bm{z}) - h(N_{\bm{z}})$,  \cref{lemma:review_offsetting_to_single_expected_value} implies
\begin{equation*}\label{eq:revenue_formula_z_review_offsetting}
    \rev(\sigma^{\textsc{newest}}, \Tilde{\rho}_{\bm{z}}) = \underbrace{\Big( \overline{h}+\rho(\bm{z})-h(N_{\bm{z}}) \Big)}_{p_{\bm{z}}} \cdot \prob_{\Theta \sim \mathcal{F}}\big[\Theta + h(N_{\bm{z}}) \geq \rho(\bm{z}) \big].
    \end{equation*}
    By adding and subtracting $\overline{h}$ from each side the purchase probability can be rewritten as 
    $$\prob_{\Theta \sim \mathcal{F}}\big[\Theta + h(N_{\bm{z}}) \geq \rho(\bm{z}) \big] = \prob_{\Theta \sim \mathcal{F}}\big[\Theta + \overline{h} \geq \overline{h} +\rho(\bm{z}) - h(N_{\bm{z}}) \big] = \prob_{\Theta \sim \mathcal{F}}\big[\Theta + \overline{h} \geq p_{\bm{z}}].$$
    Thus, the revenue of $ \Tilde{\rho}_{\bm{z}}$ is equal to the revenue of offering a price of $p_{\bm{z}}$ to a customer with valuation $\Theta + \overline{h}$, which is maximized if and only $p_{\bm{z}} =\rho(\bm{z})-h(N_{\bm{z}}) + \overline{h} \in \argmax_{p} p \prob_{\Theta \sim \mathcal{F}}[\Theta + \overline{h} \geq p]$.
\end{proof}

\subsection{Uniqueness of optimal dynamic pricing for Newest and Random (Remark~\ref{rmk_opt_dynamic_generalization})}\label{appendix: unique_optimal_dynamic_pricing}
We show that the optimal dynamic pricing policies under $\newest$ and $\random$ are unique assuming mild regularity conditions on the customer-specific value distribution $\mathcal{F}$. For a continuous random variable $V$ with bounded support, let $\overline{F}_{V}(p) = \prob[V \geq p]$ be its survival function. As $V$ has a continuous distribution, the inverse survival function $G_{V}(q) = (\overline{F}_{V})^{-1}(q)$ is well-defined for $q \in (0,1)$, i.e., for any quantile $q \in (0,1)$, $G_{V}(q)$ is the unique price which induces a purchase probability $q$. Let $r(q) \coloneqq q G_{V}(q)$ be the revenue as a function of the quantile $q$. We require a notion of strict regularity.

\begin{definition}[Strict Regularity]
    A random variable $V$ has a strictly regular distribution if the revenue function $r(q) = q G_{V}(q) $ is strictly concave in the quantile $q \in (0,1)$. 
\end{definition}

In order to show uniqueness of the optimal dynamic pricing policies, we extend the strict regularity assumption by imposing a few mild further conditions.  

\begin{definition}[Well-behavedness]\label{def:extended_strict_regularity}
    A random variable $V$ is \textit{well-behaved} if:  
        (1) $V$ is continuous with bounded support; (2) $V$ is strictly regular; (3) $\prob[V > 0] > 0$. 
\end{definition}

Condition (3) implies that we can achieve a strictly positive revenue by selling to a customer with a well-behaved valuation $V$. 

\begin{lemma}\label{lemma:V_conditions_imply_conditions_translation}
For any $x \geq 0$ and well-behaved random variable $V$, $V_{x} \coloneqq V + x$ is also well-behaved.
\end{lemma}
\begin{proof}[Proof of \cref{lemma:V_conditions_imply_conditions_translation}.]
To prove that $V_x$ is well-behaved we establish all three conditions of \cref{def:extended_strict_regularity}. First, the continuity and boundedness of $V$ imply that $V_x= V + x$ also has these properties. Second, $qG_{V}(q)$ is strictly concave and $qx$ is linear in $q$; thus $qG_{V+x}(q) = qG_{V}(q) + qx$ is strictly concave. Third, as $\prob[V > 0]$, adding a non-negative scalar $x$ yields $\prob[V_x > 0]$. 
\end{proof}

\begin{lemma}\label{lemma:unique_myerson_price_is_stricly_regular_and_conditions}
    For a well-behaved random variable $\Tilde{V}$, the revenue-maximizing price $p^{\star}$ is unique.
\end{lemma}

\begin{proof}[Proof of \cref{lemma:unique_myerson_price_is_stricly_regular_and_conditions}.]
Letting the support of $\Tilde{V}$ be $[\underline{v}, \overline{v}]$ implies that $G_{\Tilde{V}}(q) \in (\underline{v}, \overline{v})$ for all $q \in (0,1)$, $\lim_{q \to 0} G_{\Tilde{V}}(q) = \overline{v}$, and $\lim_{q \to 1} G_{\Tilde{V}}(q) = \underline{v}$. Thus, $$\lim_{q \to 0} r(q)= \lim_{q \to 0} q G_{\Tilde{V}}(q) = \lim_{q \to 0} G_{\Tilde{V}}(q) \cdot \lim_{q \to 0} q = 0.$$

By the third condition of well-behavedness, $\prob_{\Tilde{V}}[\Tilde{V} > 0]$ which implies that $\overline{v} > 0$ and thus there exists some quantile $\Tilde{q}$ such that $r(\Tilde{q}) > 0$. In particular, for  $\Tilde{q} = \frac{\overline{v}-\max(\underline{v}, 0)}{2(\overline{v}-\underline{v})}$, $r(\Tilde{q})= \Tilde{q}G_{\Tilde{V}}(\Tilde{q}) >0$ as $G_{\Tilde{V}}(\Tilde{q}) = \frac{\overline{v} + \max(0, \underline{v})}{2} > 0$. Since $\lim_{q \to 0} r(q) = 0$, $r(\Tilde{q}) > 0$, and $r(q)$ is strictly concave on $(0,1)$, it holds that either (a) $r(q)$ has a unique maximizer at $q = q^{\star} \in (0,1)$ or (b) $r(q)$ is strictly increasing for $q \in (0,1)$. We consider these two cases separately below. 

For case (a), the optimal price $p^{\star}$ has a quantile $q^{\star} = \prob_{\Tilde{V}}[\Tilde{V} \geq p^{\star}] \in (0,1)$ such that $r(q^{\star}) > r(q)$ for all $q \in (0,1) \setminus \{q^{\star}\}$. The strict concavity of $r$ implies \begin{equation}\label{ineq:r_greater_than_endpoints}
        r(q^{\star}) > \max(\lim_{q \to 1} r(q), \lim_{q \to 0} r(q)) = \max(\underline{v}, 0).
    \end{equation} We show that $p^{\star} = G_{\Tilde{V}}(q^{\star})$ is the unique revenue-maximizing price. It suffices to show that this price provides strictly higher revenue than any other price $p \neq p^{\star}$, i.e., $r(q^{\star}) = p^{\star} \prob_{\tilde{V}}[\tilde{V} \geq p^{\star}] > p \prob_{\tilde{V}}[\tilde{V} \geq p]$. Let $q = \prob_{\tilde{V}}[\tilde{V} \geq p]$ be the quantile associated with the price $p$.
\begin{itemize}
    \item If $p \in (\underline{v}, \overline{v})$ the continuity of $\Tilde{V}$ implies that $q= \prob_{\tilde{V}}[\tilde{V} \geq p] \in (0,1)$ and thus the revenue of $p^{\star}$ is strictly larger than the revenue of $p$ by the assumption of case (a) that $r(q^{\star}) > r(q)$ for all $q \in (0,1) \setminus \{q^{\star}\}$. 
    \item If $p \geq \overline{v}$, then $ p \prob_{\tilde{V}}[\tilde{V} \geq p] = 0$ which combined with $r(q^{\star}) > 0$ yields that the revenue of $p^{\star}$ is strictly larger than the revenue of $p$.
    \item Lastly, if $p \leq \underline{v}$, then $p \prob_{\tilde{V}}[\tilde{V} \geq p]  \leq \underline{v} < r(q^{\star})$ by \eqref{ineq:r_greater_than_endpoints}. 
\end{itemize}
For case (b), $r(q)$ is strictly increasing for $q \in (0,1)$. Let $p^{\star} = \underline{v}$. Since $r(1) =\overline{v} $ , $r(0) = 0$, and $r(1) > r(0)$ we obtain that $\underline{v} > 0$. To show that $p^{\star} = \underline{v}$ is the unique revenue-maximizing price, it suffices to prove that $\overline{v} = p^{\star} \prob_{\tilde{V}}[\tilde{V} \geq p^{\star}] > p \prob_{\tilde{V}}[\tilde{V} \geq p] $ for any $p \neq p^{\star}$. 
\begin{itemize}
    \item If $p \in (\underline{v}, \overline{v})$ the continuity of $\Tilde{V}$ implies that $q= \prob_{\tilde{V}}[\tilde{V} \geq p] \in (0,1)$ and thus the revenue of $p^{\star}$ is strictly larger than the revenue of $p$ by the assumption of case (b) that $r(q)$ is increasing on $(0,1)$.
    \item If $p \geq \overline{v}$ 
    then $ \prob_{\tilde{V}}[\tilde{V} \geq p] =0$ and thus $p \prob_{\tilde{V}}[\tilde{V} \geq p] = 0$ which is strictly smaller than $p^{\star} = \underline{v}$.
    \item  Lastly if $p < \underline{v}$, then $p \prob_{\tilde{V}}[\tilde{V} \geq p]  \leq p < \underline{v} = p^{\star} = p^{\star} \prob_{\tilde{V}}[\tilde{V} \geq p^{\star}]$.
\end{itemize}
As a result, in both cases, we established that there exists a unique revenue-maximizing $p^{\star}$. 
\end{proof}

\begin{proposition}\label{prop:unique_opt_dynamic_pricing_newest_random}
For a well-behaved customer-specific valuation $\Theta$, the optimal dynamic pricing policies under $\newest$ and $\random$ are unique.
\end{proposition}

\begin{proof}
The valuation of a customer at a state with review ratings $\bm{z}$ is $\Theta + h(N_{\bm{z}})$. By \cref{theorem_char_dynamic_random}, any optimal dynamic pricing policy under $\random$ outputs, at each state of review ratings $\bm{z}$, a revenue-maximizing price for a customer with valuation $\Theta + h(N_{\bm{z}})$. Combining \cref{lemma:V_conditions_imply_conditions_translation} with $x=h(N_{\bm{z}})$ and \cref{lemma:unique_myerson_price_is_stricly_regular_and_conditions}, this price is unique for every state of review ratings $\bm{z}$.

By \cref{theorem:all_optimal_dynamic_pricing_policies_newest}, any optimal dynamic pricing policy under $\newest$ outputs, at each state of review ratings $\bm{z}$, a price of $\rho^{\textsc{newest}}(\bm{z}) = h(N_{\bm{z}}) + p_{\bm{z}} - \overline{h}$ where $p_{\bm{z}}$ is a revenue-maximizing price for a customer with valuation $\Theta + \overline{h}$. Combining \cref{lemma:V_conditions_imply_conditions_translation} with $x=\overline{h}$ and \cref{lemma:unique_myerson_price_is_stricly_regular_and_conditions}, the price $\rho^{\textsc{newest}}(\bm{z})$ is unique for every state of review ratings $\bm{z}$.
\end{proof}

\subsection{Characterization of optimal dynamic pricing under Random (Proposition~\ref{theorem_char_dynamic_random})}\label{app:proof_random_optimal_pricing}

\begin{proof}[Proof of \cref{theorem_char_dynamic_random}.]
By definition, $\sigma^{\textsc{random}}$ displays $c$ i.i.d. $\bern(\mu)$ reviews at every round~$t$. If the displayed reviews are $\bm{z} = (z_1, \ldots, z_{c}) \in \{0,1\}^{c}$, the revenue obtained by offering price $\rho(\bm{z})$ is equal to 
$$\rho(\bm{z}) \prob_{\Theta \sim \mathcal{F}}\big[\Theta + h(N_{\bm{z}}) \geq \rho(\bm{z}) \big].$$
Thus the optimal price is any revenue-maximizing price $\rho^{\star}(\bm{z}) \in \argmax_{p \in \mathbb{R}}p \prob_{\Theta \sim \mathcal{F}}\big[\Theta +h(N_{\bm{z}}) \geq p\big] $. Since the number of positive reviews $N_{\bm{z}}$ is distributed as $\binomial(c, \mu)$, the optimal revenue is equal to 
$\expect_{N \sim \binomial(c,\mu)}\Big[\max_{p \in \mathbb{R}}p \prob[\Theta +h(N) \geq p] \Big] = \expect_{N \sim \binomial(c,\mu)}\Big[r^{\star}(\Theta + h(N)) \Big] $.
\end{proof}

\subsection{Comparing optimal pricing for Newest and Random (Proposition~\ref{prop: comparison_prices_dynamic_newest_random_main_body})}\label{appendix:comparison_opt_dynamic_newest_random}

We compare the optimal dynamic pricing policies under $\newest$ and $\random$ assuming that the customer-specific valuation $\Theta$ is well-behaved (Definition \ref{def:extended_strict_regularity}, Appendix \ref{appendix: unique_optimal_dynamic_pricing}). This is a sufficient condition for the optimal dynamic pricing policies under $\newest$ and $\random$ to be unique (\cref{prop:unique_opt_dynamic_pricing_newest_random}, Appendix \ref{appendix: unique_optimal_dynamic_pricing}). We restate \cref{prop: comparison_prices_dynamic_newest_random_main_body} (with the exact \emph{mild regularity conditions}).
\begin{proposition}\label{prop: comparison_prices_dynamic_newest_random} Let $\overline{h} =  \expect_{N \sim \binomial(c, \mu)}[h(N)]$. For any well-behaved customer-specific valuation~$\Theta$, the unique dynamic pricing policies $\rho^{\textsc{newest}}$ and $\rho^{\textsc{random}}$ satisfy:
    \begin{itemize}
        \item $\rho^{\textsc{newest}}(\bm{z}) \geq \rho^{\textsc{random}}(\bm{z})$ for review states $\bm{z} \in \{0,1\}^c$ with  $h(N_{\bm{z}}) > \overline{h}$
        \item $\rho^{\textsc{newest}}(\bm{z}) \leq \rho^{\textsc{random}}(\bm{z})$ for review states $\bm{z} \in \{0,1\}^c$ with $h(N_{\bm{z}}) < \overline{h}$
        \item $\rho^{\textsc{newest}}(\bm{z}) = \rho^{\textsc{random}}(\bm{z})$ review states $\bm{z} \in \{0,1\}^c$ with $h(N_{\bm{z}}) = \overline{h}$.
    \end{itemize}
\end{proposition}

To prove \cref{prop: comparison_prices_dynamic_newest_random}, we introduce the revenue maximizing price of selling to a customer with valuation $\Theta + x$ for a non-negative scalar $x \geq 0$ i.e. 
$p^{\star}(\Theta + x) = \argmax_{p \in \mathbb{R}} p \prob_{\Theta \sim \mathcal{F}}[\Theta + x \geq p]$. 

We also introduce the function  $g(x) \coloneqq p^{\star}(\Theta + x)-x$ which intuitively captures the smallest idiosyncratic valuation a customer needs to have to purchase the product and is also a proxy for the purchase probability of a customer with valuation $\Theta + x$ (since $\prob_{\Theta \sim \mathcal{F}}[\Theta +x\geq p] =  \prob_{\Theta \sim \mathcal{F}}[\Theta \geq p-x]$). 

\begin{lemma}\label{lemma:g_monotonically_decreasing}
   For any well-behaved customer-specific value distribution $\Theta$, $g(x)$ is weakly decreasing for $x \geq 0$.
\end{lemma}

\begin{proof}[Proof of Proposition \ref{prop: comparison_prices_dynamic_newest_random}.]
\cref{theorem_char_dynamic_random} implies that the optimal dynamic pricing policy under $\random$ is $\rho^{\textsc{random}}(\bm{z}) = p^{\star}(\Theta + h(N_{\bm{z}}))$ for $\bm{z} \in \{0,1\}^c$. \cref{theorem:newest_first_dynamic_optimal_policy} implies that the optimal dynamic pricing policy under $\newest$ is $\rho^{\textsc{newest}}(\bm{z}) = h(N_{\bm{z}}) + p^{\star}(\Theta + \overline{h}) - \overline{h}$. As a result:
$$\rho^{\textsc{newest}}(\bm{z}) -\rho^{\textsc{random}}(\bm{z}) = (p^{\star}(\Theta + \overline{h})-\overline{h}) -(p^{\star}(\Theta + h(N_{\bm{z}}))-h(N_{\bm{z}}))=  g(\overline{h})-g(h(N_{\bm{z}})).
$$ 
By \cref{lemma:g_monotonically_decreasing}, $g$ is monotonically decreasing which implies the result of the proposition. 
\end{proof}

To prove \cref{lemma:g_monotonically_decreasing}, we use an auxiliary lemma which shows that for large enough $x$, the revenue-maximizing price $p^{\star}(\Theta + x)$ is equal to $\underline{\theta} + x$ (and thus induces a purchase probability of one). Given that $\Theta$ is well-behaved, we let its support be $[\underline{\theta}, \overline{\theta}]$. 

\begin{lemma}\label{lemma:curoff_optimal_quantile_one}
    For any well-behaved customer-specific value distribution $\Theta$ and $x \geq 0$, there exists some threshold $M \geq 0$ such that $p^{\star}(\Theta + x) \in (\underline{\theta} + x,\overline{\theta} + x)$ for $x \in [0,M)$ and $p^{\star}(\Theta + x) = \underline{\theta} + x$ for $x \geq M$. 
\end{lemma}

\begin{proof}[Proof of \cref{lemma:curoff_optimal_quantile_one}.]
First observe that $p^{\star}(\Theta +x) \geq \underline{\theta} + x$ since setting a price below $\underline{\theta}+x$ is never optimal. Indeed, prices $p < \underline{\theta}+x$ induce a purchase probability of one. Setting a price of $p + \epsilon$ for a small enough $\epsilon$ still induces a purchase probability of one and achieves a strictly higher revenue.

Second, observe that any price $p \geq \overline{\theta} + x$ induces a purchase probability of zero and thus zero revenue. As $\Theta$ is well-behaved, then so is $\Theta + x$ (\cref{lemma:V_conditions_imply_conditions_translation}). This implies that $\prob_{\Theta \sim \mathcal{F}}[\Theta + x > 0] > 0$ and thus the optimal revenue is strictly positive (as one can find a price $p> 0$ such that  $\prob_{\Theta \sim \mathcal{F}}[\Theta + x \geq p]> 0$). Thus, $p^{\star}(\Theta + x) < \overline{\theta} + x$ for all $x \geq 0$. 

Finally, let $S = \{ x \geq 0| p^{\star}(\Theta +x) = \underline{\theta} + x\}$ be the set of increments $x \geq 0$ such that the revenue-maximizing price of selling to a customer with valuation $\Theta +x$ is equal to $\underline{\theta} + x$. Observe that setting a price of $ \underline{\theta}+x$ induces a purchase probability of one and thus a revenue of $\underline{\theta}+x$. Thus, $x \in S$ if and only if the optimal revenue of selling to a customer with valuation $\Theta + x$ is $\underline{\theta} + x$. Given that $p^{\star}(\Theta + x) \in [\underline{\theta} + x, \overline{\theta} + x)$, to conclude the proof it is sufficient to show that there exists some $M \geq 0$ such that $S = [M, +\infty)$. To prove this, it suffices to show that $S$ satisfies two properties: (a) If $x \in S$ then $x' \in S$ for all $x' \geq x$; (b) If $\{x_i\}_{i=1}^{\infty}$ is a decreasing sequence with $x_i \in S$ and $\lim_{i \to \infty} x_i = x_{\infty}$, then $x_{\infty} \in S$.

For an increment $x$ and price $p$, let $R(x,p) = p \prob_{\Theta \sim \mathcal{F}}[\Theta + x \geq p]$ be the revenue of offering a price $p$ to a customer with valuation $\Theta + x$. For property (a), let $x \in S$ and $x' \geq x$. To show that $x' \in S$, it is sufficient to show $R(x',p) \leq \underline{\theta} + x'$ for all prices $p$. Expanding we have
 \begin{align*}
        R(x',p) &= p \prob_{\Theta \sim \mathcal{F}}[\Theta + x' \geq p] = \big(p-(x'-x) + (x'-x)\big) \prob_{\Theta \sim \mathcal{F}}[\Theta + x + (x'-x) \geq p]\\
        &= \big(p-(x'-x)\big)\prob_{\Theta \sim \mathcal{F}}[\Theta + x \geq p-(x'-x)] + (x'-x) \prob_{\Theta \sim \mathcal{F}}[\Theta + x' \geq p]\\
        &= R(x, p-(x'-x)) + (x'-x) \prob_{\Theta \sim \mathcal{F}}[\Theta + x' \geq p]\\
        &\leq x + \underline{\theta} + (x'-x) =  \underline{\theta} + x' \Rightarrow x' \in S
\end{align*}
where the last inequality uses $R(x, p-(x'-x)) \leq  \underline{\theta}+x$ since 
 $x \in S$ and $\prob_{\Theta \sim \mathcal{F}}[\Theta + x' \geq p] \leq 1$.

For property (b), let $\{x_i\}_{i=1}^{\infty}$ be a decreasing sequence with $x_i \in S$ and such that $\lim_{i \to \infty} x_i = x_{\infty}$. We show that $x_{\infty} \in S$. The continuity of $\Theta$ implies that revenue function $R(x,p)$ is continuous in $x$ for any fixed price $p$. Combining this with the fact that $R(x_i,p) \leq \underline{\theta}+x_i$ for all $i$, we obtain 
 $$R(x_{\infty}, p) = \lim_{i \to \infty} R(x_i, p) \leq  \lim_{i \to \infty} \underline{\theta} +x_i=  \underline{\theta}+x_{\infty} \Rightarrow x_{\infty} \in S. $$
 \end{proof}

\begin{proof}[Proof of \cref{lemma:g_monotonically_decreasing}.]
\cref{lemma:curoff_optimal_quantile_one} shows that $p^{\star}(\Theta + x) \in (\underline{\theta} + x, \overline{\theta} + x)$ for $x \in [0,M)$ and $p^{\star}(\Theta + x) = \underline{\theta} + x$ for $x \geq M$ for some threshold $M \geq 0$. Thus, $g(x) \in ( \underline{\theta}, \overline{\theta})$ for $x \in [0,M)$ and $g(x) = \underline{\theta}$ for $x \in [M, + \infty)$. It thus suffices to show that $g(x)$ is monotonically decreasing for $x \in [0,M)$. The continuity of $\Theta$ implies that the function $p \prob_{\Theta \sim \mathcal{F}}[\Theta + x \geq p]$ is differentiable for $p \in (\underline{\theta}+x, \overline{\theta}+x)$ and thus the first-order conditions must be satisfied at $p = p^{\star}(\Theta + x)\in (\underline{\theta} + x, \overline{\theta} + x)$. Denoting the survival and density functions of $\Theta$ by $\overline{F}$ and $f$ respectively, this yields
$$\frac{d}{dp}p \prob_{\Theta \sim \mathcal{F}}[\Theta + x \geq p] = \frac{d}{dp}p \overline{F}(p-x) = \overline{F}(p-x) - p f(p-x) = 0$$
for $p = p^{\star}(\Theta + x)$. Rearranging the last equation yields
    \begin{equation}\label{equation:strict_regularity_implied}
     p^{\star}(\Theta + x)-x-\frac{\overline{F}(p^{\star}(\Theta + x)-x)}{f(p^{\star}(\Theta + x)-x)} = -x \quad \iff \quad g(x)-\frac{\overline{F}(g(x))}{f(g(x))} = -x
    \end{equation}
    for any $x \in [0,M)$. The strict regularity of $\Theta$ implies that the function $u \to u - \frac{\overline{F}(u)}{f(u)}$ is strictly increasing. Thus, the left-hand side of \eqref{equation:strict_regularity_implied} is increasing in $u = g(x)$ while the right-hand side is strictly decreasing in $x$. 
    Hence, $u = g(x)$ is decreasing for $x \in [0,M)$ concluding the proof. 
\end{proof}

\subsection{Platforms unaware of state-dependent behavior (Corollary~\ref{thm:dynamic_vs_static})}\label{appendix:what_if_platforms_are_not_aware}

\begin{proof}[Proof of \cref{thm:dynamic_vs_static}.]
Suppose that $\mathcal{F}$ has support on $[0, \overline{\theta}]$. Then \cref{thm:CoNF_opt_static_arbitrarily_bad} gives \begin{equation}\label{ineq:CoNF_static_general_lower_bound_broader}
\rev(\sigma^{\textsc{random}},\Pi^{\textsc{static}})\geq \frac{\mu^c h(c)}{h(0) + \overline{\theta}} \rev(\sigma^{\textsc{newest}},\Pi^{\textsc{static}}).
\end{equation} 
Since $\mathcal{F}$ is non-negative, Theorem \ref{theorem:dynamic_pricing_CoNF_bound} gives \begin{equation}\label{ineq:CoNF_dynamic_general_upper_bound_broader}
    \rev(\sigma^{\textsc{newest}},\Pi^{\textsc{dynamic}}) \geq \frac{1}{2} \rev(\sigma^{\textsc{random}},\Pi^{\textsc{dynamic}}).
\end{equation}
Combining (\ref{ineq:CoNF_static_general_lower_bound_broader}) and (\ref{ineq:CoNF_dynamic_general_upper_bound_broader}) with the fact that $\rev(\sigma^{\textsc{random}},\Pi^{\textsc{dynamic}}) \geq \rev(\sigma^{\textsc{random}},\Pi^{\textsc{static}})$,
$$\frac{\rev(\sigma^{\textsc{newest}},\Pi^{\textsc{dynamic}})}{\rev(\sigma^{\textsc{newest}},\Pi^{\textsc{static}})} \geq \frac{\mu^c h(c)}{2(h(0) + \overline{\theta})} .$$
For any $M > 0$ there exist $\epsilon(M) > 0$ such that when $\overline{\theta}, h(0) < \epsilon(M)$, $\frac{\rev(\sigma^{\textsc{newest}},\Pi^{\textsc{dynamic}})}{\rev(\sigma^{\textsc{newest}},\Pi^{\textsc{static}})} > M$. 
\end{proof}

\section{Supplementary material for Section~\ref{sec:dynamic_quality}}\label{app:dynamic_quality}
\subsection{Random yields higher revenue than Newest even with non-stationarity (Proposition \ref{prop: random_better_than_newest_nonstationarity})}\label{appendix_subsec:proof_prop_random_better_than_newest_nonstationarity}
Our proof relies on analyzing the Markov Chain $M_t = (r_t, \mu^{(t)})$ where $r_t$ is the newest review the customer sees at round $t$ and $\mu^{(t)}$ is the product quality at round $t$, under the $\newest$ review ordering policy. Let $\pi^{\textsc{newest}}(r, \mu;\xi)$ be the stationary distribution of the process $M_t$ at a review $r \in \{0,1\}$ and product quality $\mu \in \{\mu_L, \mu_H\}$ when the probability of change is $\xi$. We use the following lemma which shows that $\newest$ spends strictly more time at a negative review than $\random$.

\begin{lemma}\label{lemma:newest_more_time_negative_than_random_nonstationarity}
    For any product qualities $\mu_L$ and $\mu_H$ satisfying $\mu_L < \mu_H$, probability of change $\xi \in (0,1]$, and price $p$ satisfying Assumption \ref{assumption:non_abs_non_degen}, $\newest$ spends more time at a negative review than $\random$, i.e.,    $$\pi^{\textsc{newest}}(0, \mu_L;\xi) + \pi^{\textsc{newest}}(0, \mu_H;\xi) >\pi^{\textsc{random}}(0).$$
\end{lemma}

\begin{proof}[Proof of  \cref{prop: random_better_than_newest_nonstationarity}]
The revenue of $\newest$ is given by 
\begin{align*}
        \rev(\newest, p) &=  p \cdot \big[q_0 (\pi^{\textsc{newest}}(0, \mu_L;\xi) + \pi^{\textsc{newest}}(0, \mu_H;\xi)) \\
        &+ q_1 (\pi^{\textsc{newest}}(1, \mu_L;\xi) + \pi^{\textsc{newest}}(1, \mu_H;\xi)) \big].
\end{align*}
Letting $\pi^{\textsc{random}}(0) = 1-\frac{\mu_L + \mu_H}{2}$ and $\pi^{\textsc{random}}(1) = \frac{\mu_L + \mu_H}{2}$ be the probabilities of a positive and a negative review under $\random$ respectively, the revenue of $\random$ is given by
\begin{align*}
        \rev(\random, p) = p \cdot \big[ q_0 \pi^{\textsc{random}}(0) + q_1 \pi^{\textsc{random}}(1) \big].
\end{align*}
\cref{lemma:newest_more_time_negative_than_random_nonstationarity} yields that $\newest$ spends more time at a negative review than $\random$, i.e., $$\pi^{\textsc{newest}}(0, \mu_L;\xi) + \pi^{\textsc{newest}}(0, \mu_H;\xi) > \pi^{\textsc{random}}(0).$$ Combined with the fact that $q_1 > q_0$ (as $p$ is non-degenerate) the proof follows.
\end{proof}

To prove \cref{lemma:newest_more_time_negative_than_random_nonstationarity}, we first prove two auxiliary lemmas about the probability of a negative review under $\newest$ for a probability of change $\xi$ by $\pi^{\textsc{newest}}(0; \xi)$, i.e., $$\pi^{\textsc{newest}}(0;\xi) \coloneqq \pi^{\textsc{newest}}(0, \mu_L;\xi) + \pi^{\textsc{newest}}(0, \mu_H;\xi).$$
\cref{lemma:non-stationarity_newest_increasing_time_negative} shows that $\pi^{\textsc{newest}}(0;\xi)$ is non-decreasing in the switching probability $\xi$. \cref{lemma:non-stationarity_newest_q=0_greater} shows that as the switching probability $\xi$ approaches $0$, $\newest$ spends more time at a negative review than $\random$.

\begin{lemma}\label{lemma:non-stationarity_newest_increasing_time_negative}
For any product qualities $\mu_L$ and $\mu_H$ satisfying $\mu_L < \mu_H$, and any price $p$ satisfying Assumption \ref{assumption:non_abs_non_degen}, the probability of a negative review $\pi^{\textsc{newest}}(0;\xi)$ is strictly increasing in the probability of change $\xi$.
\end{lemma}

\begin{lemma}\label{lemma:non-stationarity_newest_q=0_greater}
For any product qualities $\mu_L$ and $\mu_H$ satisfying $\mu_L < \mu_H$, and any price $p$ satisfying Assumption \ref{assumption:non_abs_non_degen}, $\newest$ spends strictly more time at a negative review than $\random$ as the probability of change $\xi$ converges to zero, i. e., $\lim_{\xi \to 0} \pi^{\textsc{newest}}(0;\xi)> \pi^{\textsc{random}}(0)$.
\end{lemma}

\begin{proof}[Proof of \cref{lemma:newest_more_time_negative_than_random_nonstationarity}.]
The proof follows directly from \cref{lemma:non-stationarity_newest_increasing_time_negative} and \cref{lemma:non-stationarity_newest_q=0_greater}.
\end{proof}

To prove \cref{lemma:non-stationarity_newest_increasing_time_negative} and \cref{lemma:non-stationarity_newest_q=0_greater}, we characterize the stationary distribution of the Markov Chain $M_t = (r_t, \mu^{(t)})$ consisting of the newest review $r_t$ and the current product quality $\mu^{(t)}$ at round $t$ under the review ordering policy $\newest$. To state the characterization it is useful to define the quantities $A_H = 1- ((1-\mu_H)q_1 + \mu_H q_0)$ and $A_L = 1-((1-\mu_L)q_1 + \mu_L q_0)$.

\begin{lemma}\label{lemma:stationary_dist_non-stationary_quality}
    For any price $p$ satisfying Assumption \ref{assumption:non_abs_non_degen}, the stationary distribution of the Markov Chain $M_t$ at the states with a negative review is given by 
    \begin{align*}\pi^{\textsc{newest}}(0, \mu_H;\xi) &= \frac{q_1 \big[2(1-\mu_H)A_L (\xi-1) + (1-\mu_L)\xi +(1-\mu_H)(2-\xi) \big]}{(2-(2-\xi)A_H)(2-(2-\xi)A_L)-\xi^2 A_H A_L} \quad \text{and}\\
    \pi^{\textsc{newest}}(0, \mu_L;\xi) &= \frac{q_1 \big[2(1-\mu_L)A_H (\xi-1) + (1-\mu_H)\xi + (1-\mu_L)(2-\xi)\big]}{(2-(2-\xi)A_H)(2-(2-\xi)A_L)-\xi^2 A_H A_L}
   \end{align*}
\end{lemma}

\begin{proof}[Proof of \cref{lemma:non-stationarity_newest_increasing_time_negative}]
Using \cref{lemma:stationary_dist_non-stationary_quality}, we can rewrite 
    \begin{align*}
         \frac{\pi^{\textsc{newest}}(0; \xi)}{q_1} = \frac{((1-\mu_L)A_H + (1-\mu_H)A_L) \xi + (1-\mu_L)(1-A_H) + (1-\mu_H)(1-A_L)}{((1-A_H)A_L + (1-A_L)A_H)\xi + 2(1-A_H)(1-A_L)}
    \end{align*}
We will use the fact that a rational function $\xi \to \frac{A \xi + B}{C \xi + D}$ is strictly increasing if and only if $AD > BC$. To conclude the proof it suffices to show that this fact holds true for $A = (1-\mu_L)A_H + (1-\mu_H)A_L$, $B = (1-\mu_L)(1-A_H) + (1-\mu_H)(1-A_L)$, $C = (1-A_H)A_L + (1-A_L)A_H$, and $D = 2(1-A_H)(1-A_L)$. Expanding the product $AD$ yields
\begin{align}\label{eq:equation_expanded_ad}
    AD &= 2((1-\mu_L)A_H + (1-\mu_H)A_L)(1-A_H)(1-A_L) \nonumber\\
    &= 2(1-\mu_L) (1-A_H)(1-A_L)A_H+ 2(1-\mu_H)(1-A_L)(1-A_H)A_L.
\end{align}
Expanding the product $BC$ yields 
\begin{align}
    BC &= ((1-\mu_L)(1-A_H) + (1-\mu_H)(1-A_L))((1-A_H)A_L + (1-A_L)A_H) \nonumber\\
    &= (1-\mu_L)(1-A_H)^2 A_L + (1-\mu_H)(1-A_L)^2A_H \nonumber\\
        &+ (1-\mu_L)(1-A_H)(1-A_L)A_H  + (1-\mu_H)(1-A_L)(1-A_H)A_L.
\end{align}
Therefore, canceling the terms $(1-\mu_L)(1-A_H)(1-A_L)A_H$ and $(1-\mu_H)(1-A_L)(1-A_H)A_L$, the inequality $AD > BC$ is equivalent to showing 
\begin{align*}
&(1-\mu_L) (1-A_H)(1-A_L)A_H+ (1-\mu_H)(1-A_L)(1-A_H)A_L \\
&> (1-\mu_L)(1-A_H)^2 A_L + (1-\mu_H)(1-A_L)^2A_H \\
         \iff &(1-\mu_L)(1-A_H)((1-A_L)A_H-(1-A_H)A_L) \\
         &+ (1-\mu_H)(1-A_L)((1-A_H)A_L-(1-A_L)A_H) > 0\\
         \iff &(1-\mu_L)(1-A_H)(A_H-A_L) + (1-\mu_H)(1-A_L)(A_L-A_H) > 0 \text{  (dividing by $A_H -A_L > 0$)}\\
         \underset{(1)}{\iff}
 &(1-\mu_L)(1-A_H) > (1-\mu_H)(1-A_L) \text{ (substituting the expressions for $A_H$ and $A_L$)} \\
         \iff &(1-\mu_L)((1-\mu_H)q_1 + \mu_H q_0) > (1-\mu_H)((1-\mu_L)q_1 + \mu_L q_0) \\
         \iff &(1-\mu_L)\mu_H q_0 > (1-\mu_H)\mu_L q_0 \\
          \underset{(2)}{\iff} &\mu_H > \mu_L
    \end{align*}
    where in (1) we divided by $A_H-A_L = (\mu_H-\mu_L)(q_1 - q_0) > 0$ since $q_1 > q_0$ (as $p$ satisfies Assumption \ref{assumption:non_abs_non_degen}) and $\mu_H > \mu_L$ and in (2) we divided by $q_0 > 0$ as $p$ satisfies Assumption \ref{assumption:non_abs_non_degen}. The last inequality holds as $\mu_H > \mu_L$ finishing the proof. 
    \end{proof} 

\begin{proof}[Proof of \cref{lemma:non-stationarity_newest_q=0_greater}]
Using \cref{lemma:stationary_dist_non-stationary_quality}, 
$$\pi^{\textsc{newest}}(0;\xi) =  \frac{ q_1 \Big[ ((1-\mu_L)A_H + (1-\mu_H)A_L) \xi + (1-\mu_L)(1-A_H) + (1-\mu_H)(1-A_L) \Big]}{((1-A_H)A_L + (1-A_L)A_H)\xi + 2(1-A_H)(1-A_L)}$$
Taking the limit $\xi \to 0$ yields
$$ \lim_{\xi \to 0} \pi^{\textsc{newest}}(0;\xi) = \frac{\frac{1}{2}q_1 \big[(1-\mu_L)(1-A_H) + (1-\mu_H)(1-A_H) \big]}{(1-A_H)(1-A_L)} = \underbrace{\frac{1}{2} \frac{q_1(1-\mu_L)}{1-A_L}}_{\textsc{term}_L}+\underbrace{\frac{1}{2} \frac{q_1(1-\mu_H)}{1-A_H}}_{\textsc{term}_H}.$$
By substituting the expression for $A_L = 1-q_1 (1-\mu_L) - q_0 \mu_L$, and using that $q_1 > q_0$ (as $p$ is non-degenerate as it satisfies Assumption \ref{assumption:non_abs_non_degen}), we can lower bound $\textsc{term}_L$ as 
\begin{equation}\label{ineq:lb_on_term_term_L}
\textsc{term}_L = \frac{1}{2} \frac{q_1(1-\mu_L)}{(1-\mu_L)q_1 + \mu_L q_0} > \frac{1}{2} \frac{q_1(1-\mu_L)}{q_1} = \frac{1-\mu_L}{2}.
\end{equation}
By substituting the expression for $A_H = 1-q_1 (1-\mu_H) - q_0 \mu_H$, and using that $q_1 > q_0$ (as $p$ is non-degenerate as it satisfies Assumption \ref{assumption:non_abs_non_degen}), we can lower bound $\textsc{term}_H$ as 
\begin{equation}\label{ineq:lb_on_term_term_H}
\textsc{term}_H = \frac{1}{2} \frac{q_1(1-\mu_H)}{(1-\mu_H)q_1 + \mu_H q_0} > \frac{1}{2} \frac{q_1(1-\mu_H)}{q_1} = \frac{1-\mu_H}{2}.
\end{equation}
Combining \eqref{ineq:lb_on_term_term_L} and \eqref{ineq:lb_on_term_term_H} with the fact that  $\pi^{\textsc{random}}(0) = \frac{1-\mu_L}{2} + \frac{1-\mu_H}{2}$ finishes the proof. 
 \end{proof}
To prove \cref{lemma:stationary_dist_non-stationary_quality}, we first establish that $M_t$ is irreducible and aperiodic, which are conditions which guarantee the existence of a unique stationary distribution.

\begin{lemma}\label{lemma:M_t_review_quality_aperiodic_irreducible}
The Markov Chain $M_t$ is irreducible and aperiodic.
\end{lemma}

\begin{proof}[Proof of \cref{lemma:M_t_review_quality_aperiodic_irreducible}]
We first show that the Markov Chain $M_t$ is irreducible and aperiodic. For states $M^{(1)} = (r_1, \mu_1), M^{(2)} =(r_2, \mu_2) \in \{0,1\} \times \{\mu_L, \mu_H\}$, $M_t$ can reach state $M^{(2)}$ from $M^{(1)}$ after a state transition of the product quality (with probability $\xi$) and a new state of $\mu_1$ (with probability $\frac{1}{2}$), a purchase (with probability at least $\min(q_1, q_0) = q_0$), and a review $r_2$ from $\mu_1$ (with probability $\mu_1^{r_2}(1-\mu_1)^{1-r_2}$). Therefore, the one-step transition probability from $M^{(1)}$ to $M^{(2)}$ is at least \begin{equation}\label{ineq:transition_lb}
    \prob[M_t = M^{(2)}|M_{t-1} = M^{(1)}] \geq \frac{\xi}{2} q_0 \mu_1^{r_2} (1-\mu_1)^{1-r_2}.
\end{equation}
Therefore, using that $\mu_L < 1$ and $\mu_H > 0$ (as $\mu_L < \mu_H$) and applying this formula when $(M^{(1)}, M^{(2)}) \in \Big\{((0,\mu_L),(0,\mu_H)), ((0,\mu_H),(1,\mu_H)), ((1,\mu_H),(1,\mu_L)), ((1,\mu_L)),(0,\mu_L))\Big\}$,
\begin{align*}\prob[M_t = (0,\mu_H)|M_{t-1} = (0,\mu_L)] &\geq \frac{\xi}{2} q_0 (1-\mu_L) > 0,\\\prob[M_t = (1, \mu_H)|M_{t-1} = (0,\mu_H)] &\geq \frac{\xi}{2} q_0 \mu_H > 0,\\\prob[M_t = (1, \mu_L)|M_{t-1} = (1, \mu_H)] &\geq \frac{\xi}{2} q_0 \mu_H > 0, \quad \text{and}\\\prob[M_t = (0, \mu_L)|M_{t-1} = (1, \mu_L)] &\geq \frac{\xi}{2} q_0 (1-\mu_L)> 0.\end{align*} Thus, any state of $M_t$ can be reached from any other state, i.e., $M_t$ is irreducible. 

We next show that $M_t$ is aperiodic. Applying \eqref{ineq:transition_lb} to $M^{(1)} = M^{(2)} = (0, \mu_L)$ we obtain $\prob[M_t = (0, \mu_L)|M_{t-1} = (0, \mu_L)] \geq \frac{\xi}{2} q_0 (1-\mu_L)> 0$. Thus, $M_t$ has a self-loop, i.e., it is aperiodic. \end{proof}

\begin{proof}[Proof of \cref{lemma:stationary_dist_non-stationary_quality}]
By \cref{lemma:M_t_review_quality_aperiodic_irreducible},  $M_t$ is aperiodic and irreducible. Thus, $M_t$ has a unique stationary distribution $\pi(\cdot, \cdot; \xi)$. We next show that $\pi(0, \mu_L; \xi)$ and $\pi(0, \mu_H; \xi)$ are given by the expressions in the statement. Given that the product quality $\mu^{(t)}$ spends an equal fraction of rounds in each of $\mu_L$ and $\mu_H$ in steady-state, then 
\begin{equation}\label{eq:stationary_mu_equal_states}
        \pi(0, \mu_L;\xi) +\pi(1, \mu_L;\xi) = \pi(0, \mu_H;\xi) +\pi(1, \mu_H;\xi) = \frac{1}{2}.
    \end{equation}
    We first write down the steady-state equations of $\pi$ for the state $(0, \mu_L)$. This state can be reached from
    the state $(0, \mu_L)$ if the product quality stays the same (with probability $(1-\frac{\xi}{2})$) and there is either a purchase and a negative review (with probability $q_0(1-\mu_L)$) or a non-purchase (with probability $1-q_0$). Thus, the state $(0, \mu_L)$ can be reached from the state $(0, \mu_L)$ with probability of $(1-\frac{\xi}{2})(q_0(1-\mu_L) + 1-q_0)$. The state $(0, \mu_L)$ can be reached from the state $(0, \mu_H)$ if the quality changes (with probability $\frac{\xi}{2}$) and there is either a purchase and a negative review (with probability $q_0 (1-\mu_H)$) or a non-purchase (with probability $1-q_0$). Thus, the state $(0, \mu_L)$ can be reached from the state $(0, \mu_H)$ with probability $\frac{\xi}{2}(q_0(1-\mu_H) + 1-q_0)$. The state $(0, \mu_L)$ can be reached from the state $(1, \mu_L)$ if the quality stays the same and there is a purchase with a negative review (with probability $(1-\frac{\xi}{2})q_1(1-\mu_L)$). The state $(0, \mu_L)$ can be reached from the state $(1, \mu_H)$ if the quality stays the same and there is a purchase with negative review (with probability $\frac{\xi}{2} q_1(1-\mu_H)$). Therefore, the steady-state equation at the state $\pi(0, \mu_L)$ is given by 
    \begin{align*}
     \pi(0, \mu_L; \xi) &= (1-\frac{\xi}{2})(q_0(1-\mu_L) + 1-q_0) \pi(0, \mu_L;\xi) + \frac{\xi}{2}(q_0(1-\mu_H) + 1-q_0) \pi(0, \mu_H; \xi) \\
     &+ (1-\frac{\xi}{2})q_1(1-\mu_L) \pi(1, \mu_L;\xi) + \frac{\xi}{2} q_1(1-\mu_H) \pi(1, \mu_H; \xi).
    \end{align*}
    Using \eqref{eq:stationary_mu_equal_states} \and the expressions for $A_L$ and $A_H$, the above is rewritten as
    $$
     \pi(0, \mu_L;\xi) = \pi(0, \mu_L;\xi) (1-\frac{\xi}{2})A_L + \pi(0, \mu_H;\xi) \frac{\xi}{2} A_H + \frac{1}{2} q_1 \big((1-\frac{\xi}{2})(1-\mu_L) + \frac{\xi}{2}(1-\mu_H) \big).$$
Rearranging terms,  we obtain     \begin{equation}\label{eq:steady-state-eq-(0,mu_L)}\pi(0, \mu_L;\xi) = \frac{\pi(0, \mu_H;\xi) \xi  A_H + \frac{1}{2} q_1\big((2-\xi)(1-\mu_L) + \xi(1-\mu_H) \big)}{2-(2-\xi)A_L}.\end{equation}
Due to an analogous argument, the steady-state equation at the state $(0, \mu_H)$ is 
    $$
     \pi(0, \mu_H; \xi) = \pi(0, \mu_L;\xi) \frac{\xi}{2} A_L + \pi(0, \mu_H;\xi) (1-\frac{\xi}{2}) A_H + \frac{1}{2} q_1 \big(\frac{\xi}{2}(1-\mu_L) + (1-\frac{\xi}{2})(1-\mu_H) \big).
    $$
    which can similarly be rewritten as 

    \begin{equation}\label{eq:steady-state-eq-(0,mu_H)}\pi(0, \mu_H;\xi) = \frac{\pi(0, \mu_L;\xi) \xi A_L + \frac{1}{2} q_1\big((2-\xi)(1-\mu_H) + \xi(1-\mu_L) \big)}{2-(2-\xi)A_H}.\end{equation}
Substituting \eqref{eq:steady-state-eq-(0,mu_L)} into \eqref{eq:steady-state-eq-(0,mu_H)} yields that 
\begin{align}\label{eq:expression_pi(0,mu_H)_with_X}
    \pi(0, \mu_H;\xi) &= \frac{\frac{\pi(0, \mu_H;\xi) \xi  A_H + \frac{1}{2} q_1\big((2-\xi)(1-\mu_L) + \xi(1-\mu_H) \big)}{2-(2-\xi)A_L} \xi A_L + \frac{1}{2} q_1\big((2-\xi)(1-\mu_H) + \xi(1-\mu_L) \big)}{2-(2-\xi)A_H} \nonumber \\
    &= \frac{\pi(0, \mu_H;\xi) \xi^2  A_L A_H + \frac{1}{2} q_1 X}{(2-(2-\xi) A_H)(2-(2-\xi) A_L)}
\end{align}
where
\begin{align*}
    X &=  \xi A_L \Big((2-\xi)(1-\mu_L) + \xi(1-\mu_H)\Big) + (2-(2-\xi) A_L )\Big((2-\xi)(1-\mu_H) + \xi(1-\mu_L) \Big) \\
    &=A_L\Big( \cancel{\xi (2-\xi)(1-\mu_L)} + \xi^2 (1-\mu_H) -(2-\xi)^2(1-\mu_H) - \cancel{\xi(2-\xi) (1-\mu_L)}\Big) \\
    &+ 2\Big((2-\xi)(1-\mu_H) + \xi(1-\mu_L) \Big) \\
    &= A_L(1-\mu_H)(\xi^2-(2-\xi)^2) + 2\Big((2-\xi)(1-\mu_H) + \xi(1-\mu_L) \Big)\\
    &= 4(1-\mu_H) A_L(\xi-1) + 2\Big((2-\xi)(1-\mu_H) + \xi(1-\mu_L) \Big)
\end{align*}
Substituting the expression for $X$ in \eqref{eq:expression_pi(0,mu_H)_with_X} yields 
\begin{align*}
     \pi(0, \mu_H;\xi) &=\frac{\pi(0, \mu_H;\xi) \xi^2  A_L A_H +  q_1 \big[2(1-\mu_H) A_L(\xi-1) + (2-\xi)(1-\mu_H) + \xi(1-\mu_L) \big]}{(2-(2-\xi) A_H)(2-(2-\xi) A_L)}\\
     \iff \pi(0, \mu_H;\xi) &= \frac{q_1 \big[2(1-\mu_H)A_L (\xi-1) + (1-\mu_L)\xi + (1-\mu_H)(2-\xi)\big]}{(2-(2-\xi)A_H)(2-(2-\xi)A_L)-\xi^2 A_H A_L}
\end{align*}
which is exactly the expression for $\pi^{\textsc{newest}}(0, \mu_H;\xi)$ given in the theorem statement. Due an analogous argument substituting \eqref{eq:steady-state-eq-(0,mu_H)} into \eqref{eq:steady-state-eq-(0,mu_L)} and solving for $\pi(0, \mu_L;\xi)$ yields 
$$\pi^{\textsc{newest}}(0, \mu_L;\xi) = \frac{q_1 \big[2(1-\mu_L)A_H (\xi-1) + (1-\mu_H)\xi + (1-\mu_L)(2-\xi)\big]}{(2-(2-\xi)A_H)(2-(2-\xi)A_L)-\xi^2 A_H A_L} $$
concluding the proof.
\end{proof}

\subsection{Newest always yields lower belief error than Random under non-stationarity (Proposition~\ref{prop: nonstationarity-belieferror-newest-better-than-random})}\label{appendix_subsec:non-stationarity-belief-error}
Our proof relies on analyzing the Markov Chain $M_t = (r_t, \mu^{(t)})$ where $r_t$ is the newest review the customer sees at round $t$ and $\mu^{(t)}$ is the product quality at round $t$, under the $\newest$ review ordering policy. Let $\pi^{\textsc{newest}}(r, \mu;\xi)$ be the stationary distribution of the process $M_t$ at a review $r \in \{0,1\}$ and product quality $\mu \in \{\mu_L, \mu_H\}$ when the probability of change is $\xi$. To prove the proposition we will use the characterization the stationary distribution of process $M_t$ from \cref{lemma:stationary_dist_non-stationary_quality}.

\begin{proof}[Proof of  \cref{prop: nonstationarity-belieferror-newest-better-than-random}.]
For a review ordering policy $\sigma$ and probability of change $\xi$, let $\pi^{\sigma}(r,\mu; \xi)$ be the fraction of time this policy spends at a state where the current review is $r \in \{0,1\}$ and the current product quality is $\mu \in \{\mu_L, \mu_H\}$. The belief error $\textsc{BeliefError}(\sigma, p;\xi)$ is given by:
\begin{align*}
    &\big(\expect[\mathrm{Beta}(a,b+1)]-\mu_L \big)^2 \pi^{\sigma}(0,\mu_L; \xi) + \big(\expect[\mathrm{Beta}(a+1,b)]-\mu_L \big)^2 \pi^{\sigma}(1,\mu_L; \xi)\\
    &+ \big(\expect[\mathrm{Beta}(a,b+1)]-\mu_H \big)^2 \pi^{\sigma}(0,\mu_H; \xi) + \big(\expect[\mathrm{Beta}(a+1,b)]-\mu_H \big)^2 \pi^{\sigma}(1,\mu_H; \xi).
\end{align*}
Given that $\mu_L = \expect[\mathrm{Beta}(a,b+1)]$ and $\mu_H = \expect[\mathrm{Beta}(a+1,b)]$, the belief error can be rewritten as
$\textsc{BeliefError}(\sigma, p;\xi) = (\mu_H-\mu_L)^2 (\pi^{\sigma}(1,\mu_L; \xi) + \pi^{\sigma}(0,\mu_H; \xi))$.
Given that the product quality $\mu^{(t)}$ spends an equal amount of time each of the states $\mu_L$ and $\mu_H$,  $\pi^{\sigma}(1, \mu_L;\xi) = \frac{1}{2}-\pi^{\sigma}(0, \mu_L;\xi)$, the belief error can be further rewritten as 
\begin{align}\label{eq:belief_error_general_policy}
    \textsc{BeliefError}(\sigma, p;\xi) &=(\mu_H-\mu_L)^2 (\frac{1}{2}-\pi^{\sigma}(0,\mu_L; \xi) + \pi^{\sigma}(0,\mu_H; \xi))\nonumber\\
    &= (\mu_H-\mu_L)^2(\pi^{\sigma}(0,\mu_H; \xi)-\pi^{\sigma}(0,\mu_L; \xi))+\frac{(\mu_H-\mu_L)^2}{2} .
\end{align}

Taking $\sigma = \random$, $\pi^{\sigma}(0,\mu_H; \xi)=\pi^{\sigma}(0,\mu_L; \xi) = \frac{1}{2}(1-\frac{\mu_H + \mu_L}{2})$, and thus \eqref{eq:belief_error_general_policy} implies that the belief error of $\random$ is $\textsc{BeliefError}(\random, p;\xi)  = \frac{(\mu_H-\mu_L)^2}{2}$. 

Taking $\sigma = \newest$ and using \cref{lemma:stationary_dist_non-stationary_quality} and the expressions for $A_L$ and $A_H$, the difference $ \pi^{\textsc{newest}}(0,\mu_H; \xi) -  \pi^{\textsc{newest}}(0,\mu_L; \xi)$ equals
\begin{align*}
    &\frac{q_1 \big[2(1-\mu_H)A_L (\xi-1) + (1-\mu_L)\xi + (1-\mu_H)(2-\xi)\big]}{(2-(2-\xi)A_H)(2-(2-\xi)A_L)-\xi^2 A_H A_L}\\
    &-  \frac{q_1 \big[2(1-\mu_L)A_H (\xi-1) + (1-\mu_H)\xi + (1-\mu_L)(2-\xi)\big]}{(2-(2-\xi)A_H)(2-(2-\xi)A_L)-\xi^2 A_H A_L}\\ 
&= \frac{2 q_1 (\xi-1) \big( (1-\mu_L)(1-A_H) - (1-\mu_H)(1-A_L) \big)}{(2-(2-\xi)A_H)(2-(2-\xi)A_L)-\xi^2 A_H A_L}\\
&= \frac{2 q_1 (\xi-1) \big[ (1-\mu_L)((1-\mu_H) q_1 + \mu_H q_0) - (1-\mu_H)((1-\mu_L)q_1 + \mu_L q_0) \big]}{(2-(2-\xi)A_H)(2-(2-\xi)A_L)-\xi^2 A_H A_L}\\
&= \frac{2 q_1 q_0(\xi-1)(\mu_H - \mu_L)}{(2-(2-\xi)A_H)(2-(2-\xi)A_L)-\xi^2 A_H A_L}.
\end{align*}
Therefore \eqref{eq:belief_error_general_policy}combined with the fact that $\textsc{BeliefError}(\random, p;\xi)  = \frac{(\mu_H-\mu_L)^2}{2}$ implies
\begin{align*}
    \textsc{BeliefError}(\newest, p;\xi) &= (\mu_H -\mu_L)^2 \frac{2 q_1 q_0 (\xi-1)(\mu_H - \mu_L)}{(2-(2-\xi)A_H)(2-(2-\xi)A_L)-\xi^2 A_H A_L} \\
    & + \textsc{BeliefError}(\random, p;\xi) \\
    &= 2 q_1 q_0(\mu_H -\mu_L)^3 g(\xi) + \textsc{BeliefError}(\random, p;\xi)
\end{align*}
where $g(\xi) = \frac{(\xi-1)}{(2-(2-\xi)A_H)(2-(2-\xi)A_L)-\xi^2 A_H A_L} = \frac{\xi-1}{2((1-A_L)A_H+ (1-A_H)A_L) \xi + 4(1-A_H)(1-A_L)}$. Given that for any $\xi \in (0,1)$, the term $g(\xi)$ is strictly negative and $\mu_H -\mu_L > 0$, $\textsc{BeliefError}(\newest, p;\xi) < \textsc{BeliefError}(\random, p;\xi)$ showing the main claim of the proposition. 

To show the second claim, it suffices to show that  $g(\xi)$ is strictly increasing in $\xi$. A function of the form $\xi \to \frac{A \xi + B}{C \xi + D}$ is strictly increasing if and only if $AD - BC>0$. The function $g$ corresponds to parameters $A = 1$, $B = -1$, $C = 2((1-A_L)A_H+ (1-A_H)A_L)$, and $D = 4(1-A_H)(1-A_L)$. The condition $AD -BC > 0$ is equivalent to $4(1-A_H)(1-A_L)  + 2((1-A_L)A_H+ (1-A_H)A_L) > 0$ which holds as $A_H > 0$ and $1-A_L > 0$, showing the second claim of the proposition. Given that $\textsc{BeliefError}(\random, p;\xi) = \frac{(\mu_H-\mu_L)^2}{2}$ does not depend on $\xi$, the third claim follows.
 \end{proof}

\section{Motivation behind using Random as benchmark (Section \ref{subsec: policy_decison})}\label{appendix:why_random_right_benchmark}
To motivate our choice of $\random$ as a benchmark, we establish that it is the revenue-maximizing policy out of a class of review ordering policies which do not reveal review ratings, i.e., they
\begin{enumerate}
    \item do not take the review ratings into account, and
    \item do not take the actual times of posting of the reviews into account.
\end{enumerate}
Without the first requirement, the platform can maximize revenue by always ordering all positive reviews first; this strategy provides a clearly biased set of reviews which makes it undesirable. Without the second requirement, the platform can use the time between consecutive reviews to infer whether a review is positive -- if another review follows soon after, then it is more likely that the review is positive as positive reviews induce purchases faster than negative reviews. Incorporating the actual time therefore has a similar criticism with incorporating review ratings. As a result, we restrict our attention to ordering policies that only take the relative time of posting of each of the reviews (which does not allow inference of whether a review is positive). We further restrict our attention to stationary policies: the exact ordering rule does not depend on the round.

We define a general class of ordering policies that satisfy these requirements. Let $S_w$ be the family of all $c$-sized subsets of $\{1, \ldots, w\}$. 
We denote by $\sigma^{\textsc{random}(w, \psi)}$ the review ordering policy that selects $c$ reviews from the most recent $w$ according to a fixed probability distribution $\psi \in \Delta(S_w)$. We denote $\sigma^{\textsc{random}(w, \psi)}$ as $\sigma^{\textsc{random}(w)}$ when $\psi = \mathcal{U}(S_w)$ is uniform distribution over $\mathcal{S}_w$. Let $\mathcal{C} = \{\sigma^{\textsc{random}(w, \psi)}, w \geq c, \psi \in \Delta(S_w)\}$ be the class of review ordering policies $\sigma^{\textsc{random}(w, \psi)}$ for some $w \geq c$ and distribution $\psi \in \Delta(S_w)$. Note that $\random$ is not strictly in the class $\mathcal{C}$ but it can defined as the limit of $\sigma^{\textsc{random}(w)}$ as $ w \to \infty$ (see \cref{prop: limit_dist_window_w_to_random}, Appendix \ref{subsec_app:formal_defintion_random}). In Proposition \ref{prop:random_highest_revenue_amongst_class} (Appendix \ref{app_subsec_proof_random_highest_revenue}), we show that  (1) $\random$ yields strictly larger revenue than any ordering policy in $\mathcal{C}$ and (2) the revenue of $\random$ can be approximated arbitrarily well by ordering policies in the class $\mathcal{C}$. Therefore, we can think of the revenue under $\random$ as the maximum revenue achievable by any policy in this class.

\subsection{Formal definition of Random as the limit of window-random policies}\label{subsec_app:formal_defintion_random}
We now show that the policy $\random(\mathcal{H}_t)$ can be defined as the limit of $\sigma^{\textsc{random}(w)}(\mathcal{H}_t)$ as $w \to \infty$.   For a vector of review ratings $\bm{z} \in \{0,1\}^c$, we denote $N(\bm{z}) \coloneqq N_{\bm{z}} = \sum_{i=1}^c z_i$. \cref{prop: limit_dist_window_w_to_random} implies that as $w \to \infty$, the draws from $\sigma^{\textsc{random}(w)}$ across different rounds are independent and the distribution $\sigma^{\textsc{random}(w)}(\mathcal{H}_t)$ approaches the distribution of $\random$.

\begin{proposition}\label{prop: limit_dist_window_w_to_random}
    For any review rating vectors $\bm{z}^{(1)}, \bm{z}^{(2)}, \ldots, \bm{z}^{(k)} \in \{0,1\}^c$ and rounds $t^{(1)} < t^{(2)}<\ldots <t^{(k)}$, the distribution of reviews by $\sigma^{\textsc{random}(w)}$ approaches the one of $\sigma^{\textsc{random}}$ as $w\to \infty$, i.e.,
    $$\lim_{w \to \infty} \prob \Big[\forall i \in \{1, \ldots, k\}: \sigma^{\textsc{random}(w)}(\mathcal{H}_{t^{(i)}})= \bm{z}^{(i)} \Big] 
    = \prod_{i=1}^k \prob \Big[\sigma^{\textsc{random}}(\mathcal{H}_{t^{(i)}})= \bm{z}^{(i)}\Big].
    $$
\end{proposition}
We first show that as $w$ goes to infinity, the $c$ reviews selected by $\sigma^{\textsc{random}(w)}(\mathcal{H}_{t^{(i)}})$ at each round~$t^{(i)}$ come from the infinite pool of reviews $\{X_{-i}\}_{i=1}^{\infty}$ with high probability. Let $n_w = \floor{\sqrt{w - t^{(k)}}}$ be the largest integer such that $n_w^2 \leq w-t^{(k)}$ and let $w$ be large enough so that $n_w \geq 1$. We define $\mathcal{E}^{\textsc{OrderInf}}$ as the event that for every $i \in \{1, \ldots, k\}$ the $c$ reviews selected by $\sigma^{\textsc{random}(w)}(\mathcal{H}_{t^{(i)}})$ come from the set $\{X_{-1}, \ldots, X_{-n_w^2}\}$. Lemma~\ref{lemma_reviews_from_inf_pool} shows that $\mathcal{E}^{\textsc{OrderInf}}$ happens with high probability. 
\begin{lemma}\label{lemma_reviews_from_inf_pool}
    $\prob[\mathcal{E}^{\textsc{OrderInf}}] \geq \Big[ \frac{\binom{n_w^2}{c}}{\binom{w}{c}} \Big]^k$. Furthermore, $\lim_{w \to \infty}\prob[\mathcal{E}^{\textsc{OrderInf}}] = 1$. 
\end{lemma}
Given that $\sigma^{\textsc{random}(w)}$ selects from the infinite negative pool of reviews with high probability (if $\mathcal{E}^{\textsc{OrderInf}}$ holds), we now derive a concentration bound for the reviews in this pool. Recall that $X_{-\ell} \sim_{i.i.d.} \bern(\mu)$ for $\ell \in \{1, \ldots, n_w^2\}$. 
We partition those reviews into $n_w$ groups each containing $n_w$ reviews. Let $\mathcal{E}^{\textsc{Conc}}$ be the event that, for each group, the average review rating concentrates around the group's mean, i.e. 

$$\mathcal{E}^{\textsc{Conc}} = \Big\{ \forall j \in \{1, \ldots, n_w\}: \sum_{\ell= (j-1)\cdot n_w+1}^{j\cdot n_w} X_{-\ell} \in [\mu(n_w-n_w^{2/3}), \mu(n_w+n_w^{2/3})] \Big\}.$$ Our next lemma shows that the concentration event $\mathcal{E}^{\textsc{Conc}}$ happens with high probability. 

\begin{lemma}\label{lemma_bernoulli_concentration_nc}
    $\prob\Big[ \mathcal{E}^{\textsc{Conc}} \Big] \geq 1-2n_w \exp \Big(-\frac{n_w^{1/3} \mu}{3} \Big)$. Furthermore, $\lim_{w \to \infty} \prob\Big[ \mathcal{E}^{\textsc{Conc}} \Big] = 1$.
\end{lemma}
In order to decompose the left-hand-side of \cref{prop: limit_dist_window_w_to_random} into a product of probabilities, we show the following independence lemma. 
\begin{lemma}\label{lemma_independence_across_rounds}
    Let $(y_1, \ldots, y_{n_w^2}) \in \{0,1\}^{n_w^2}$. Then conditioned on $(X_{-1}, \ldots, X_{-n_w^2}) = (y_1, \ldots, y_{n_w^2})$ and $\mathcal{E}^{\textsc{OrderInf}}$, the events $\{\sigma^{\textsc{random}(w)}(\mathcal{H}_{t^{(i)}}) = \bm{z}^{(i)}\}$ for $i = 1, \ldots, k$ are independent.
\end{lemma}
Having shown independence across rounds, we prove that, for any round $t^{(i)}$, $\sigma^{\textsc{random}(w)}(\mathcal{H}_{t^{(i)}})$ is close to $\sigma^{\textsc{random}}$. Let $\mathcal{A}$ be the set of review rating sequences for which event $\mathcal{E}^{\textsc{Conc}}$ holds i.e. 
$$\mathcal{A} = \Big\{ (y_1, \ldots, y_{n_w^2}) \in \{0,1\}^{n_w^2}: \forall j \in \{1, \ldots, n_w\},  \sum_{\ell= (j-1) \cdot n_w+1}^{j \cdot n_w} y_{\ell} \in [\mu(n_w-n_w^{2/3}), \mu(n_w+n_w^{2/3})] \Big\}$$
Our next lemma shows that if the events $\mathcal{E}^{\textsc{OrderInf}}$ and $\mathcal{E}^{\textsc{Conc}}$ hold, then for any round $t^{(i)}$, the distribution of $\sigma^{\textsc{random}(w)}(\mathcal{H}_{t^{(i)}})$ is close to the distribution of $\sigma^{\textsc{random}}(\mathcal{H}_{t^{(i)}})$. To ease notation, we let $X_{-1:n_w^2}\coloneqq (X_{-1}, \ldots, X_{-n_w^2})$ and $y_{1:n_w^2}\coloneqq(y_1, \ldots, y_{n_w^2})$.

\begin{lemma}\label{lemma_uniform_bound_probabilities}
There exists some function $f: \mathbb{N} \to \mathbb{R}$ satisfying $\lim_{n \to \infty} f(n) = 0$ such that assuming that $\mathcal{E}^{\textsc{OrderInf}}$ and $\mathcal{E}^{\textsc{Conc}}$ hold, for any $\bm{z} = (z_1, \ldots, z_c) \in \{0,1\}^c$ and $(y_1, \ldots, y_{n_w^2}) \in \mathcal{A}$:
    \begin{equation*}\label{eq:prob_approximation_bd}
\Big|\prob\Big[\sigma^{\textsc{random}(w)}(\mathcal{H}_{t^{(i)}}) = \bm{z}|X_{-1:n_w^2} = y_{1:n_w^2}\Big] - \prob\Big[\sigma^{\textsc{random}}(\mathcal{H}_{t^{(i)}}) = \bm{z}\Big] \Big| \leq f(n_w).
\end{equation*}
\end{lemma}

\begin{proof}[Proof of Proposition~\ref{prop: limit_dist_window_w_to_random}.]
By \cref{lemma_reviews_from_inf_pool} and \cref{lemma_bernoulli_concentration_nc}, we can upper bound the probability that neither of $\mathcal{E}^{\textsc{OrderInf}}$  nor $\mathcal{E}^{\textsc{Conc}}$ holds by 
\begin{equation}\label{event_prob_upper_bound_zero}
    \prob[(\neg \mathcal{E}^{\textsc{OrderInf}}) \cup  (\neg \mathcal{E}^{\textsc{Conc}})] \leq 1-\Big[ \frac{\binom{n_w^2}{c}}{\binom{w}{c}} \Big]^k + 2n_w \exp \Big(-\frac{n_w^{1/3} \mu}{3} \Big) \to_{w \to \infty} 0.
\end{equation}
Hence, we assume that $\mathcal{E}^{\textsc{OrderInf}}$  and $\mathcal{E}^{\textsc{Conc}}$ hold and thus focus on sequences $(X_{-1}, \ldots, X_{-n_w^2}) \in \mathcal{A}$.  By the law of total probability and the independence of $\sigma^{\textsc{random}(w)}$ across rounds (i.e. \cref{lemma_independence_across_rounds}):
\begin{align}\label{prob_decomposition_sigma_random}
    &\prob \Big[\forall i \in \{1, \ldots, k\}: \sigma^{\textsc{random}(w)}(\mathcal{H}_{t^{(i)}}) = \bm{z}^{(i)}|\mathcal{E}^{\textsc{OrderInf}}, \mathcal{E}^{\textsc{Conc}} \Big] \\
    &=\sum_{(y_1, \ldots, y_{n_w^2}) \in \mathcal{A}} \underbrace{\prod_{i=1}^k \prob \Big[ \sigma^{\textsc{random}(w)}(\mathcal{H}_{t^{(i)}}) = \bm{z}^{(i)}|\mathcal{E}^{\textsc{OrderInf}},  X_{-1:n_w^2} =y_{1:n_w^2} \Big]}_{(\textsc{Dec})} \cdot \underbrace{\prob\big[X_{-1:n_w^2} =y_{1:n_w^2}|\mathcal{E}^{\textsc{OrderInf}},\mathcal{E}^{\textsc{Conc}}\big]}_{\textsc{Seq}(y_{1:n_w^2})}. \nonumber
\end{align}
By \cref{lemma_uniform_bound_probabilities}, we can upper and lower bound this decomposition by 
\begin{equation}\label{lower_upper_bound_decomp}
\prod_{i=1}^k  \Bigg(\prob\Big[ \sigma^{\textsc{random}}(\mathcal{H}_{t^{(i)}}) = \bm{z}^{(i)} \Big] -f(n_w) \Bigg)
\leq (\textsc{Dec}) \leq
\prod_{i=1}^k  \Bigg(\prob\Big[ \sigma^{\textsc{random}}(\mathcal{H}_{t^{(i)}}) = \bm{z}^{(i)} \Big] +f(n_w) \Bigg).
\end{equation}
Recall that $\mathcal{A}$ consists of all sequences $(y_1, \ldots, y_{n_w^2})$ that satisfy the event $\mathcal{E}^{\textsc{Conc}}$. Summing across those sequences and using the independence between the choice of $\sigma^{\textsc{random}(w)}$ (which determines $\mathcal{E}^{\textsc{OrderInf}}$) and the values of $X_{-1:n_w^2}$ (which determine $\mathcal{E}^{\textsc{Conc}}$), we obtain \begin{equation}\label{eq:sum_seq}\sum_{(y_1, \ldots, y_{n_w^2}) \in \mathcal{A}} \textsc{Seq}(y_{1:n_w^2}) = \prob[\mathcal{E}^{\textsc{Conc}}| \mathcal{E}^{\textsc{OrderInf}}] = \underbrace{\prob[\mathcal{E}^{\textsc{Conc}}] \to_{w \to \infty} 1}_{\cref{lemma_bernoulli_concentration_nc}}. \end{equation} 
We conclude the proof by combining (\ref{event_prob_upper_bound_zero}), (\ref{prob_decomposition_sigma_random}), (\ref{lower_upper_bound_decomp}), (\ref{eq:sum_seq}) and using the fact that $f(n_w) \to_{w \to \infty} 0 $. 
\end{proof}

\subsubsection{High probability bound on reviews being drawn from infinite pool (Lemma~\ref{lemma_reviews_from_inf_pool})}
\begin{proof}[Proof of \cref{lemma_reviews_from_inf_pool}.]
Fix a particular round $t^{(i)}$. Note that since $n_w =\floor{\sqrt{w-t^{(k)}}}$, the $w$ most recent reviews that $\sigma^{\textsc{random}(w)}(\mathcal{H}_{t^{(i)}})$ considers contain $(X_{-1}, \ldots, X_{-n_w^2})$. The probability that all~$c$ reviews come from that set is $\frac{\binom{n_w^2}{c}}{\binom{w}{c}}$. For each round in $t^{(1)}, \ldots, t^{(k)}$, the subset of $c$ reviews is independently drawn by $\sigma^{\textsc{random}(w)}$. Thus, $\prob[\mathcal{E}^{\textsc{OrderInf}}] \geq \Big[ \frac{\binom{n_w^2}{c}}{\binom{w}{c}} \Big]^k$. Next notice that $n_w^2 \geq w-t^{(k)}-1$ and therefore, $\frac{\binom{n_w^2}{c}}{\binom{w}{c}} \geq \frac{\binom{w-t^{(k)}-1}{c}}{\binom{w}{c}} \to_{w \to \infty} 1$ showing the second part. 
\end{proof}

\subsubsection{Concentration for infinite pool (Lemma~\ref{lemma_bernoulli_concentration_nc})}
\begin{proof}[Proof of \cref{lemma_bernoulli_concentration_nc}.]
For a fixed $j \in \{1, \ldots, n_w\}$, by Chernoff bound, it holds that $$\prob \Big[ \sum_{\ell=(j-1)n_w +1}^{jn_w} X_{-\ell} \in [\mu (n_w-n_w^{2/3}), \mu(n_w+n_w^{2/3})] \Big] \geq 1-2 \exp \Big(-\frac{n_w^{1/3} \mu}{3} \Big).$$ By union bound on the $n_w$ groups, the event $\mathcal{E}^{\textsc{Conc}}$ has probability at least $1-2n_w \exp \Big(-\frac{n_w^{1/3} \mu}{3} \Big)$. Note that  $n_w \exp \Big(-\frac{n_w^{1/3} \mu}{3} \Big) \to_{w \to \infty} 0$, which yields the second part. 
\end{proof}

\subsubsection{Independence of window-random policy across rounds (Lemma~\ref{lemma_independence_across_rounds})}
\begin{proof}[Proof of \cref{lemma_independence_across_rounds}.]
    By definition, $\sigma^{\textsc{random}(w)}$ samples a $c$-sized subset of reviews independently at every round. Thus, conditioning on the subset being sampled from $(X_{-1}, \ldots, X_{-n_w^2})$, i.e. $\mathcal{E}^{\textsc{OrderInf}}$, and also conditioning on the exact values of these ratings, i.e. $(X_{-1}, \ldots, X_{-n_w^2}) = (y_1, \ldots, y_{n_w^2})$, the draws of $\sigma^{\textsc{random}(w)}(\mathcal{H}_{t^{(i)}})$ are independent for $i = 1, \ldots, k$. 
\end{proof}

\subsubsection{Single round approximation (Lemma~\ref{lemma_uniform_bound_probabilities})}
Let $\sigma^{\textsc{random}(w)}(\mathcal{H}_{t^{(i)}}) = (Z_1, \ldots, Z_c) \in \{0,1\}^c$ be the $c$ review ratings chosen from the most recent $w$ reviews. Recall that when   $\mathcal{E}^{\textsc{OrderInf}}$ holds,  $\sigma^{\textsc{random}(w)}$ selects reviews from $(X_{-1}, \ldots, X_{-n_w^2})$
and that $(X_{-1}, \ldots, X_{-n_w^2}) = (y_1, \ldots, y_{n_w^2})$. Let $S_j = \{y_{\ell}\}_{\ell = (j-1) \cdot n_w +1}^{j \cdot n_w}$ for $j \in \{1, \ldots, n_w\}$ be a partition of all reviews $(y_1, \ldots, y_{n_w^2})$ into $n_w$ groups of $n_w$ reviews each. We first show that the reviews drawn by $\sigma^{\textsc{random}(w)}(\mathcal{H}_{t^{(i)}})$ come from different groups $S_j$ with high probability. Let $\mathcal{E}^{\textsc{diff\_groups}}$ be the event that no two indices $m_1 < m_2 \in \{1, \ldots, c\}$ are such that $Z_{m_1}$ and $Z_{m_2}$ belong to the same group $S_j$ for some $j \in \{1, \ldots, n_w\}$. Our next lemma lower bounds the probability of $\mathcal{E}^{\textsc{diff\_groups}}$. 
\begin{lemma}\label{lemma_diff_groups_whp}
Assume that $\mathcal{E}^{\textsc{OrderInf}}$ and $\mathcal{E}^{\textsc{Conc}}$ hold. The probability of any two selected indices being of different groups is:
$\prob[\mathcal{E}^{\textsc{diff\_group}}] \geq  1-c! \cdot c^2 \frac{n_w^{2c-1}}{(n_w^2-c)^c}$ and thus $\lim_{w \to \infty}\prob[\mathcal{E}^{\textsc{diff\_group}}] = 1 $. 
\end{lemma}
\begin{proof}
Let $\mathcal{E}_{m_1,m_2,j}$ be the event that $Z_{m_1}, Z_{m_2} \in S_j$. Notice that $\mathcal{E}^{\textsc{diff\_group}}$ is exactly the event that none of the 
$\mathcal{E}_{m_1,m_2,j}$ hold. Further, note that $Z_{m_1}, Z_{m_2} \in S_j$ implies that $Z_{m_1}, Z_{m_1+1}, \ldots, Z_{m_2} \in S_j$ since we output an ordered set of $c$ reviews by recency. There are thus $\binom{n_w}{m_2-m_1+1}$ ways to choose $Z_{m_1}, Z_{m_1+1}, \ldots, Z_{m_2} $ to be in the same group $S_j$. Hence the  probability of $\mathcal{E}_{m_1,m_2,j}$ is at most 
$$\prob[\mathcal{E}_{m_1,m_2,j}] \leq \frac{\binom{n_w^2}{m_1-1} \binom{n_w}{m_2-m_1+1} \binom{n_w^2}{c-m_2}}{\binom{n_w^2}{c}}$$
as there are at most $\binom{n_w^2}{m_1-1}$ ways to choose $Z_1, \ldots, Z_{m_1-1}$ and at most $\binom{n_w^2}{m_2-m_1}$ ways to choose $Z_{m_2+1}, \ldots, Z_c$. Using the inequalities $\frac{(n-k)^k}{k!} \leq \binom{n}{k} \leq n^k$, we obtain that 
$$\prob[\mathcal{E}_{m_1,m_2,j}] \leq c! \cdot \frac{n_w^{2(m_1-1)} n_w^{m_2-m_1+1} n_w^{2(c-m_2)}}{(n_w^2-c)^c} = c! \cdot \frac{n_w^{2c-(m_2-m_1+1)}}{(n_w^2-c)^c} \leq c! \cdot \frac{n_w^{2c-2}}{(n_w^2-c)^c}$$ since $m_2-m_1+1 \geq 2$. Thus, via union bound over $m_1, m_2 \in \{1, \ldots, c\}$ and $j \in \{1, \ldots, n_w\}$, i.e., $c^2n_w$ events, the probability of $\mathcal{E}^{\textsc{diff\_group}}$ is lower bounded as follows
$$\prob[\mathcal{E}^{\textsc{diff\_group}}] \geq 1- c^2 n_w \cdot c! \cdot  \frac{n_w^{2c-2}}{(n_w^2-c)^c} = 1-c! \cdot c^2 \frac{n_w^{2c-1}}{(n_w^2-c)^c} \to_{n_w \to \infty} 1.$$
 \end{proof}
When the event $\mathcal{E}^{\textsc{diff\_groups}}$ holds then the reviews $Z_m$ for $m = 1, \ldots, c$ come from different groups $\{S_j\}_{j =1}^{n_w}$. We show that when this happens the values of the reviews $Z_m$ are independent. Let $\mathcal{E}_{j_1, \ldots, j_c}$ be the event that review $Z_m$ comes from group $S_{j_m}$ for $m = 1, \ldots, c$. The next lemma shows that conditioned on the event $\mathcal{E}_{j_1, \ldots, j_c}$, the values of $Z_1, \ldots, Z_c$ are independent.\footnote{This is not the case without conditioning on $\mathcal{E}_{j_1, \ldots, j_c}$. As an example, consider two ordered reviews $(Z_1, Z_2)$ drawn uniformly from the three ordered review ratings $(1,0,1,1)$. Conditioned on $Z_1 = 0$ then $Z_2 = 1$ deterministically while conditioned on $Z_1 = 1$ then $\prob[Z_2 = 0|Z_1 = 1] = \frac{1}{4}$. As a results $Z_1$ and $Z_2$ are correlated.} 
Recall that $\mathcal{E}^{\textsc{Conc}}$ implies that $(y_1, \ldots, y_{n_w^2}) \in \mathcal{A}$. 
\begin{lemma}\label{lemma:indep_across_diff_groups}Let $j_1, j_2 \ldots, j_c \in \{1, \ldots, n_w\}$ be group indices with $j_1<j_2<\ldots,\leq j_c$. Conditioned on $\mathcal{E}_{j_1, \ldots, j_c}, \mathcal{E}^{\textsc{Conc}}$,  $\mathcal{E}^{\textsc{OrderInf}}$, and $(y_1, \ldots, y_{n_w^2}) \in \mathcal{A}$, for any vector of review ratings $\bm{z} \in \{0,1\}^c$, the events $\{Z_m = z_m\}$ for $m = 1, \ldots, c$ are independent. Furthermore, $$\prob\big[Z_m = z_m\big] = \Big(\frac{\sum_{\ell \in S_{j_m}}y_{\ell}}{n_w} \Big)^{z_{m}} \Big(1-\frac{\sum_{\ell \in S_{j_m}}y_{\ell}}{n_w}\Big)^{1-z_{m}}.$$
\end{lemma}
\begin{proof}
For any specific review ratings $y_{\ell_1} \in S_{j_1}, \ldots, y_{\ell_m} \in S_{j_m}$, applying Bayes rule
\begin{align*}
    \prob\big[Z_1 = y_{\ell_1}, \ldots, Z_c = y_{\ell_c}|\mathcal{E}_{j_1, \ldots, j_c}\big] &= \frac{\prob\big[Z_1 = y_{\ell_1}, \ldots, Z_m = y_{\ell_m}, \mathcal{E}_{j_1, \ldots, j_c}\big]}{\prob \big[\mathcal{E}_{j_1, \ldots, j_c} \big]}
\end{align*}

There are $(n_w)^c$ choices for the set of $c$ ratings $(Z_1, \ldots, Z_c)$ so that $Z_m \in S_{j_m}$ for $m =1 , \ldots, c$ and $\binom{(n_w)^2}{c}$ total number of choices. 
Hence $\mathcal{E}_{j_1, \ldots, j_c}$ holds with probability $\prob \big[\mathcal{E}_{j_1, \ldots, j_c} \big] = \frac{(n_w)^c}{\binom{(n_w)^2}{c}}$. Given that there is exactly one choice for the $c$ reviews $(Z_1, \ldots, Z_c)$ which satisfies $Z_1 = y_{\ell_1}, \ldots, Z_c = y_{\ell_c}$:
\begin{equation}\label{eq:independece_across_groups}
\prob\big[Z_1 = y_{\ell_1}, \ldots, Z_c = y_{\ell_c}|\mathcal{E}_{\ell_1, \ldots, \ell_c}\big] =  \frac{\frac{1}{\binom{(n_w)^2}{c}}}{\frac{(n_w)^c}{\binom{(n_w)^2}{c}}} =\frac{1}{(n_w)^c}. 
\end{equation}
Since this holds for any $y_{\ell_1} \in S_{j_1}, \ldots, y_{\ell_m} \in S_{j_m}$ we obtain independence of $\{Z_m = z_m\}$ for $m =1, \ldots, c$. By summing (\ref{eq:independece_across_groups}) over $y_{\ell_2}, \ldots, y_{\ell_c}$, for any particular $y_{\ell_m} \in S_{j_m}$, we have 
$\prob[Z_m = y_{\ell_m}|\mathcal{E}_{\ell_1, \ldots, \ell_c}] = \frac{1}{n_w}$. Therefore, the probability that $Z_m$ has a particular value $z_m \in \{0,1\}$ is equal to the fraction of $y_{\ell_m} \in S_{j_m}$ which have value $z_m$ i.e. 
$$\prob\big[Z_m = z_m\big] = \Big(\frac{\sum_{\ell \in S_{j_m}}y_{\ell}}{n_w} \Big)^{z_{m}} \Big(1-\frac{\sum_{\ell \in S_{j_m}}y_{\ell}}{n_w}\Big)^{1-z_{m}}.$$
 \end{proof}

We next show that conditioned on the events $\mathcal{E}^{\textsc{diff\_group}}$,$\mathcal{E}^{\textsc{Conc}}$, and $\mathcal{E}^{\textsc{OrderInf}}$ the distribution of $\sigma^{\textsc{random}(w)}(\mathcal{H}_{t^{(i)}})$ is close to the distribution of $\sigma^{\textsc{random}}$. Given that the latter consists of i.i.d. $\bern(\mu)$, $\prob(\sigma^{\textsc{random}} = \bm{z}) = \mu^{N(\bm{z})}(1-\mu)^{c-N(\bm{z})}$ for vector of review ratings $\bm{z} \in \{0,1\}^c$. Recall that $N(\bm{z}) \coloneqq N_{\bm{z}} = \sum_{i=1}^c z_i$. Using the independence acorss different groups (\cref{lemma:indep_across_diff_groups}), the following lemma shows that $\prob[\sigma^{\textsc{random}(w)}(\mathcal{H}_{t^{(i)}}) = \bm{z}]$ has a similar decomposition as $\prob[\sigma^{\textsc{random}} = \bm{z}]$. 

\begin{lemma}\label{lemma_close_conditioned_on_diff_groups}
Conditioned on $\mathcal{E}_{j_1, \ldots, j_c}, \mathcal{E}^{\textsc{Conc}}$,  $\mathcal{E}^{\textsc{OrderInf}}$, and $(y_1, \ldots, y_{n_w}) \in \mathcal{A}$, for any vector of review ratings $\bm{z} \in \{0,1\}^c$, $\prob[\sigma^{\textsc{random}(w)}(\mathcal{H}_{t^{(i)}}) = \bm{z}]$ is lower bounded by 
    $(\mu(1-n_w^{-\frac{1}{3}}))^{N(\bm{z})}(1-\mu(1+n_w^{-\frac{1}{3}}))^{c-N(\bm{z})}$ and upper bounded by $(\mu(1+n_w^{-\frac{1}{3}}))^{N(\bm{z})}(1-\mu(1-n_w^{-\frac{1}{3}}))^{c-N(\bm{z})}$
\end{lemma}
\begin{proof}
 Using \cref{lemma:indep_across_diff_groups}: 
$$\prob[\sigma^{\textsc{random}(w)}(\mathcal{H}_{t^{(i)}}) = \bm{z}|\mathcal{E}_{j_1, \ldots, j_c}]  = \prod_{m=1}^{c}(\frac{\sum_{\ell \in S_{j_m}}y_{\ell}}{n_w})^{z_{m}}(1-\frac{\sum_{\ell \in S_{j_m}}y_{\ell}}{n_w})^{1-z_{m}}.$$
As $(y_1, \ldots, y_{n_w^2}) \in \mathcal{A}$, we have that $\mu(1-n_w^{-\frac{1}{3}}) \leq \frac{\sum_{\ell \in S_{j_m}}y_{\ell}}{n_w}  \leq \mu(1+n_w^{-\frac{1}{3}})$ for all $m = 1, \ldots, c$. Applying these inequalities for each $m$ yields
\begin{align*}
&\prob[\sigma^{\textsc{random}(w)}(\mathcal{H}_{t^{(i)}}) = \bm{z}|\mathcal{E}_{j_1, \ldots, j_c}] \\
&\in \Big[(\mu(1-n_w^{-\frac{1}{3}}))^{N(\bm{z})}(1-\mu(1+n_w^{-\frac{1}{3}}))^{c-N(\bm{z})},(\mu(1+n_w^{-\frac{1}{3}}))^{N(\bm{z})}(1-\mu(1-n_w^{-\frac{1}{3}}))^{c-N(\bm{z})} \Big].
\end{align*}  
for any $\mathcal{E}_{j_1, \ldots, j_c}$. As $\mathcal{E}_{\textsc{diff\_groups}}$ is the union of events $\mathcal{E}_{j_1, \ldots, j_c}$ for all indices $j_1, \ldots, j_c$, this implies
\begin{align*}
&\prob[\sigma^{\textsc{random}(w)}(\mathcal{H}_{t^{(i)}}) = \bm{z}|\mathcal{E}_{\textsc{diff\_groups}}] \\
&\in \Big[(\mu(1-n_w^{-\frac{1}{3}}))^{N(\bm{z})}(1-\mu(1+n_w^{-\frac{1}{3}}))^{c-N(\bm{z})},(\mu(1+n_w^{-\frac{1}{3}}))^{N(\bm{z})}(1-\mu(1-n_w^{-\frac{1}{3}}))^{c-N(\bm{z})} \Big].
\end{align*}
 \end{proof}

\begin{proof}[Proof of \cref{lemma_uniform_bound_probabilities}.]
The law of total probability yields the following decomposition
\begin{align*}
    \prob[(Z_1, \ldots, Z_c) = \bm{z}] &= \prob[(Z_1, \ldots, Z_c) = \bm{z}|\mathcal{E}^{\textsc{diff\_group}}] \prob[\mathcal{E}^{\textsc{diff\_group}}]\\
    &+\prob[(Z_1, \ldots, Z_c) = \bm{z}|\neg \mathcal{E}^{\textsc{diff\_group}}] \prob[\neg \mathcal{E}^{\textsc{diff\_group}}].
\end{align*}
By this decomposition and Lemma \ref{lemma_close_conditioned_on_diff_groups}, we can lower bound $\prob[(Z_1, \ldots, Z_c) = \bm{z}]$ by
\begin{align*}
    &\prob[(Z_1, \ldots, Z_c) = \bm{z}|\mathcal{E}^{\textsc{diff\_group}}] \prob[\mathcal{E}^{\textsc{diff\_group}}] \geq \underbrace{(\mu(1-n_w^{-\frac{1}{3}}))^{N(\bm{z})}(1-\mu(1+n_w^{-\frac{1}{3}}))^{c-N(\bm{z})} \prob[\mathcal{E}^{\textsc{diff\_group}}]}_{\textsc{LB}(n_w, \bm{z})}
\end{align*}
and upper bound the same probability by 
\begin{align*}
    \prob[(Z_1, \ldots, Z_c) = \bm{z}|\mathcal{E}^{\textsc{diff\_group}}]  &+ 1-\prob[\mathcal{E}^{\textsc{diff\_group}}] \nonumber \\
    &\leq \underbrace{(\mu(1+n_w^{-\frac{1}{3}}))^{N(\bm{z})}(1-\mu(1-n_w^{-\frac{1}{3}}))^{c-N(\bm{z})} +1-\prob[\mathcal{E}^{\textsc{diff\_group}}]}_{\textsc{UB}(n_w, \bm{z})}
\end{align*}
Therefore, 
$|\prob[(Z_1, \ldots, Z_c) = \bm{z}] - \mu^{N(\bm{z})} (1-\mu)^{c-N(\bm{z})}| \leq f(\bm{z}, n_w)$
where 
\begin{align*}
f(\bm{z}, n_w) = \max\Big(|\textsc{LB}(n_w, \bm{z}) -\mu^{N(\bm{z})}(1-\mu)^{c-N(\bm{z})}|,|\textsc{UB}(n_w, \bm{z}) -\mu^{N(\bm{z})}(1-\mu)^{c-N(\bm{z})}|\Big).
\end{align*}
Recall that $\prob[\sigma^{\textsc{random}} = \bm{z}] = \mu^{N(\bm{z})}(1-\mu)^{c-N(\bm{z})}$. Since $\prob[\mathcal{E}^{\textsc{diff\_group}}] \to_{n_w \to \infty} 1$ (\cref{lemma_diff_groups_whp}) and $n_w^{-\frac{1}{3}} \to 0$,  we obtain that $f(\bm{z}, n_w) \to_{n \to \infty} 0$ for any $\bm{z} \in \{0,1\}^c$ because
$$\textsc{UB}(n_w, \bm{z}), \textsc{LB}(n_w, \bm{z}) \to_{n_w \to \infty} \mu^{N(\bm{z})}(1-\mu)^{c-N(\bm{z})}.$$ Hence, $f(n_w) = \max_{\bm{z} \in \{0,1\}^c} f(\bm{z}, n_w)$ satisfies the desired property and concludes the proof. 
\end{proof}

\subsection{Random maximizes revenue in a class of review-oblivious policies}\label{app_subsec_proof_random_highest_revenue}

\begin{proposition}\label{prop:random_highest_revenue_amongst_class}
For any fixed price $p$ satisfying Assumption \ref{assumption:non_abs_non_degen} and any finite window $w$,
\begin{itemize}
    \item $\mathcal{U}(S_w)$ maximizes revenue over $\Delta(S_w)$;
    $\rev(\sigma^{\textsc{random}(w)},p) = \max_{\psi \in \Delta(S_w)} \rev(\sigma^{\textsc{random}(w, \psi)},p)$
    \item $\rev(\sigma^{\textsc{random}(w)},p)$ is strictly increasing in $w$
    \item $\lim_{w \to \infty} \rev(\sigma^{\textsc{random}(w)},p) = \rev(\random, p)$.
\end{itemize}
\end{proposition}
To show the proposition we show three useful lemmas. 
The first lemma (proof in \ref{subsubsec_proof_lemma_uniform_dist_best_window_w}) shows that $\mathcal{U}(\mathcal{S}_{w})$ maximizes the revenue over all distributions $\psi$ over $\mathcal{S}_{w}$, which suggests for any fixed window $w$, more randomness in the ordering implies more revenue.

\begin{lemma}\label{lemma: entropy_max_revenue}
For any  fixed price $p$ satisfying Assumption \ref{assumption:non_abs_non_degen} the uniform distribution $\mathcal{U}(\mathcal{S}_{w})$ maximizes the revenue over $\Delta (\mathcal{S}_{w})$, i.e.,
$\textsc{Rev}(\sigma^{\textsc{random}(w)},p) = \max_{\psi \in \Delta (\mathcal{S}_{w})} \textsc{Rev}(\sigma^{\textsc{random}(w,\psi)},p)$.
\end{lemma}

The next lemma (proof in Section \ref{subsubsec_proof_lemma_rev_w_strictly_increasing}) shows the revenue of $\sigma^{\textsc{random}(w)}$ is strictly increasing in the window size $w$ for any price $p$ satisfying Assumption \ref{assumption:non_abs_non_degen}.

\begin{lemma}\label{lemma:rev_w_strictly_increasing}
    For any fixed price $p$ satisfying Assumption \ref{assumption:non_abs_non_degen}, $\rev(\sigma^{\textsc{random}(w)},p)$ is strictly increasing in the window $w$.
\end{lemma}

The third lemma (proof in Section \ref{subsubsec_lemma_limit_rev_w_to_infty}) shows that as the window $w$ goes to infinity the revenue of $\sigma^{\textsc{random}(w)}$ approaches the revenue of $\random$.

\begin{lemma}\label{lemma:limit_rev_w_to_infty}
   For any fixed price $p$ satisfying Assumption \ref{assumption:non_abs_non_degen}, $\lim_{w \to \infty} \rev(\sigma^{\textsc{random}(w)},p) = \rev(\random, p)$. 
\end{lemma}

\begin{proof}[Proof of \cref{prop:random_highest_revenue_amongst_class}]
The proof follows by combining Lemmas~\ref{lemma: entropy_max_revenue}, \ref{lemma:rev_w_strictly_increasing}, and \ref{lemma:limit_rev_w_to_infty}.
\end{proof}

\subsubsection{Uniform distribution maximizes revenue (Lemma~\ref{lemma: entropy_max_revenue})}\label{subsubsec_proof_lemma_uniform_dist_best_window_w}
To prove \cref{lemma: entropy_max_revenue}, we first characterize the stationary distribution of the reviews generated by $\sigma^{\textsc{random}(w, \psi)}$ in way analogous to \cref{lemma:stationary_state_distribiton_newest_first}. In particular for any state of the $w$ most recent review ratings $\bm{z} \in \{0,1\}^w$, the purchase probability is $q_{\bm{z}}^{\psi} \coloneqq \sum_{S \in \mathcal{S}_{w}} \psi(S) \prob_{\Theta \sim \mathcal{F}}[\Theta + h(\sum_{i \in S} z_i) \geq p]$.
A useful quantity is the inverse purchase rate conditioned on $k$ positive reviews:
$$\iota_k^{\psi} =  \frac{1}{\binom{w}{k}}\sum_{\bm{z} \in \{0,1\}^w: N_{\bm{z}} = k} \frac{1}{q_{\bm{z}}^{\psi}}.$$  Intuitively, $\iota_k^{\psi}$ is the average number of rounds the process of the most recent $w$ reviews spends at a review state with $k$ positive reviews. 
The next lemma (proof in \cref{app:purchase_rate_at_k_inequality}) shows that the uniform distribution minimizes this purchase rate for any $k$ and characterizes the equality conditions.

\begin{lemma}\label{lemma: purchase_rate_at_k_inequality}
    For any $k \in \{0,\ldots, c\}$ and any probability distribution $\psi$, the inverse purchase rate under $\psi$ is at least the inverse purchase rate under $\mathcal{U}(\mathcal{S}_w)$, i.e., $\iota^\psi_k \geq \iota^{\mathcal{U}(\mathcal{S}_w)}_k$. Equality $\iota^\psi_k = \iota^{\mathcal{U}(\mathcal{S}_w)}_k$ is achieved if and only if $q_{\bm{z}}^{\psi} =q_{\bm{z}'}^{\psi}$ for all states $\bm{z}, \bm{z}' \in \{0,1\}^w$ with $N_{\bm{z}} = N_{\bm{z}'} =k$.
\end{lemma}

Let $\bm{Z}_t = (Z_{t,1}, \ldots, Z_{t,w})$ be the most recent $w$ reviews at round $t$. We show that $\bm{Z}_t$ is a time-homogenous Markov chain on $\{0,1\}^w$. If $\bm{Z}_t$ is at state $\bm{z} \in \{0,1\}^w$, a purchase occurs with probability~$q_{\bm{z}}^{\psi}$. If a purchase occurs, a new review is left and the state transitions to $\bm{Z}_{t+1} = (1, z_1, \ldots, z_{w-1})$ if the review is positive (with probability $\mu$) and to $\bm{Z}_{t+1} = (0, z_1, \ldots, z_{w-1})$ if the review is negative (with probability $1-\mu$). If there is no purchase, the state remains the same $(\bm{Z}_{t+1}=\bm{Z}_{t})$. Given that $p$ is a non-absorbing price (Assumption \ref{assumption:non_abs_non_degen}), for every state of reviews $\bm{z} \in \{0,1\}^w$, the purchase probability $q_{\bm{z}}^{\psi}$ is positive and the probability of any new review is strictly positive (since $\mu \in (0,1)$). Then $\bm{Z}_t$ can reach every state from every other state with positive probability (i.e. it is a single-recurrence-class Markov chain with no transient states), and hence $\bm{Z}_t$ has a unique stationary distribution denoted by $\pi$. Our next lemma characterizes the form of this stationary distribution. For a a vector of review ratings $\bm{z} \in \{0,1\}^w$, $N_{\bm{z}} = \sum_{i=1}^w z_i$ denotes the number of positive ratings.

\begin{lemma}\label{lemma:stationary_dist_window_W} For any price $p$ satisfying Assumption \ref{assumption:non_abs_non_degen} and any distribution $\psi$, the stationary distribution of the Markov chain $\bm{Z}_t$ under the review ordering policy $\sigma^{\textsc{random}(w,\psi)}$ is given by
$$\pi_{\bm{z}}^{\psi} = \kappa^{\psi} \cdot \frac{\mu^{N_{\bm{z}}}(1-\mu)^{w-N_{\bm{z}}}}{q_{\bm{z}}^{\psi}} \quad \text{ where } 
\kappa^{\psi} = \frac{1}{\expect_{K \sim \binomial(w, \mu)}\big[\iota_{K}^{\psi}\big]}.
$$ 
\end{lemma}
\begin{proof}
Similar to the proof of \cref{lemma:stationary_state_distribiton_newest_first}, we invoke 
 \cref{lemma: general_theorem_markov_chains_stationary}. Recall that \cref{lemma: general_theorem_markov_chains_stationary} starts with a Markov chain $\mathcal{M}$ with state space $\mathcal{S}$, transition matrix $M$, and stationary distribution $\{\pi(s)\}_{s \in \mathcal{S}}$. For any function $f$ on $\mathcal{S}$, it then transforms $\mathcal{M}$ into a new Markov chain $\mathcal{M}_f$ which remains at every state $s$ with probability $1-f(s)$ and transitions according to $M$ otherwise. The lemma establishes that the stationary distribution of $\mathcal{M}_f$ is given by $\pi_{f} = \{\kappa \cdot \frac{\pi(s)}{f(s)}\}_{s \in \mathcal{S}}$ where $\kappa = 1/\sum_{s \in \mathcal{S}}\frac{\pi(s)}{f(s)}$.
 
In the language of \cref{lemma: general_theorem_markov_chains_stationary}, $\bm{Z}_t$ corresponds to $\mathcal{M}_f$, the state space $\mathcal{S}$ to $\{0,1\}^w$, and $f$ is a function that expresses the purchase probability at a given state, i.e., $f(\bm{z}) = q_{\bm{z}}^{\psi}$. Note that with probability $1-f(\bm{z})$, $\bm{Z}_t$ remains at the same state (as there is no purchase). 

To apply \cref{lemma: general_theorem_markov_chains_stationary}, we need to show that whenever there is a purchase, $\bm{Z}_t$ transitions according to a Markov chain with stationary distribution $\mu^{N_{\bm{z}}}(1-\mu)^{w-N_{\bm{z}}}$. Consider the Markov chain $\mathcal{M}$ which always replaces the $w$-th last review with a new $\bern(\mu)$ review. This process has stationary distribution equal to the above numerator and $\bm{Z}_t$ transitions according to $\mathcal{M}$ upon a purchase, i.e., with probability $f(\bm{z})$. Hence, by \cref{lemma: general_theorem_markov_chains_stationary}, a stationary distribution for $\bm{Z}_t$ is 
$$\pi_{\bm{z}}^{\psi} = \kappa^{\psi} \cdot \frac{\mu^{N_{\bm{z}}}(1-\mu)^{w-N_{\bm{z}}}}{q_{\bm{z}}^{\psi}} \quad \text{ where } \kappa^{\psi} = \frac{1}{\expect_{Z_1, \ldots, Z_{w} \sim_{i.i.d.} \bern(\mu)}\Big[\frac{1}{q_{(Z_1, \ldots, Z_w)}^{\psi}} \Big]}.$$
This is the unique stationary distribution as $\bm{Z}_t$ is irreducible and aperiodic. Expanding over the number of positive reviews $k \in \{0, \ldots, w\}$, the lemma follows as the expectation in the denominator of $\kappa^{\psi}$ can be expressed as
$$\sum_{k=0}^w \Bigg[ \sum_{\bm{z} \in \{0,1\}^w: N_{\bm{z}} = k} \frac{1}{q_{\bm{z}}^{\psi}} \Bigg] \mu^{k}(1-\mu)^{w-k} =  \sum_{k=0}^w \iota^\psi_k \binom{w}{k}\mu^k (1-\mu)^{w-k} = \expect_{K \sim \binomial(w, \mu)}\big[\iota_{K}^{\psi}\big].$$
\end{proof}

Having established the stationary distribution of $\bm{Z}_t$ we give an expression for the revenue of $\sigma^{\textsc{random}(w,\psi)}$ (similar to \cref{theorem:most_recent_C_revenue}).

\begin{lemma}\label{lemma:revenue_sliding_window_W}
For any price $p$ satisfying Assumption \ref{assumption:non_abs_non_degen} and any distribution $\psi \in \Delta(\mathcal{S}_{w})$, the revenue of $\sigma^{\textsc{random}(w,\psi)}$ is given by $\textsc{Rev}(\sigma^{\textsc{random}(w,\psi)}, p) =p \cdot \kappa^{\psi}  \quad \text{where}\quad \kappa^{\psi} = \frac{1}{\expect_{K \sim \binomial(w, \mu)}\big[\iota_{K}^{\psi}\big]}$.
\end{lemma}

\begin{proof}
Using Eq.~\eqref{equation:reveneue_def} and the Ergodic theorem, the revenue of $\sigma^{\textsc{random}(w,\psi)}$ is
   \begin{equation*}
  \textsc{Rev}_1(\sigma^{\textsc{random}(w,\psi)}, p) = \liminf_{T \to \infty} \frac{\expect[\sum_{t=1}^T p \mathbf{1}_{\Theta_t + h(\Phi_t) \geq p}]}{T} = p\liminf_{T \to \infty} \frac{\expect[\sum_{t=1}^T q_{\bm{Z}_t}^{\psi}]}{T} = p\sum_{\bm{z} \in \{0,1\}^{w}} \pi_{\bm{z}}^{\psi}q_{\bm{z}}^{\psi}.\end{equation*}
   The second equality uses that the ex-ante purchase probability for review state $\bm{Z}_t$ is $\expect[\mathbf{1}_{\Theta_t + h(\Phi_t) \geq p}] =q_{\bm{Z_t}}^{\psi}$
   and the law of iterated expectation. The third equality expresses the revenue of the stationary distribution via the Ergodic theorem. Expanding $\pi_{\bm{z}}^{\psi}$ by Lemma~\ref{lemma:stationary_dist_window_W}, the $q_{\bm{z}}^{\psi}$ term cancels out: 
    \begin{align*}
     \textsc{Rev}_1(\sigma^{\textsc{random}(w,\psi)}, p)
       &= p \cdot \kappa^{\psi} \cdot \Big[\sum_{\bm{z} \in \{0,1\}^{w}} \mu^{N_{\bm{z}}} (1-\mu)^{w-N_{\bm{z}}} \Big].
    \end{align*}
    The proof is concluded by noting that the term in the square brackets equals 1, since it is the sum over all probabilities of $\binomial(w, \mu)$.
\end{proof} 

\begin{proof}[Proof of \cref{lemma: entropy_max_revenue}.]
For $\kappa^{\psi} = \frac{1}{\expect_{K \sim \binomial(w, \mu)}\big[\iota_{K}^{\psi}\big]}$, the expectation in $\frac{1}{\kappa^{\psi}}$ can be expressed by summing over the number of positive review ratings $k \in \{0, 1, \ldots, w\}$:
\begin{equation*}
     \frac{1}{\kappa^{\psi}} =  \sum_{k=0}^w \iota^\psi_k \binom{w}{k}\mu^k (1-\mu)^{w-k} \geq \sum_{k=0}^w \iota^{\mathcal{U}(\mathcal{S}_w)}_k \binom{w}{k}\mu^k (1-\mu)^{w-k} = \frac{1}{\kappa^{\mathcal{U}(\mathcal{S}_w)}}.
\end{equation*}
    where the inequality follows from \cref{lemma: purchase_rate_at_k_inequality}. Thus, $\kappa^{\psi} \leq \kappa^{\mathcal{U}(\mathcal{S}_w)}$ for any distribution $\psi$. 
    \cref{lemma:revenue_sliding_window_W} then implies that 
    $\textsc{Rev}_1(\sigma^{\textsc{random}(w,\psi)}, p) = p \cdot \kappa^{\psi} \leq p \cdot \kappa^{\mathcal{U}(\mathcal{S}_w)} =\textsc{Rev}_1(\sigma^{\textsc{random}(w)}, p).$ \end{proof}

\subsubsection{Uniform distribution minimizes the inverse purchase rate (Lemma~\ref{lemma: purchase_rate_at_k_inequality})}\label{app:purchase_rate_at_k_inequality}
The following lemma characterizes the inverse purchase rate for $k$ positive review ratings under the uniform distribution, i.e., $\iota^{\mathcal{U}(\mathcal{S}_w)}_k$. To ease notation, let $q(n) \coloneqq \prob_{\Theta \sim \mathcal{F}}[\Theta + h(n) \geq p]$ be the purchase probability given $n \in \{0,1 \ldots, c\}$ positive reviews.

\begin{lemma}\label{lemma:inverse_purchase_rate_k_uniform}
For any number of positive review ratings $k \in \{0, \ldots, w\}$, the inverse purchase rate for $k$ under the uniform distribution on $\mathcal{S}_w$ is given by
$\iota^{\mathcal{U}(\mathcal{S}_w)}_k = \frac{\binom{w}{k}}{\sum_{n= \max(0,c+k-w)}^{\min(c,k)}q(n) \binom{c}{n} \binom{w-c}{k-n} }$.
\end{lemma}

\begin{proof}[Proof of \cref{lemma:inverse_purchase_rate_k_uniform}.]
Recall that $\mathcal{U}(\mathcal{S}_w)$ places a probability of $\frac{1}{\binom{w}{c}}$ on every subset of $\mathcal{S}_w$. Thus, 
by counting the number of occurrences of each term $q(n)$, the purchase probability at any state $\bm{z} \in \{0,1\}^w$ with $N_{\bm{z}} = k$ can be expressed as:
\begin{equation}\label{eq:purchase_prob_z_uniform_indep_z}
q_{\bm{z}}^{\mathcal{U}(\mathcal{S}_w)} = \sum_{S \in \mathcal{S}_w}\frac{1}{\binom{w}{c}} q(\sum_{i \in S} z_i) = \sum_{n= \max(0,c+k-w)}^{\min(c,k)} q(n) \frac{\binom{k}{n}\binom{w-k}{c-n}}{\binom{w}{c}}= \sum_{n= \max(0,c+k-w)}^{\min(c,k)} q(n) \frac{\binom{c}{n}\binom{w-c}{k-n}}{\binom{w}{k}}. 
\end{equation}
The second equality uses that for any number of positive reviews $n \in [\max(0,c+k-w), \min(c,k)]$ there are $\binom{k}{n} \binom{w-k}{c-n}$ ways to choose $S \in \mathcal{S}_w$ with $\sum_{i \in S} z_i = n$: $\binom{k}{n}$ ways to choose the $n$ indices $i \in S$ from the $k$ indices $i$ where $z_i = 1$, and $\binom{w-k}{c-n}$ ways to choose the remaining $c-n$ indices $i \in S$ from the $w-k$ indices where $z_i = 0$. The third equality uses the binomial coefficient identity 
$$\frac{\binom{k}{n}\binom{w-k}{c-n}}{\binom{w}{c}} = \frac{\frac{k!}{n!(k-n)!} \cdot \frac{(w-k)!}{(c-n)!(w-k-c+n)!}}{\frac{w!}{c!(w-c)!}} = \frac{\frac{c!}{n!(c-n)!} \cdot \frac{(w-c)!}{(k-n)!(w-k-c+n)!}}{\frac{w!}{k!(w-k)!}} = \frac{\binom{c}{n}\binom{w-c}{k-n}}{\binom{w}{k}}.$$
Given that the right-hand side of \cref{eq:purchase_prob_z_uniform_indep_z} does not depend on the review ratings $\bm{z}$,
$$\iota^{\mathcal{U}(\mathcal{S}_w)}_k =  \frac{1}{\binom{w}{k}}\sum_{\bm{z} \in \{0,1\}^w: N_{\bm{z}} = k} \frac{1}{q_{\bm{z}}^{\mathcal{U}(\mathcal{S}_w)}} = \frac{1}{\sum_{n= \max(0,c+k-w)}^{\min(c,k)} q(n) \frac{\binom{c}{n}\binom{w-c}{k-n}}{\binom{w}{k}}} =\frac{\binom{w}{k}}{\sum_{n= \max(0,c+k-w)}^{\min(c,k)}q(n) \binom{c}{n} \binom{w-c}{k-n} }$$
which completes the proof.
\end{proof}

\begin{proof}[Proof of \cref{lemma: purchase_rate_at_k_inequality}.]
   By Jensen's inequality applied to the convex function $x \to \frac{1}{x}$,  \begin{equation}\label{ineq:jensen_inverse_purchase_rate_k}
    \iota^\psi_k =  \frac{1}{\binom{w}{k}} \sum_{\bm{z} \in \{0,1\}^w: N_{\bm{z}} = k} \frac{1}{q_{\bm{z}}^{\psi}} \geq \frac{\binom{w}{k}}{\sum_{\bm{z} \in \{0,1\}^w: N_{\bm{z}} = k} q_{\bm{z}}^{\psi}}.
   \end{equation}
   To show the lemma, it suffices to prove that for any probability distribution $\psi  \in \Delta(\mathcal{S}_w)$, the lower bound term $\frac{\binom{w}{k}}{\sum_{\bm{z} \in \{0,1\}^w: N_{\bm{z}} = k} q_{\bm{z}}^{\psi}}$ on the left-hand side is independent of $\psi$ and equal to $\iota^{\mathcal{U}(\mathcal{S}_w)}_k$. We start by analyzing the sum in the denominator $\sum_{\bm{z} \in \{0,1\}^w: N_{\bm{z}} = k} q_{\bm{z}}^{\psi}$. 
   Letting $q(n) \coloneqq \prob_{\Theta \sim \mathcal{F}}[\Theta + h(n) \geq p]$ be the purchase probability given $n \in \{0,1 \ldots, c\}$ positive reviews, recalling that $ q_{\bm{z}}^{\psi} = \sum_{S \in \mathcal{S}_w} \psi(S) q(\sum_{i \in S} z_i)$, and rearranging, we express the sum of interest as:
   \begin{equation*}
\sum_{\bm{z} \in \{0,1\}^w: N_{\bm{z}} = k} q_{\bm{z}}^{\psi} = \sum_{\bm{z} \in \{0,1\}^w: N_{\bm{z}} = k} \sum_{S \in \mathcal{S}_w} \psi(S) q(\sum_{i \in S} z_i) = \sum_{S \in \mathcal{S}_w} \psi(S) \sum_{\bm{z} \in \{0,1\}^w: N_{\bm{z}} = k}  q(\sum_{i \in S} z_i). 
   \end{equation*} 
By counting the number of occurrences of each term $q(n)$, the inner sum can be expressed as:
   \begin{equation}\label{eq:sum_of_q_independent_of_S}
\sum_{\bm{z} \in \{0,1\}^w: N_{\bm{z}} = k}  q(\sum_{i \in S} z_i)  = \sum_{n= \max(0,c+k-w)}^{\min(c,k)}q(n) \binom{c}{n} \binom{w-c}{k-n}
   \end{equation}
   since for any number of positive reviews $n \in [\max(0,c+k-w), \min(c,k)]$ there are $\binom{c}{n} \binom{w-c}{k-n}$ ways to choose sequences $\bm{z} \in \{0,1\}^w$ with $N_{\bm{z}} = k$ such that $\sum_{i \in S} z_i = n$: there are $\binom{c}{n}$ ways to choose the subsequence $(z_i)_{i \in S} \in \{0,1\}^c$ such that $\sum_{i \in S} z_i = n$ and $\binom{w-c}{k-n}$ ways to choose the complementary subsequence $(z_i)_{i \not \in S} \in \{0,1\}^{w-c}$ such that $\sum_{i \not \in S} z_i = k-n$. 
   
Observing that the right-hand side of \eqref{eq:sum_of_q_independent_of_S} is independent of $S$, using the fact that $\sum_{S \in \mathcal{S}_w} \psi(S) = 1$ (as $\psi$ is a probability distribution) as well as inequality \cref{ineq:jensen_inverse_purchase_rate_k} and \cref{lemma:inverse_purchase_rate_k_uniform}, we obtain:
\begin{equation*}
\iota^{\psi}_k \geq \frac{\binom{w}{k}}{\sum_{\bm{z} \in \{0,1\}^w: N_{\bm{z}} = k} q_{\bm{z}}^{\psi}} = \frac{\binom{w}{k}}{\sum_{n= \max(0,c+k-w)}^{\min(c,k)}q(n) \binom{c}{n} \binom{w-c}{k-n} } = \iota^{\mathcal{U}(\mathcal{S}_w)}_k. 
\end{equation*}

This completes the proof of the inequality statement of the lemma.
With respect to the equality statement, for any $k \in \{0,\ldots, c\}$, the equality conditions of Jensen's inequality imply that  \cref{ineq:jensen_inverse_purchase_rate_k} holds with equality if and only if $q_{\bm{z}}^{\psi} = q_{\bm{z}'}^{\psi}$ for any $\bm{z}, \bm{z}' \in \{0,1\}^w$ with $N_{\bm{z}}= N_{\bm{z}'} = k$.
\end{proof} 

\subsubsection{Revenue of window-random policies is increasing in window size (Lemma~\ref{lemma:rev_w_strictly_increasing})}\label{subsubsec_proof_lemma_rev_w_strictly_increasing}
We first show the suboptimality of the distribution $\psi_{w} \in \Delta(\mathcal{S}_{w+1})$ which ignores the $(w+1)$-th most recent review and places equal probability on every $c$-sized subset of $\{1, \ldots, w\}$.

\begin{lemma}\label{lemma: psi_2_w_plus_1_suboptimal}
    For any price $p$ satisfying Assumption \ref{assumption:non_abs_non_degen}, the revenue of $\psi_{w}$ is strictly suboptimal, i.e.,  $\textsc{Rev}(\sigma^{\textsc{random}(w+1, \psi_{w})},p)< \max_{\psi \in \Delta (\mathcal{S}_{w+1})} \textsc{Rev}(\sigma^{\textsc{random}(w+1,\psi)},p)$.
\end{lemma}

\begin{proof}
By \cref{lemma:revenue_sliding_window_W}, the revenue of $\psi_{w}$ can be expressed as:
    $$\textsc{Rev}_1(\sigma^{\textsc{random}(w+1,\psi_{w})}, p) = p \cdot \kappa^{\psi_{w}} = \frac{p}{ \sum_{k=0}^w \iota^{\psi_{w}}_k \binom{w}{k} \mu^k (1-\mu)^{w-k}}.$$
    By \cref{lemma: purchase_rate_at_k_inequality}, it is thus sufficient to show that there is some number of positive ratings $k \in \{0, \ldots, c\}$ such that $\iota^{\psi_{w}}_k > \iota^{\mathcal{U}(\mathcal{S}_w)}_ k$. For the sake of contradiction assume this is not the case, i.e., $\iota^{\psi_{w}}_k = \iota^{\mathcal{U}(\mathcal{S}_w)}_ k$ for all $k \in \{0,\ldots, c\}$. The equality condition of \cref{lemma: purchase_rate_at_k_inequality} applied for each $k \in \{0,\ldots, c\}$ implies \begin{equation}\label{eq:equality_condtion_implied_by_lemma_inverse_purchase_k}
        q_{\bm{z}}^{\psi_{w}} =q_{\bm{z}'}^{\psi_{w}} \text{ for any } \bm{z}, \bm{z}' \in \{0,1\}^w \text{ with } N_{\bm{z}} = N_{\bm{z}'}.
    \end{equation}
    Letting $q(n) \coloneqq \prob_{\Theta \sim \mathcal{F}}[\Theta + h(n) \geq p]$ denote the purchase probability given $n$ positive reviews, we show below (by induction on $n$) that \eqref{eq:equality_condtion_implied_by_lemma_inverse_purchase_k} implies $q(n) = q(0)$ for any number of positive review ratings $n \in \{0,\ldots, c\}$, which contradicts the fact that the price $p$ is non-degenerate. 
   
   For the base case of the induction ($n=1$), consider two settings with one positive review: in the first one the positive review is the $(w+1)$-th most recent review, while in the second one the positive review is the $w$-th most recent review. Formally, we apply \cref{eq:equality_condtion_implied_by_lemma_inverse_purchase_k} for $\bm{z} = (0, \ldots,0, 0,1)$  and $\bm{z}' = (0, \ldots, 0,1, 0)$. As $\psi_{w}$ places zero mass on any subset containing the last review implies that, the purchase probability at $\bm{z}$ is $q_{\bm{z}}^{\psi_{w}} = q(0)$. The purchase probability at $\bm{z}'$ equals $q(1)$ when the $w$-th review is chosen and $q(0)$ otherwise implying $q_{\bm{z}'}^{\psi_{w}} =q(1) \prob_{S \sim \psi_w}[w \in S] + q(0)\prob_{S \sim \psi_w}[w \not \in S)]$. Observing that $\prob_{S \sim \psi_w}[w \in S] > 0$ and solving the equality $q_{\bm{z}}^{\psi_{w}} = q_{\bm{z}'}^{\psi_{w}}$ implies $q(1) = q(0)$.
        
    For the induction step ($n-1 \rightarrow n$), suppose that $q(n') = q(0)$ for all $n' < n$. We apply \cref{eq:equality_condtion_implied_by_lemma_inverse_purchase_k} for $\bm{z} = (0, \ldots, 0, 1, \ldots, 1)$ ($w-n$ zeros followed by $n$ ones) and $\bm{z}' = (0, \ldots, 0, 1, \ldots, 1, 0)$ ($w-n-1$ zeros, followed by $n$ ones, followed by one). If the state of the reviews is $\bm{z}$, as $\psi_w$ never selects the $(w+1)$-th most recent review, any $c$ reviews selected by $\psi_w$ contain at most $n - 1$ positive review ratings and the induction hypothesis implies that the purchase probability is thus equal to $q(0)$ regardless of the choice of the $c$ reviews. Thus, $q_{\bm{z}}^{\psi_{w}} = q(0)$. Suppose the state of the review is $\bm{z}'$.  If the selected set of the $c$ selected reviews contains all the reviews at indices $I_{n} = \{w-n,w-n+1 \ldots, w\}$, the number of positive reviews is $n$ and the purchase probability $q(n)$. Otherwise, the $c$ selected reviews contain at most $n -1$ positive review ratings and by the induction hypothesis the purchase probability is $q(0)$. Thus,    $q_{\bm{z}}^{\psi_{w}} = q(n) \prob_{S \sim \psi_w}[I_{n} \subseteq S] +q(0) \prob_{S \sim \psi_w}[I_{n} \not \subseteq S]$. Observing that $\prob_{S \sim \psi_w}[I_{n} \subseteq S] > 0$ and solving the equality $q_{\bm{z}}^{\psi_{w}} = q_{\bm{z}'}^{\psi_{w}}$ implies $q(n) = q(0)$, which finishes the induction step and the proof. 
\end{proof}

\begin{proof}[Proof of \cref{lemma:rev_w_strictly_increasing}.]
Note that $\sigma^{\textsc{random}(w+1, \psi_{w})} = 
\sigma^{\textsc{random}(w)}$. By \cref{lemma: entropy_max_revenue} and \cref{lemma: psi_2_w_plus_1_suboptimal},   $\textsc{Rev}(\sigma^{\textsc{random}(w)},p) $ is strictly increasing in $w$ as for any window $w \geq c$,
\begin{align*}
    \textsc{Rev}(\sigma^{\textsc{random}(w)},p) &= \textsc{Rev}(\sigma^{\textsc{random}(w+1, \psi_{w})},p)\\
    &< \max_{\psi \in \Delta (\mathcal{S}_{w+1})} \textsc{Rev}(\sigma^{\textsc{random}(w+1,\psi)},p) = \textsc{Rev}(\sigma^{\textsc{random}(w+1)},p).
\end{align*}
\end{proof}

\subsubsection{Window-random revenue converges to revenue of Random (Lemma~\ref{lemma:limit_rev_w_to_infty})}\label{subsubsec_lemma_limit_rev_w_to_infty}
Given that  $\textsc{Rev}(\sigma^{\textsc{random}(w)}, p)$ can be expressed as a function of a $w$ i.i.d. $\bern(\mu)$ trials (\cref{lemma:revenue_sliding_window_W}), we first show that $K \sim \binomial(w, \mu)$ with $w$ trials and success probability
 $\mu$ is well-concentrated around its mean. Formally, the event $\mathcal{E}^{\textsc{Conc}}(w)= \{ K \in [ \mu(w - w^{\frac{2}{3}}),\mu(w + w^{\frac{2}{3}})]\}$ occurs with high probability. 

\begin{lemma}\label{lemma:event_K_concentrates_around_mu_w}
    For any $w \geq c$, $\prob[\mathcal{E}^{\textsc{Conc}}(w)] \geq 1-2 \exp \Big(\frac{w^{-\frac{1}{3}}\mu}{3} \Big)$ and thus $\lim_{w \to \infty}\prob[\mathcal{E}^{\textsc{Conc}}(w)] = 1$. 
\end{lemma}

\begin{proof}
By Chernoff bound, $\prob_{K \sim \binomial(w, \mu)}\Big[|K-\mu w| \leq \mu w^{\frac{2}{3}} \Big]\geq 1-2 \exp \Big(\frac{w^{-\frac{1}{3}}\mu}{3} \Big)$ . Taking the limit as $w\rightarrow \infty$ yields the result.
\end{proof} 

Recall that $\iota^{\mathcal{U}(\mathcal{S}_w)}_k = \frac{1}{\binom{w}{c}} \sum_{\bm{z} \in \{0,1\}^w: N_{\bm{z}} = k} \frac{1}{q_{\bm{z}}^{\mathcal{U}(\mathcal{S}_w)}}$ is the inverse purchase rate conditioned on $k$ positive review ratings and $q(n) \coloneqq \prob_{\Theta \sim \mathcal{F}}[\Theta + h(n) \geq p]$ denotes the purchase probability given $n \in \{0,1 \ldots, c\}$ positive reviews. The following lemma connects the inverse purchase rate of $\sigma^{\textsc{random}(w)}$ (assuming that the above concentration holds) with the revenue of $\sigma^{\textsc{random}}$. 

\begin{lemma}\label{lemma:expectation_iota_concentration}
    There exists a threshold $M > 0$ such that for any window $w > M$, it holds that
  $$\rev(\sigma^{\textsc{random}}, p) = p \cdot \lim_{w \to \infty} \Bigg(\frac{1}{\expect_{K \sim \binomial(w, \mu)}\Big[\iota_{K}^{\mathcal{U}(\mathcal{S}_w)} \Big| \mathcal{E}^{\textsc{Conc}}(w) \Big]} \Bigg).$$
\end{lemma}

\begin{proof}
Let $M$ be large enough so that for $w > M$,
 $ \mu(w -w^{-\frac{2}{3}}) \geq c$ and $w \geq c  + \mu(w +w^{-\frac{2}{3}}) $, which implies that $\min(c, k) = c$ and $\max(0, c+k-w) = 0$ for $k \in [ \mu(w - w^{\frac{2}{3}}),\mu(w + w^{\frac{2}{3}})]$. By \cref{lemma:inverse_purchase_rate_k_uniform} for $w > M$ and  $k \in [ \mu(w - w^{\frac{2}{3}}),\mu(w + w^{\frac{2}{3}})]$, it holds that
\begin{equation*}\label{equation:iota_without_max_min_in_sum}
\iota^{\mathcal{U}(\mathcal{S}_w)}_k  =\frac{\binom{w}{k}}{\sum_{n= \max(0,c+k-w)}^{\min(c,k)}q(n) \binom{c}{n} \binom{w-c}{k-n} }  =\frac{1}{\sum_{n= \max(0,c+k-w)}^{\min(c,k)}q(n) \frac{\binom{c}{n} \binom{w-c}{k-n}}{\binom{w}{k}} } = \frac{1}{\sum_{n= 0}^{c}q(n) \frac{\binom{c}{n} \binom{w-c}{k-n}}{\binom{w}{k}} }. 
\end{equation*}
One of the terms in the denominator is the binomial coefficient ratio $\frac{\binom{w-c}{k-n}}{\binom{w}{k}}$. Letting $\underline{\mu}_{w} = \mu - w^{-\frac{1}{3}}$ and $\overline{\mu}_{w} = \mu + w^{-\frac{1}{3}}$, we show below
\begin{equation}\label{eq:binom_coeff_bounds}
(\frac{w-c+1}{w})^c  (\overline{\mu}_{w})^c (1-\underline{\mu}_{w})^{c-n} \geq \frac{\binom{w-c}{k-n}}{\binom{w}{k}} \geq (\underline{\mu}_{w} -\frac{n-1}{w})^n(1-\overline{\mu}_{w} -\frac{c-n-1}{w})^{c-n}
\end{equation}
which implies that $\expect_{K \sim \binomial(w, \mu)}\Big[\iota_{K}^{\mathcal{U}(\mathcal{S}_w)}|\mathcal{E}^{\textsc{Conc}}(w)\Big]$ is lower and upper bounded respectively by 
$$\frac{(\frac{w}{w-c+1})^c}{\sum_{n=0}^{c} q(n) \binom{c}{n} (\overline{\mu}_{w})^c(1-\underline{\mu}_{w})^{c-n}} \text{ and }\frac{1}{\sum_{n=0}^c q(n) \binom{c}{n} (\underline{\mu}_{w} -\frac{n-1}{w})^n(1-\overline{\mu}_{w} -\frac{c-n-1}{w})^{c-n}}.$$ Given that  $\underline{\mu}_{w}, \overline{\mu}_{w} \to_{w \to \infty} \mu$ (as $w^{-\frac{1}{3}} \to_{w \to \infty} 0$), $\frac{w-c+1}{w} \to_{w \to \infty} 1$ and $\frac{n-1}{w}, \frac{c-n-1}{w} \to_{w \to \infty} 0$, these lower and upper bounds converge to $\frac{1}{\expect_{N \sim \binomial(c, \mu)}[q(N)]}$ as $w \to \infty$ which is equal to $ \frac{p}{\rev(\random,p)}$  (by \cref{theorem:random_C_revenue}) and yields the result. We conclude the proof by showing Eq. \eqref{eq:binom_coeff_bounds}. We expand the binomial coefficient ratio in the denominator of Eq. \eqref{eq:binom_coeff_bounds} as:
$$ \frac{\binom{w-c}{k-n}}{\binom{w}{k}} = \frac{\frac{(w-c)!}{(k-n)!(w-c-k+n)!}}{\frac{w!}{k!(w-k)!}} = \frac{\prod_{i=1}^n(k-i+1) \prod_{j=1}^{c-n}(w-k-j+1)}{\prod_{l=1}^c(w-l+1)}.
$$
Observing that $k \geq k-i+1 \geq k-n+1$ for $k \in [1,n]$, $w-k \geq w-k-j+1 \geq w-k-c+n+1$ for $j \in [1, c-n]$, and $w \geq w-l+1 \geq w-c+1$ for $l \in [1,c]$, we can bound the binomial coefficient
\begin{equation*}
\underbrace{\Big(\frac{k}{w-c+1}\Big)^{n} \Big(\frac{w-k}{w-c+1} \Big)^{c-n}}_{ = (\frac{w-c+1}{w})^c (\frac{k}{w})^n (1-\frac{k}{w})^{c-n}} \geq \frac{\binom{w-c}{k-n}}{\binom{w}{k}} \geq \underbrace{\Big(\frac{k-n+1}{w} \Big)^n \Big(\frac{w-k-c+n+1}{w}\Big)^{c-n}}_{=\Big(\frac{k}{w}-\frac{n-1}{w} \Big)^n \Big(1-\frac{k}{w}-\frac{c-n-1}{w}\Big)^{c-n} } .
\end{equation*}
Given that $\frac{k}{w} \in [\underline{\mu}_{w},\overline{\mu}_{w}]$ (as $k \in [ \mu(w - w^{\frac{2}{3}}),\mu(w + w^{\frac{2}{3}})]$), the above proves \eqref{eq:binom_coeff_bounds}.
\end{proof} 

\begin{proof}[Proof of \cref{lemma:limit_rev_w_to_infty}.]
\cref{lemma:revenue_sliding_window_W} connects the revenue of $\sigma^{\textsc{random}(w)}$ to $\expect_{K \sim \binomial(w, \mu)}\Big[\iota^{\mathcal{U}(\mathcal{S}_w)}_K \Big]$. By the law ot total expectation, the latter term can be expanded to:
\begin{align*}
\underbrace{\expect_{K \sim \binomial(w, \mu)}\Big[\iota^{\mathcal{U}(\mathcal{S}_w)}_K|\mathcal{E}^{\textsc{Conc}}(w) \Big]}_{ \to_{w \to \infty} \frac{p}{\rev(\random,p)} \text{(\cref{lemma:expectation_iota_concentration})} }\underbrace{\prob[\mathcal{E}^{\textsc{Conc}}(w)]}_{ \to_{w \to \infty} 1
\text{ (\cref{lemma:event_K_concentrates_around_mu_w})}} + \expect_{K \sim \binomial(w, \mu)}\Big[\iota^{\mathcal{U}(\mathcal{S}_w)}_K|\neg \mathcal{E}^{\textsc{Conc}}(w) \Big] \underbrace{\prob[\neg \mathcal{E}^{\textsc{Conc}}(w)]}_{ \to_{w \to \infty} 0
\text{ (\cref{lemma:event_K_concentrates_around_mu_w})}} 
\end{align*}
We now show that $\iota^{\mathcal{U}(\mathcal{S}_w)}_k \leq \frac{1}{q(0)}$, which combined with \cref{lemma:event_K_concentrates_around_mu_w} implies that the second term goes to 0 as $w \to \infty$. As $h(n)$ is increasing, $q(n) \geq q(0)$ for all number of positive review ratings $n$ and thus $q_{\bm{z}}^{\mathcal{U}(\mathcal{S}_w)} = \frac{1}{\binom{w}{c}}\sum_{S \in \mathcal{S}_w} q(\sum_{i \in S} z_i) 
\geq q(0)$ for all review rating states $\bm{z} \in \{0,1\}^w$. Hence, 
$$\iota^{\mathcal{U}(\mathcal{S}_w)}_k =  \frac{1}{\binom{w}{k}}\sum_{\bm{z} \in \{0,1\}^w: N_{\bm{z}} = k} \frac{1}{q_{\bm{z}}^{\mathcal{U}(\mathcal{S}_w)}} \leq \frac{1}{q(0)}.$$
The proof is concluded by invoking \cref{lemma:revenue_sliding_window_W} and bounding the limit as $w\rightarrow \infty$ for the expansion of $\expect_{K \sim \binomial(w, \mu)}\Big[\iota^{\mathcal{U}(\mathcal{S}_w)}_K \Big]$ (using Lemmas \ref{lemma:event_K_concentrates_around_mu_w} and \ref{lemma:expectation_iota_concentration} as well as the boundedness of $\iota^{\mathcal{U}(\mathcal{S}_w)}_k$):  $$\textsc{Rev}_1(\sigma^{\textsc{random}(w)}, p) = \frac{p}{\expect_{K \sim \binomial(w, \mu)}\Big[\iota^{\mathcal{U}(\mathcal{S}_w)}_K \Big]} \to_{w \to \infty} \frac{p}{\frac{p}{\rev(\random, p)}} = \rev(\random, p). $$
\end{proof}

\section{Numerics with exponential idiosyncratic valuations (Remark~\ref{rem:exponential})}\label{appendix:numerics_dynamic_pricing_exponential_idiosyncratic}
In the main body, we considered uniform customer-specific idiosyncratic distributions. We now  compare the values of $\chi(\Pi^{\textsc{static}})$ and $\chi(\Pi^{\textsc{dynamic}})$ on a class of instances where the customer-specific distribution is exponential. We consider instances where the customer-specific valuation is an exponential distribution with parameter $\lambda$, i.e., $\mathcal{F} = \textsc{Exp}(\lambda)$. All other parameters of the instance remain the same as in Section~\ref{subsec:numerics_dynamic_vs_static_conf}. 
The quantity $\frac{1}{\lambda}$ equals the standard deviation of $\textsc{Exp}(\lambda)$, and thus represents the customer-specific valuation variability. We vary $\frac{1}{\lambda}$ on the interval $(0,3)$. We also vary the magnitude of $a$, which represents the strength of the prior belief. We plot the values of $\chi(\Pi^{\textsc{static}})$ and $\chi(\Pi^{\textsc{dynamic}})$ in Figure~\ref{fig:chi_pi_dynamic_chi_pi_static_comparison_exponential}. We plot the values of $\rev(\newest, \Pi^{\textsc{static}})$ and $\rev(\newest, \Pi^{\textsc{dynamic}})$ in Figure~\ref{fig:revenue_newest_static_dynamic_exponential}.

\begin{figure}[!htbp]
\centering
\includegraphics[width=1\textwidth]{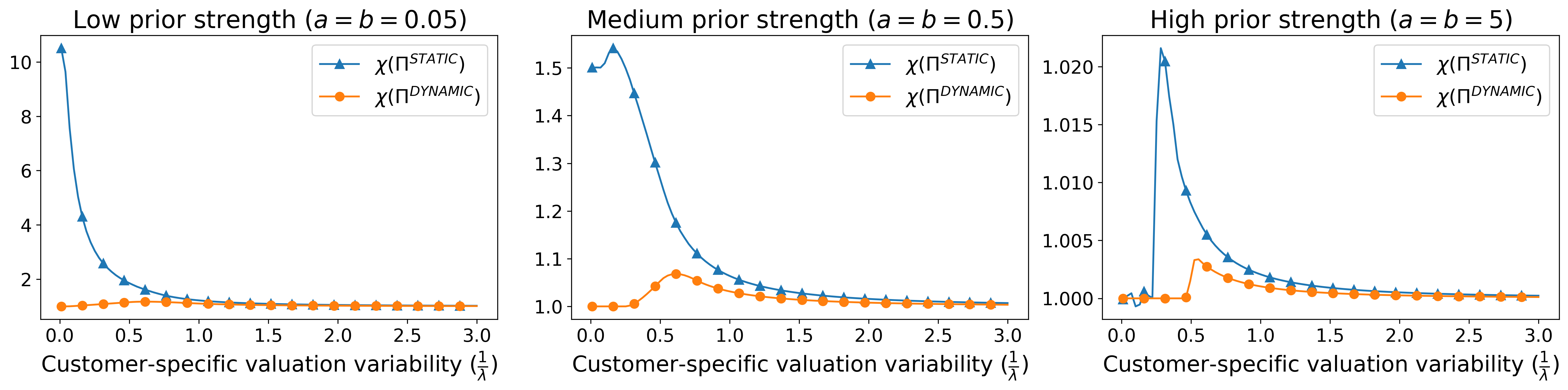}

\caption{Comparison of $\chi(\Pi^{\textsc{static}})$ and $\chi(\Pi^{\textsc{dynamic}})$.} 
\label{fig:chi_pi_dynamic_chi_pi_static_comparison_exponential}
\end{figure}

\begin{figure}[!htbp]
\centering
\includegraphics[width=1\textwidth]{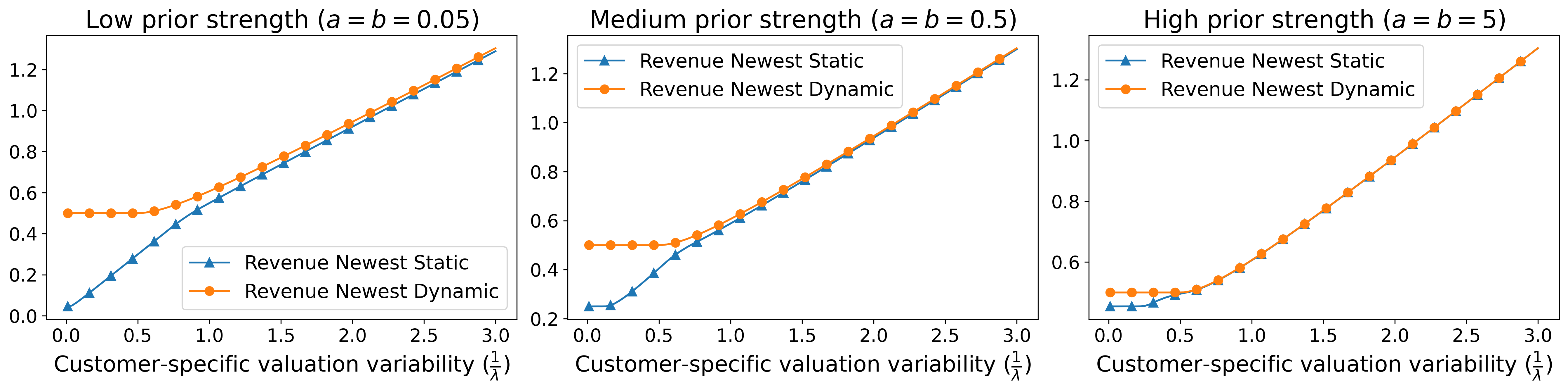}

\caption{Comparison of $\rev(\newest, \Pi^{\textsc{static}})$ and $\rev(\newest, \Pi^{\textsc{dynamic}})$.} 
\label{fig:revenue_newest_static_dynamic_exponential}
\end{figure}

We observe similar results to Section~\ref{subsec:numerics_dynamic_vs_static_conf}. In particular, the CoNF is small (less than $1.1$)  under dynamic pricing under all instances.
For static pricing, the CoNF can be large, especially when the prior strength is low and $\frac{1}{\lambda}$ is small. Importantly, we observe that $\chi(\Pi^{\textsc{dynamic}})$ is strictly smaller than $\chi(\Pi^{\textsc{static}})$ for all instances that were tested. In Figure~\ref{fig:revenue_newest_static_dynamic_exponential}, we see that the revenue of Newest under dynamic pricing can be significantly larger than the revenue of Newest under static pricing, especially when the prior strength is low and $\frac{1}{\lambda}$ is small. The difference between Newest's revenue under dynamic and static pricing decreases as $\frac{1}{\lambda}$ increases. This shows that our results from Section~\ref{subsec:numerics_dynamic_vs_static_conf} hold for exponential idiosyncratic valuations and are not tailored to uniform idiosyncratic valuations.

\section{Comparison to negative bias from quality variability (Section~\ref{sec:relatedwork})}\label{appendix:comparison_negative_bias_due_to_quality_variability}
\cite{decroix2021service} identifies a new explanation for how variability in service quality hurts a firm's revenue. To do so, they study a repeated interaction with a single customer who decides whether to purchase a service or not. The customer's purchase decision follows a logit choice model based on a belief formed by an exponential smoothing model (ESM) of the past service qualities experienced. They consider a benchmark in which the quality is not variable. They establish that when quality is variable, the customer's belief is negatively biased compared to the benchmark. Furthermore, they show that for any fixed-price policy, the revenue when quality is variable is smaller than the revenue under the benchmark. They establish that, under the optimal pricing policy, the revenue when quality is variable matches the revenue under the benchmark. The optimal pricing policy ensures that the purchase probability in each belief state is the same. There are two high-level similarities with our work. First, they also identify a negative bias in average beliefs driven by quality variability in experiences and the customer weighing recent experiences more highly. Second, they show that the optimal pricing policy ensures the same purchase probability in each state.

Our work fundamentally differs from theirs on both conceptual and technical levels. For the purposes of contrast, we think of the ESM in \cite{decroix2021service} and the customer's behavior under Newest in our work as analogous. The following aspects significantly differentiate our work from theirs.

\paragraph{Conceptual differences.} \cite{decroix2021service} studies repeated interactions with a \textit{single} customer who forms beliefs based on their previous experiences based on a fixed ESM model; the focus is on how quality variability impacts revenue under this setting.\footnote{They also provide an extension (Section 7.1 in \cite{decroix2021service}) where multiple customers simultaneously interact and show that such interactions mitigate the negative bias (as experiences aggregate across customers).} In contrast, we study a stream of \textit{different} customers who leave reviews; we ask how the review ordering policy impacts revenue. In our setting, the platform chooses a review ordering policy that shapes how customers behave. Given this review ordering, the customers use the $c$ reviews displayed by the platform to form their belief (Section~\ref{sec:model}). Thus, the platform's review ordering policy \textit{influences} the customers' beliefs.

Moreover, the benchmark considered in our work is conceptually different. The benchmark in \cite{decroix2021service} assumes that the \textit{quality is constant (not variable)}. Instead, our benchmark keeps all the model parameters the same as Newest and only changes the ordering policy to Random. Random is an implementable review-ordering policy and it is the revenue-maximizing review ordering policy under a class of ordering policies which do not take the review ratings into account (Appendix~\ref{appendix:why_random_right_benchmark}). As we illustrate in the technical differences, our benchmark is stronger than the one in \cite{decroix2021service}.

\paragraph{Technical differences. } 
On a technical level, our negative bias result is not an extension or corollary of the result in \cite{decroix2021service}, which relies heavily on the convexity property of the logit model. In particular, there exist models in which the revenue under Newest is arbitrarily larger than the benchmark in \cite{decroix2021service} as we demonstrate in Proposition~\ref{prop:newest_arbitrarily_better_than_knowledgeable_static_price} (Section~\ref{subsec:negative_bias_with_respect_to_no_variability_benchmark_does_not_always_hold_in_our_model}). In contrast, negative bias with respect to our random ordering benchmark does not require any assumptions on the purchase probability (other than monotonicity) and holds irrespective of the behavioral model. 
    
Furthermore, our optimality results for dynamic pricing are fundamentally new and hold against a stronger benchmark.
Our Random benchmark for the optimal dynamic pricing is always at least the benchmark in \cite{decroix2021service} and can be larger by a factor of $\frac{4}{3}$ (Propositions \ref{prop:random_benchmark_is_always_at_least_no_quality_variability_benchmark} and \ref{prop:random_benchmark_can_be_four_thirds_larger_than_no_quality_variability_benchmark} in Section~\ref{subsec:stronger_benchmark_under_dynamic_pricing}). We also establish that the revenue of Newest is at least half of the revenue under our benchmark (Theorem~\ref{theorem:dynamic_pricing_CoNF_bound}). This positive result is not implied by any of their results.

\paragraph{Other important differences.} There are several other differences between our work and \cite{decroix2021service} First, we extend our main result to a dynamically changing product quality (theoretically for a two-state Markov chain in Section~\ref{sec:dynamic_quality}, and empirically for increasing product quality in Section~\ref{subsec:numerics_nonstationarity}). Our Random benchmark can seamlessly apply to these settings while the fixed-quality benchmark of \cite{decroix2021service} is not suitable for non-stationary settings. 
Second, the customer's valuation in our model differentiates between the product's observable and unobservable characteristics. This allows for the potential existence of self-selection bias in our model (as customers with higher observable idiosyncratic value purchase and leave reviews). In Section~\ref{subsec:conf_with_self_selection_bias} we empirically demonstrate that our results continue to hold then. In contrast,  self-selection bias cannot arise in \cite{decroix2021service} as the customer's experiences are drawn from a single distribution and are thus unbiased i.i.d. samples.

\subsection{Brittleness of no-variability benchmark in our model.}\label{subsec:negative_bias_with_respect_to_no_variability_benchmark_does_not_always_hold_in_our_model}

Let $\textsc{Rev}^{K}(p) = p \prob[\Theta + \mu \geq p]$ be the revenue when there is no-variability in the reviews (reviews reveals the true product quality). $\textsc{Rev}^{K}(p)$ is the analgoue of the no-quality-variability benchmark in \cite{decroix2021service}. The following proposition shows that even under static pricing the revenue of $\newest$ can be arbitrarily greater than the revenue under the no-variability benchmark. This result holds even when $h$ is the mean and the customer's prior $\mathrm{Beta}(a,b)$ is correct. 
\begin{proposition}\label{prop:newest_arbitrarily_better_than_knowledgeable_static_price}
    For any $M > 0$, there exists an instance where the customer's prior is correct ($\frac{a}{a+b} = \mu$) and $h$ is the mean, and a price $p > 0$ such that $\frac{\textsc{Rev}(\newest, p)}{\textsc{Rev}^{K}(p)} > M$.
\end{proposition}

\begin{proof}[Proof of Proposition~\ref{prop:newest_arbitrarily_better_than_knowledgeable_static_price}.]
Let $c = 1$, $\mathcal{F} = \bern(q)$, $h$ be the mean, $a>0$ and $b>0$ such that $\frac{a}{a+b} = \mu$, and $p = h(1)$. The revenue under no quality variability is 
\begin{equation}\label{eq:revenue_no_quality_variability_expression}
    \textsc{Rev}^{K}(p) = p \prob[\Theta + \mu \geq p] =  h(1) \prob[\Theta + \mu \geq h(1)] = h(1) q
\end{equation}
where the last equality uses that $\prob[\Theta + \mu \geq h(1)] = \prob[\Theta = 1]$ since $\mu < \frac{a+1}{a+b+1} = h(1)$ (as $\frac{a}{a+b}= \mu$). 
Using Proposition~\ref{theorem:most_recent_C_revenue}
, the revenue of $\newest$ is
\begin{equation}\label{eq:closed_form_for_revenue_of_newest_static_price} 
    \rev(\newest,p) = \frac{p}{\expect_{N \sim \bern(\mu)}\Big[ \frac{1}{\prob[\Theta + h(N) \geq p]}\Big]}.
\end{equation}
The denominator equals
\begin{align}\label{eq:expected_duration_of_stay_in_a_state}
    &\expect_{N \sim \binomial(1,\mu)}\Bigg[ \frac{1}{\prob[\Theta + h(N) \geq p]}\Bigg] = \expect_{N \sim \binomial(1,\mu)}\Bigg[ \frac{1}{\prob[\Theta + h(N) \geq h(1)]}\Bigg] \nonumber\\
    &= \frac{1}{\prob[\Theta + h(0) \geq h(1)]} \prob[N = 0]  + \frac{1}{\prob[\Theta + h(1) \geq h(1)]} \prob[N = 1] =  \frac{1}{q} (1-\mu) + \mu.
\end{align}
The last equality uses $\prob[\Theta + h(1) \geq h(1)] = \prob[\Theta \geq 0] = 1$, $\prob[N = 1] = \mu$, $\prob[N = 0] = 1-\mu$, and
\begin{align*}
    \prob[\Theta + h(0) \geq h(1)] &= \prob\Bigg[\Theta + \frac{a}{a+b+1}   \geq \frac{a+1}{a+b+1}\Bigg] = \prob\Bigg[\Theta  \geq \frac{1}{a+b+1}\Bigg] = \prob[\Theta = 1] = q.
\end{align*}
Substituting \eqref{eq:expected_duration_of_stay_in_a_state} into \eqref{eq:closed_form_for_revenue_of_newest_static_price} and using that $p = h(1)$ yields 
$$\rev(\newest,p) = \frac{p}{\expect_{N \sim \bern(\mu)}\Big[ \frac{1}{\prob[\Theta + h(N) \geq p]}\Big]}= \frac{h(1)}{\frac{1-\mu}{q} + \mu}.$$
Hence, dividing with \eqref{eq:revenue_no_quality_variability_expression}, and taking $\mu = 1- \frac{1}{2M}$ and $q = \frac{1}{2M}$ completes the proof as:
$$\frac{\rev(\newest,p)}{\textsc{Rev}^{K}(p)}= \frac{ \frac{h(1)}{\frac{1-\mu}{q} + \mu}}{h(1)q} = \frac{1}{1-\mu + \mu q} = \frac{1}{\frac{1}{2M} +(1-\frac{1}{2M})\frac{1}{2M}} > \frac{1}{\frac{1}{2M} +\frac{1}{2M}} = M.
$$
\end{proof}

\subsection{Stronger benchmark under dynamic pricing}\label{subsec:stronger_benchmark_under_dynamic_pricing}
We now show that, under the optimal dynamic pricing, the revenue of our Random benchmark  is always at least $\max_{p \geq 0} \textsc{Rev}^{K}(p)$. This is the the revenue under the benchmark in \cite{decroix2021service} and corresponds to the revenue under their optimal dynamic pricing (see Proposition 6 in \cite{decroix2021service}).

\begin{proposition}\label{prop:random_benchmark_is_always_at_least_no_quality_variability_benchmark}
   For any instance where $h$ is the mean and the customer's prior is correct ($\frac{a}{a+b} = \mu$), $\max_{p \geq 0} \textsc{Rev}^{K}(p) \leq \rev(\random, \Pi^{\textsc{dynamic}})$.
\end{proposition}
We next show that the revenue of our Random benchmark can be larger than the revenue of benchmark in \cite{decroix2021service} by a factor of $\nicefrac{4}{3}$. This is similar to Proposition~\ref{lemma:CoNF_dynamic_pricing}.
\begin{proposition}\label{prop:random_benchmark_can_be_four_thirds_larger_than_no_quality_variability_benchmark}
    For any $\alpha < \frac{4}{3}$, there exists an instance where $h$ is the mean and the customer's prior is correct ($\frac{a}{a+b} = \mu$), where $\frac{\rev(\random, \Pi^{\textsc{dynamic}})}{\max_{p \geq 0} \textsc{Rev}^{K}(p)} > \alpha$.
\end{proposition}
To prove both propositions, the following lemma establishes that the optimal-dynamic pricing revenue under Newest equals to the optimal revenue under the benchmark of constant-review rating.
\begin{lemma}\label{lemma:optimal_pricing_revenue_newest_equals_optimal_revenue_under_no_variability}
    For any instance where $h$ is the mean and the customer's prior is correct ($\frac{a}{a+b} = \mu$),
    $$\textsc{Rev}(\newest, \Pi^{\textsc{dynamic}}) = \max_{p \geq 0} \textsc{Rev}^{K}(p).$$
\end{lemma}

\begin{proof}[Proof of Proposition~\ref{prop:random_benchmark_is_always_at_least_no_quality_variability_benchmark}.]
The proof combines Lemma~\ref{lemma:optimal_pricing_revenue_newest_equals_optimal_revenue_under_no_variability} and Proposition~\ref{lemma:CoNF_dynamic_pricing}.
\end{proof}

\begin{proof}[Proof of Proposition~\ref{prop:random_benchmark_can_be_four_thirds_larger_than_no_quality_variability_benchmark}.]
The proof combines Lemma~\ref{lemma:optimal_pricing_revenue_newest_equals_optimal_revenue_under_no_variability} and Proposition~\ref{prop:conf_still_exits_opt_dynamic_pricing}.
\end{proof}

\begin{proof}[Proof of Lemma~\ref{lemma:optimal_pricing_revenue_newest_equals_optimal_revenue_under_no_variability}.]
The expected customer's belief equals:
$$\expect_{N \sim \binomial(c, \mu)}[h(N)] = \expect_{N \sim \binomial(c, \mu)}\Big[\frac{a+N}{a+b+c} \Big] =  \frac{a+c \mu}{a+b+c} = \mu$$
where the first equality uses that $h(N) = \frac{a+N}{a+b+c}$ (since $h$ is the mean), and the last equality uses that $\frac{a+c \mu}{a+b+c} = \mu$ (since $\frac{a}{a+b} = \mu$). Combining this with Corollary~\ref{thm:most_recent_dynamic_opt_rev}, yields 
\begin{align*}
    \textsc{Rev}(\newest, \Pi^{\textsc{dynamic}}) &= \max_{p \geq 0} p \prob[\Theta + \expect_{N \sim \binomial(c, \mu)}[h(N)] \geq p] \\
    &=\max_{p \geq 0} p \prob[\Theta + \mu \geq p]  = \max_{p \geq 0} \textsc{Rev}^{K}(p).
\end{align*}
\end{proof}

\end{document}